\newif\ifsubmission
\newif\ifanon
\newif\ifextendedabstract
\newif\ifstocproceedings

\submissiontrue

\documentclass[11pt]{article}
\usepackage{headers}
\title{ 
A General Quantum Duality for Representations of Groups
\\
{
\large
with Applications to 
Quantum Money, 
Lightning, 
and
Fire 
}
}
\date{}

\ifanon
    \author{}
\else
    \author[1]{John Bostanci}
    \affil[1]{{\small 
    Columbia University, New York, NY, USA
    }}
    \author[2]{Barak Nehoran}
    \affil[2]{{\small 
    Princeton University, Princeton, NJ, USA
    }}
    \author[3]{Mark Zhandry}
    \affil[3]{{\small 
    NTT Research, Sunnyvale, CA, USA
    }}
\fi

\begin{document}

\maketitle

\newtoggle{useitems}
\ifstocproceedings
    \togglefalse{useitems}
\else
    \toggletrue{useitems}
\fi

\newenvironment{toggleditemize}
  {\iftoggle{useitems}{\begin{itemize}}{}}
  {\iftoggle{useitems}{\end{itemize}}{}}

\newcommand{\toggleditem}[1]{%
  \iftoggle{useitems}{\item }{ (#1) }%
}

\begin{abstract}
    Aaronson, Atia, and Susskind 
    (2020)
    established that efficiently mapping between quantum states $\ket{\psi}$ and $\ket{\phi}$ is computationally equivalent to distinguishing their superpositions $\frac{1}{\sqrt{2}}(\ket{\psi} + \ket{\phi})$ and $\frac{1}{\sqrt{2}}(\ket{\psi} - \ket{\phi})$. 
    We generalize this insight into a broader duality principle in quantum computation, wherein manipulating quantum states in one basis is equivalent to extracting their value in a complementary basis. In its most general form, this duality principle states that for a given group, the ability to implement a unitary representation of the group is computationally equivalent to the ability to perform a \fse from the invariant subspaces corresponding to its irreducible representations.

    Building on our duality principle, we present the following applications:
    \begin{toggleditemize}
        \toggleditem{1}
        Quantum money, which captures quantum states that are verifiable but unclonable, and its stronger variant, quantum lightning, have long resisted constructions based on concrete cryptographic assumptions. 
        While (public-key) quantum money has been constructed from indistinguishability obfuscation (iO)—an assumption widely considered too strong—quantum lightning has not been constructed from any such assumptions, with previous attempts based on assumptions that were later broken.
        We present the first construction of quantum lightning with a rigorous security proof, grounded in a plausible and well-founded cryptographic assumption.
        We extend the construction of Zhandry 
        (2024)
        from Abelian group actions to non-Abelian group actions, and eliminate Zhandry's reliance on a black-box model for justifying security. Instead, we prove a direct reduction to a computational assumption -- the pre-action security of cryptographic group actions. We show how these group actions can be realized with various instantiations, including with the group actions of the symmetric group implicit in the McEliece cryptosystem.
        \toggleditem{2}
        We provide an alternative quantum money and lightning construction from one-way homomorphisms, showing that security holds under specific conditions on the homomorphism. Notably, our scheme exhibits the remarkable property that four distinct security notions—quantum lightning security, security against both worst-case cloning and average-case cloning, and security against preparing a specific canonical state—are all equivalent.
        \toggleditem{3}
        Quantum fire captures the notion of a samplable distribution on quantum states that are efficiently clonable, but not efficiently telegraphable, meaning they cannot be efficiently encoded as classical information. These states can be spread like fire, provided they are kept alive quantumly and do not decohere.
        The only previously known construction relied on a unitary quantum oracle, whereas we present the first candidate construction of quantum fire using a classical oracle.
    \end{toggleditemize}

\end{abstract}

\newpage
\tableofcontents
\newpage

\ifextendedabstract
\else
    \section{Introduction}
\fi

\ifextendedabstract
    \noindent
\else
\fi
Let $\ket{\psi_0},\ket{\psi_1}$ be two orthogonal quantum states, and let $\ket{\phi_+}$ be proportional to $\ket{\psi_0} + \ket{\psi_1}$ and $\ket{\phi_-}$ be proportional to $\ket{\psi_0} - \ket{\psi_1}$. The \emph{Swap Complexity} of $\ket{\psi_0}$ and $\ket{\psi_1}$ is the size of the smallest circuit that maps $\ket{\psi_0}$ to $\ket{\psi_1}$ and vice versa. Meanwhile, the \emph{Distinguishing Complexity} of $\ket{\phi_+}$ and $\ket{\phi_-}$ is the size of the smallest circuit that accepts $\ket{\phi_+}$ and rejects $\ket{\phi_-}$. 
A fundamental result of Aaronson, Atia, and Susskind~\cite{aaronson2020hardness} establishes that the swap complexity of $\ket{\psi_0}$ and $\ket{\psi_1}$ is essentially equivalent to the distinguishing complexity of $\ket{\phi_+}$ and $\ket{\phi_-}$.  
This duality principle, known as the ``AAS duality'', has emerged as a simple yet powerful tool in quantum complexity theory and cryptography.

In this work, we ask: \emph{Can the AAS equivalence be extended to the more general context of many quantum states and multidimensional subspaces?}  We give an affirmative answer to this question.  First, we extend the notion of the swap complexity to a notion of ``representation complexity'': given a subspace, $V$, spanned by states $\ket{\psi_1}, \dots, \ket{\psi_k}$, and a (potentially non-Abelian) group $\G$, a \emph{representation} of $\G$ on the subspace $V$ 
is a homomorphism from $\G$ to the unitaries acting on the subspace 
(or, roughly, it is a collection of unitaries $\{\rep(g)\}_{g \in \G}$ acting on $V$ which respects the structure of $\G$).%
\footnote{
    For instance, in~\cite{aaronson2020hardness}, the representation of $\G = \Z_2$ on the subspace spanned by $\ket{\psi_0}$ and $\ket{\psi_1}$ maps the 
    sole non-identity element of $\Z_2$ to the unitary swapping $\ket{\psi_0}$ for $\ket{\psi_1}$ and vice versa.
}
We can call the \emph{Representation Complexity} of $\rep$ the size of the smallest circuit that implements the representation, that is, by mapping 
\begin{align*}
    \ket{g} \otimes \ket{\psi_i} \mapsto \ket{g} \otimes \rep(g)\ket{\psi_i}\,.
\end{align*}


When restricting to groups that have an efficient quantum Fourier transform (including all Abelian groups, all constant-sized or poly\-nomial-sized groups, and several important expo\-nential-sized non-Abelian groups),
we show that the representation complexity is essentially equivalent to the complexity of implementing a \emph{\fse}, or in other words, performing a partial measurement of 
the invariant subspaces preserved by the representation
(i.e., its irreducible representation subspaces)
and extracting the quantum state encoded in each such subspace (see 
\ifextendedabstract
    the \hyperref[sec:subspace-ext]{Duality section} below
\else
    \Cref{sec:subspace-ext,sec:fourier_extraction_tech_over} 
\fi
for more discussion on \se).
For Abelian groups, this simplifies to a full projective measurement,
and in particular, for the swapping representation of AAS, this is a measurement between $\ket{\phi_+}$ and $\ket{\phi_-}$.
Thus the AAS duality is recovered by setting $\G=\Z_2$. We additionally prove an \emph{approximate} notion of this duality, where the circuit only has to approximately map between states.%
\footnote{
    While our approximate duality theorem works for all groups, it achieves weaker error bounds than the duality of \cite{aaronson2020hardness} for $\mathbb{Z}_2$.  
}

\subsection{Applications to Cryptography}

In cryptography, the AAS duality has proven quite fruitful. Cryptographic security properties come in two types: ``search'' type properties which stipulate the hardness of computing a specific unknown quantity, and ``decision'' type properties which stipulate the hardness of distinguishing between two distributions. 
The AAS duality has played a crucial role in establishing the equivalence between certain search-type and decision-type properties, leading to a number of significant results~\cite{yan2022general,hhan2023hardness,C:KMNY24,C:MalWal24,TCC:HKNY24,AC:MorYamYam24}.

We show that our new duality theorem is useful for cryptography beyond the AAS setting, by giving novel results for quantum money.

\paragraph{Quantum Money from Group Actions.} Quantum money~\cite{Wiesner83} uses the no-cloning principle to generate unforgeable banknotes. These banknotes are quantum states that can be verified but cannot be cloned. A central problem has been to construct quantum money that can be \emph{publicly verified} by anyone, and yet only the mint can create new banknotes. This is called public-key quantum money~\cite{CCC:Aaronson09}.%
\footnote{
    When it is otherwise clear from context, we will refer to public key quantum money as simply ``quantum money''.
}
Quantum \emph{lightning} posits a stronger security notion for public-key quantum money, with a collision-resistance property that ensures that even the mint can only ever create one copy of each banknote~\cite{zhandry2021quantum}.

A long-standing challenge for public-key quantum money is to derive security from concrete computational assumptions (and in particular, assumptions that do not bake the unclonability of the banknotes directly into the assumption). The only prior scheme with such a proof is an instantiation of~\cite{STOC:AarChr12} using indistinguishability obfuscation (iO), as suggested by~\cite{BenDavid2023quantumtokens} and proved in~\cite{zhandry2021quantum}. However, iO is a powerful cryptographic tool whose quantum security is still uncertain. 
Moreover, no existing unbroken scheme has been shown to have such a security proof for the stronger security notion of quantum lightning. 

Recently, \cite{ITCS:Zhandry24a} gave a plausible construction of quantum money and quantum lightning from Abelian group actions. A group action consists of a group $G$, a set $X$, and a binary operation $*:G\times X\rightarrow X$, denoted $g*x=y$. This operation respects the group structure: $g*(h*x)=(gh)*x$. An Abelian group action is a group action where $G$ is Abelian.%
\footnote{
    Abelian groups are those for which all the elements commute: $gh = hg \;\; \forall g, h \in \G$.
}
Unfortunately, the security of the scheme of~\cite{ITCS:Zhandry24a} requires both a computational assumption \emph{and} an idealized modeling of group actions as a black box.

Using our duality principle, we show how to generalize this construction to work with \emph{non-Abelian} group actions. 
This shift is not merely a superficial adjustment\textemdash it significantly improves on the framework in two critical ways:

\begin{enumerate}
    \item 
    It allows us to prove the hardness of our quantum money and lightning scheme in the \emph{plain model}, using only a concrete assumption on the group action. This assumption also identifies an interesting potential source of hardness for non-Abelian group actions. Very roughly, for non-Abelian groups, in addition to a group action $g \,*\, (h*x) = (gh)*x$, we can also define a ``pre-action'' $g \circ (h*x) =  (hg^{-1})*x$, or more generally a ``bi-action'' $(g_0,g_1) \circledast (h*x) = (g_0 h g_1^{-1})*x$. Our assumption states that it is hard via a quantum query to distinguish a random action from a random bi-action. Importantly, this problem only makes sense for non-Abelian group actions, as actions and pre-actions are identical in the Abelian case. Thus, the quantum money result requires us to use the full power of our non-Abelian generalization of the duality. 
    \item 
    The shift to non-Abelian groups opens up the possibility for potentially more varied instantiations of the group actions. In particular, we explain how to instantiate our quantum money scheme on (a significant generalization of) the symmetric group action implicit in the McEliece cryptosystem~\cite{Mceliece78}.
\end{enumerate}

\begin{theorem}[informal]
    There is a public-key quantum money and quantum lightning scheme 
    for any (non-Abelian) cryptographic group action, 
    such that the money/lightning scheme
    is secure if the group action is preaction-secure.
\end{theorem}

To the best of our knowledge, this represents the first (unbroken) quantum lightning scheme with a plain-model security proof based on a computational assumption that does inherently include unclonability.

\paragraph{Quantum Money from One-way Homomorphisms.}
A one-way (group) homomorphism is a function, $f(h)$, that is group homomorphic%
\footnote{
    That is, it is a homomorphism between two groups $G$ and $H$, such that $f(gh) = f(g) \cdot f(h)$ for all $g,h \in G$.
}
and efficiently computable, but computationally intractable to invert.%
\footnote{
    Note that Shor's algorithm~\cite{shor94algorithms} allows efficiently inverting group homomorphosms when the domain and codomain groups are Abelian. Thus, these results inherently require non-Abelian groups, and hence our generalized~duality.
}
A one-way homomorphism can be seen as an instance of a group action, with the domain group acting on the codomain as $g * f(h) = f(gh)$. However, unlike in the previous case above, the preactions for this action (i.e., $g \circ f(h) = f(h g^{-1})$) are as efficiently computable as the action itself, so security cannot be shown as before. Nevertheless, we give sufficient conditions on the one-way homomorphism such that the resulting quantum lightning scheme is secure. 

We note that unlike our construction above from group actions that are pre-action secure\textemdash for which we give concrete instantiations that can be implemented in practice\textemdash we do not know if any instantiations of homomorphic functions satisfy these security conditions. 
But we observe that a one-way group homomorphism is essentially a group action where the computational Diffie-Hellman (CDH) problem is \emph{easy} but yet discrete logarithms are still \emph{hard}. While CDH is quantumly equivalent to discrete logarithms for Abelian groups~\cite{montgomery2022full}, this equivalence does not seem to follow for non-Abelian groups. Strangely, it is a hypothetical security \emph{failure} for group actions which gives rise to plausible instantiations for quantum lightning and quantum fire (see more on the construction of quantum fire below).

We concede the disadvantage of this construction as compared to the concrete one above from preaction security, but we note that it has some unique properties that the other does not. Specifically, by leveraging our duality principle we are able to prove the remarkable fact that four distinct quantum money security notions\textemdash namely, the collision-resistance of quantum lightning security, the hardness of both worst-case cloning and average-case cloning, and the hardness of preparing the uniform superposition over the image of the homomorphism\textemdash are all identical. Thus for any particular instantiation of the one-way homomorphism, it is sufficient to prove any one of these security notions in order to get the other three.

\paragraph{Quantum Fire.} Quantum fire refers to a collection of efficiently samplable quantum states that can be efficiently cloned,
but cannot be efficiently telegraphed.%
\footnote{
    Note that while the no-cloning theorem prohibits cloning \emph{general} quantum states, this prohibition does not apply to quantum states chosen from an orthogonal set. The same applies to the no-telegraphing theorem, which prohibits sending \emph{general} quantum states through a classical channel without pre-shared entanglement. States from an orthogonal set can clearly be telegraphed by measuring them in this basis and later recreating them accordingly. Such states can be cloned in a similar fashion. In other words, any states chosen from an orthogonal set can be both cloned and telegraphed \emph{information-theoretically}, but these tasks are not necessarily both efficient. In fact, it was shown in~\cite{ITCS:NZ23} that there are likely to be state families where cloning is efficient and yet telegraphing is not. Quantum fire is the cryptographic primitive that samples such states efficiently.
}
That is, despite the ability to make an unbounded number of copies of a quantum fire state, there is no way to efficiently encode it as classical information from which it can later be recovered.
Much like a flame can be easily spread from a single source as long as it is kept alive, quantum fire can be cloned from a single quantum state as long as it is kept coherently in quantum storage.

The concept of quantum fire was first introduced in the work of Nehoran and Zhandry~\cite{ITCS:NZ23}, where it was shown to be essential for solving the key exfiltration problem. However, it was not formally defined or named in that work. \cite{ITCS:NZ23} provided a secure construction of quantum fire relative to a unitary quantum oracle, but this oracle construction relied on an inherently inefficient computation and baked clonability into the oracle itself. Consequently, it does not provide a pathway for instantiation in the standard model. It has not even been clear if any classical oracle could allow efficient cloning of quantum states that are inherently quantum (and thus not telegraphable).

Inspired by the duality principle, we give a plausible candidate construction of quantum fire relative to a one-way group homomorphism.
Remarkably, despite the similarity to the construction of quantum lightning from group homomorphisms, where the states are \emph{unclonable}, the states in this scheme are inherently \emph{clonable}, and efficiently so. Nevertheless, there is no apparent way to telegraph the states efficiently.
Moreover, it is straightforward to define a classical oracle that gives a candidate one-way group homomorphism. 
Thus, we obtain a candidate construction of quantum fire with conjectured security relative to a classical oracle, improving upon the unitary oracle construction of~\cite{ITCS:NZ23}.%
\footnote{
    Note that, as observed in~\cite{ITCS:NZ23}, an unconditional security proof relative to such a classical oracle would likely require proving a classical oracle separation between the complexity classes $\QMA$ and $\QCMA$, a major open problem of Aharonov and Naveh~\cite{aharonov2002quantum}, which, despite recent progress, has evaded resolution.
}

\subsection{A Generalized Duality}

\paragraph{\fse.}
\label{sec:subspace-ext}
A major stepping stone towards our duality theorem is the idea of a \fse.  Every group representation preserves some set of invariant subspaces $\{\reg{W}_{\lambda}\}_{\lambda \in [n]}$.%
\footnote{
    That is, these subspaces are invariant under all of the unitaries $U_g$ corresponding to each group element $g\in \G$. In some cases, the only invariant subspace may be the full Hilbert space, in which case we say that it is \emph{irreducible}, but this is not generically the case. We consider here only invariant subspaces which are irreducible, and do not break down further into smaller invariant subspaces.
}
A \emph{course} Fourier measurement%
\footnote{
    Often called weak Fourier sampling in many contexts
}
of the representation
is, roughly, a projection onto these subspaces. 
We get a classical label $\lambda$ indicating the subspace we have projected onto, as well as a collapsed state, $\ket{\psi}$, within the subspace $\reg{W}_{\lambda}$.
A \emph{fine} Fourier measurement%
\footnote{
    Commonly called strong Fourier sampling
}
further measures within each of those subspaces, in a basis that depends on the algorithm.
For instance, if $\{\ket{\psi^{\lambda}_{j}}\}_{j \in \dim(\reg{W}_{\lambda})}$ is a basis for $\reg{W}_{\lambda}$, we get outcomes $\lambda$ and $j$, and collapse our state to $\ket{\psi^{\lambda}_{j}}$.%
\footnote{
    To simplify the notation, we assume here that there is no multiplicity, or degeneracy, in the irreducible representations. We will see later how to handle multiplicity.
}
In either case, the state after the measurement is still within the subspace.

In some applications, we care about the \textit{coherent} information encoded within each subspace. That is, it is not enough to know which collapsed state $\ket{\psi^{\lambda}_{j}}$ we received. 
We want to have, in our hands, the coherent superposition that appeared in the subspace.
That is, if the original state was $\sum_{j\in [\dim(\reg{W}_{\lambda})]} \alpha_{j} \ket{\psi^{\lambda}_{j}}$, we want to \emph{extract} the full superposition $\sum_{j} \alpha_{j} \ket{j}$.  
This transformation,
which we call a \emph{\se},
extracts the full state coherently from the subspace.%
\footnote{Note that such extraction is not generally an efficient transformation for arbitrary subspaces.}

If implemented na{\"i}vely, Fourier measurements do not suffice for this task. They either do not recover the information about where the state was \emph{within} each subspace (in the case of course Fourier measurement), or they recover it in a collapsed form (in the fine case).
In our work, we consider the stronger notion of a ``\emph{\fse}'', an operation that measures the subspace and \emph{coherently recovers} the encoded state.

\needspace{4\baselineskip}
\paragraph{Duality.}
We show that the implementations of representations and the implementations of their \fses are essentially computationally dual to each other.

\begin{samepage}
\begin{theorem}[Duality, \emph{informal}]
    \label{thm:duality_informal}
    Let $\G$ be a group with an efficient quantum Fourier transform%
    \footnote{
        Note that while not every group is known to have an efficient quantum Fourier transform, this does include a very wide class of groups, and includes, at the very least, all Abelian groups, as well as many important non-Abelian groups. Moreover, \emph{every} fixed-size group is technically efficient (whether Abelian or not), so this condition is only important for some families of groups whose sizes grow exponentially.
    }
    and let $\rep: \G \to U(\mathcal{H})$ be a representation of $\G$. Then the following are equivalent:
    \begin{itemize}
        \item $\rep$ has an efficient implementation, i.e. $\ket{g} \otimes \ket{\psi} \mapsto \ket{g} \otimes \rep(g) \ket{\psi}$.
        \item $\rep$ has an efficient \fse, i.e. $\ket{\psi^{\lambda}_{i, j}} \mapsto \ket{\phi_i^{\lambda}}\ket{\lambda, j}$.%
    \end{itemize}
\end{theorem}
\end{samepage}

\paragraph{Further Discussion of \FSE.}
In the above discussion, we have glossed over the possibility of \emph{multiplicity} or \emph{degeneracy}, in which the representation acts identically on several different invariant subspaces 
$\reg{W}^{\lambda}_{1}, \reg{W}^{\lambda}_{2}, \dots, \reg{W}^{\lambda}_{m}$. 
Such subspaces are degenerate in the sense that a course Fourier measurement produces the same outcome, $\lambda$, on all of them. Thus we have an additional index, $i$, that runs over this multiplicity of $\lambda$. 

We write a \fse as an isometry $\ket{\psi^{\lambda}_{i, j}} \mapsto \ket{\phi_i^{\lambda}}\ket{\lambda, j}$, where for each $\lambda$ and $i$, the states $\{\ket{\psi^{\lambda}_{i, j}}\}_{j}$ are a basis for the subspace $\reg{W}^{\lambda}_{i}$, and the state $\ket{\phi^{\lambda}_{i}}$ is an arbitrary ``junk'' state that is left behind after measuring $\lambda$ and extracting $j$. 

In order for it to be an \emph{extraction} of $j$, rather than a measurement of $j$, it is crucial that this leftover state has no dependence on $j$.
Consider a superposition $\sum_{j\in [\dim(\reg{W}^{\lambda}_{i})]} \alpha_{j} \ket{\psi^{\lambda}_{i, j}}$ over the subspace $\reg{W}^{\lambda}_{i}$. Performing this isometry yields $\sum_{j} \alpha_{j} \ket{\phi_i^{\lambda}}\ket{\lambda, j} = \ket{\phi_i^{\lambda}}\ket{\lambda} \otimes \sum_{j} \alpha_{j} \ket{j}$, which extracts the original superposition into a quantum state on the last register with those exact amplitudes. If the leftover junk state had depended on $j$, for instance if we instead had $\ket{\psi^{\lambda}_{i, j}} \mapsto \ket{\phi_{i,j}^{\lambda}}\ket{\lambda, j}$, then this would not extract the state properly. We would instead get $\sum_{j} \alpha_{j} \ket{\phi_{i,j}^{\lambda}}\ket{\lambda, j}$, where the last register is still entangled with the rest of the state, and thus has not been extracted. This is the difference between a \emph{measurement} of $j$ and an \emph{extraction} of $j$. (See \Cref{sec:fourier_extraction_tech_over} for more discussion on \fse.)

We observe that since these leftover junk states $\ket{\phi_i^{\lambda}}$ are independent of $j$---that is, they do not depend on which state we started from within the subspace $\reg{W}^{\lambda}_{i}$---we can see that these states are instead characteristic of the subspace $\reg{W}^{\lambda}_{i}$ itself. That is, the \fse collapses each subspace $\reg{W}^{\lambda}_{i}$ to a single distinct quantum state $\ket{\phi_i^{\lambda}}$, which we therefore call the ``\arch'' states of these subspaces.
Despite appearing to be just the ``junk'' that is left behind during the \fse, these \arch states are in fact quite useful.

For instance, the existence of these \arch states allows us to use a swap test to distinguish whether two quantum states are in the same subspace or different subspaces. 
Consider two states $\ket{\psi^{\lambda}_{i_1, j_1}} \in \reg{W}^{\lambda}_{i_1}$ and $\ket{\psi^{\lambda}_{i_2, j_2}} \in \reg{W}^{\lambda}_{i_2}$ that live in subspaces corresponding to the same $\lambda$, but potentially different such subspaces (that is, $\reg{W}^{\lambda}_{i_1}$ and $\reg{W}^{\lambda}_{i_2}$ are potentially different), and suppose that we wanted to test whether they in fact belong to the same subspace (that is, if $i_1 = i_2$). The ability to perform the representation \emph{does not} in general allow us to measure $i$. Intuitively, this is because both these states behave identically under the representation.
A Fourier measurement/sampling of these states would give us only $\lambda$, or both $\lambda$ and $j$, but not $i$. So how can we test if they are in the same subspace? This is in general not possible from such a measurement. However, \fse is more powerful than Fourier measurement and gives us this ability. Performing a \fse on both states gives us $\ket{\phi_{i_1}^{\lambda}} \ket{\lambda} \ket{j_1}$ for the first state and $\ket{\phi_{i_2}^{\lambda}} \ket{\lambda} \ket{j_2}$ for the second state. Now we can ignore and discard the last register---the one that indicates which state we had within each subspace---and perform a swap test only on the first register, that is between the \arch states that characterize the subspaces.
This turns out to be a crucial tool in the security proof of our quantum lightning construction (see \Cref{sec:preaction_sec_tech_over} for a discussion on this).

\paragraph{The Special Case of Abelian Groups.}
Abelian groups have the special property that all of the (irreducible) invariant subspaces are one-dimensional. Since the ``quantum state'' extracted by the \fse in this case is one-dimensional, it is actually just a complex phase. We can see that the corresponding isometry simplifies to $\ket{\psi^{\lambda}_{i}} \mapsto \ket{\phi_i^{\lambda}}\ket{\lambda}$, where we have absorbed the phase into $\ket{\phi_i^{\lambda}}$.
This is computationally equivalent to the isometry $\ket{\psi^{\lambda}_{i}} \mapsto \ket{\psi_i^{\lambda}}\ket{\lambda}$ (by copying $\lambda$ and uncomputing), which we can see is 
just the course Fourier measurement for the representation\textemdash that is, a projective measurement onto the subspaces $\reg{W}^{\lambda}$.
We therefore get the following simplified duality for Abelian groups as a special case: a duality between the efficiency of implementing the representation and that of performing a Fourier measurement,%
\footnote{
    Note that because representations of Abelian groups have only one-dimensional representations, there is no distinction between the course/weak and the fine/strong versions of Fourier measurement/sampling. Thus, we refer to it as simply Fourier measurement.
}
a projective measurement on the subspaces spanned by its invariant states.%
\footnote{
    The irreducible invariant subspaces are one-dimensional, and are thus individual quantum states.
}

\begin{corollary}[Duality for Abelian Groups, \emph{informal}]
    \label{cor:duality_abelian_informal}
    Let $\G$ be an Abelian group
    and let $\rep: \G \to U(\mathcal{H})$ be a representation of $\G$. Then the following are equivalent:
    \begin{itemize}
        \item $\rep$ has an efficient implementation, i.e. $\ket{g} \otimes \ket{\psi} \mapsto \ket{g} \otimes \rep(g) \ket{\psi}$.
        \item $\rep$ has an efficient Fourier measurement, i.e. $\ket{\psi^{\lambda}_{i}} \mapsto \ket{\psi_i^{\lambda}}\ket{\lambda}$.%
    \end{itemize}
\end{corollary}

\subsection{Related Work}

\paragraph{Quantum Money, Lightning, etc.} There have been several attempts at constructing public-key quantum money~\cite{bennett1983quantum,CCC:Aaronson09,ITCS:FGHLS12,STOC:AarChr12,zhandry2021quantum,KanShaSil22,STOC:AGKZ20,KLS22,EC:LiuMonZha23,ITCS:Zhandry24a}. Unfortunately, a number of them later turned out to be broken~\cite{LAFGHKS09,PDFHP19,EC:Robers21,EC:LiuMonZha23}. In order to gain confidence in constructions, it is therefore important to give security proofs under computational assumptions that have received significant scrutiny from the cryptographic community. Here, the best we currently have are:
\begin{itemize}
    \item Quantum money from hidden subspaces~\cite{STOC:AarChr12}, which was proved secure assuming quantum-resistant \emph{indistinguishability obfuscation} (iO) in~\cite{zhandry2021quantum}. Unfortunately, while candidates for quantum-resistant iO are known, their status is still very much open. This scheme also only achieves quantum money, but not quantum lightning.
    \item Quantum money from random walks~\cite{ITCS:FGHLS12,EC:LiuMonZha23}, which was shown to be secure under strong quantum ``knowledge'' assumptions. Such assumptions are not ``falsifiable'', and there is some doubt about the plausibility of such assumptions~\cite{ITCS:Zhandry24a}.
    \item Quantum money from Abelian group actions~\cite{ITCS:Zhandry24a}, which is proven secure under an assumption \emph{plus} in an idealized model of group actions as a black box. 
\end{itemize}
We provide a scheme whose quantum lightning security we prove in the \emph{plain model} (i.e. without making idealized model assumptions) from a plausible and falsifiable computational assumption. We hope that our work motivates further study of the cryptographic uses of non-Abelian group actions, and in particular, of the hardness of preactions.


\paragraph{Comparison to the Duality of Aaronson, Atia, and Susskind.}
\cite{aaronson2020hardness} show that there is a duality between, on the one hand, swapping between two orthogonal states, and on the other hand, measuring the positive and negative superpositions of the two states.  As a representation, this ``swapping'' operation, together with the identity, is a representation of $\Z_2$.  The invariant subspaces of this representation are the positive and negative superpositions, with eigenvalues $+1$ and $-1$.
Our duality theorem precisely yields the duality theorem from \cite{aaronson2020hardness} as a special case when applied to $\Z_2$, and recovers the same circuits, showing that our results are a proper generalization.

\Cref{thm:duality_informal} extends far beyond $\Z_2$, and even beyond Abelian groups, to many of the non-Abelian groups that are important for cryptography.
We expect our duality theorem to be applicable to many more settings in quantum cryptography and complexity theory.
Our applications to building quantum money, lightning, and fire are just a few demonstrations of the usefulness of our theorem and techniques, and demonstrate the usefulness of considering this quantum duality in its full non-Abelian generalization.

\paragraph{Fourier Measurements and their Applications}

Fourier measurements/sampling for Abelian groups play an important role in many famous quantum speedups, including the efficient quantum algorithm for Simons problem~\cite{simon1997power}, a key subroutine in the quantum speedup of Shor's factoring algorithm~\cite{shor94algorithms}.  Fourier measurements also play an important role more generally in algorithms for hidden subgroup problems for general groups~\cite{childs2007weak}.  As a result, there has been a long line of work on developing efficient algorithms for Fourier measurements.  \cite{harrow2005applications} showed how to implement Fourier sampling for arbitrary representations of the symmetric group using an identical circuit to the implementation of \fse, which they call the generalized phase estimation algorithm.  This algorithm can be extended to 
any arbitrary group that has an efficient quantum Fourier transform.  Several works have made progress towards efficient algorithms for performing strong Fourier sampling (fine Fourier measurements) in specific bases of the symmetric group~\cite{bacon2005quantum, krovi2019efficient}.  

Fourier measurements also appear often in learning theory in the context of the symmetric group.  When treating the action of permuting $n$ identical registers as a representation of the symmetric group, Fourier sampling along the joint irreducible subspaces of the symmetric and unitary groups (which arises as a result of Schur-Weyl duality) appears as an important subroutine in many tomography algorithms~\cite{keyl2006quantum, o2016efficient, o2017efficient, haah2016sample, chen2024optimal} and compression algorithms~\cite{yang2016optimal}.

Problems of computing certain constants associated with the irreducible representations of the symmetric group (Kostka numbers, Littlewood-Richardson coefficients, and Plethysm coefficients) are thought to be problems that might be outside of $\mathsf{NP}$, but in $\mathsf{BQP}$.  Fourier measurements play an important role in designing efficient quantum algorithms for computing these numbers~\cite{larocca2024quantum}.  The Kronecker coefficients are another set of constants associated with the symmetric group whose complexity remains unknown.  Recent work~\cite{bravyi2024quantum} uses generalized phase estimation and weak Fourier measurement, to design a $\mathsf{QMA}$ protocol for computing them.  
\section{Technical Overview}

This work uses the tools of representation theory, and moreover, understanding the principles of representation theory is important for understanding both the statement of our duality theorem as well as the technical details of our quantum lightning construction and security proof.
However, since we expect that the duality theorem and our other results will be applicable to researchers from a wide variety backgrounds, we have done our best to write this technical overview with an expository flavor, by giving a high-level intuition behind the representation-theoretic tools we use and how they lead to our contributions.
For a more in-depth introduction to representation theory in the context of quantum computing, we refer to lecture notes by Andrew Childs~\cite{childs2017lecture}, accompanied by any of a number of great representation theory textbooks, e.g.~\cite{serre1977linear,steinberg2009representation}.

\subsection{Quantum Fourier Transforms}

In order to give intuition about the duality theorem and quantum lightning construction, we will first need to make a few remarks about the generalized notion of the quantum Fourier transform, specifically its generalization to representations of non-Abelian groups.  Every representation $\rep$, has a basis in which it is block-diagonal, with the blocks corresponding to its irreducible sub\-represen\-tations---these are the smallest building blocks that cannot be further decomposed.%
\footnote{
    Note that in the special case of an Abelian representation, all the unitaries commute, and we can therefore go further and fully diagonalize everything simultaneously. In that case, the resulting blocks are simply all of size one. However, for non-Abelian representations, such a simultaneous diagonalization is not possible, and the closest we can get is this block-diagonal form. 
}
In other words, in this basis, $\rep$ is a direct sum of irreducible representations $\{\irrep_{\lambda}\}_{\lambda \in \widehat{G}}$, where $\widehat{G}$ is the set of labels corresponding to each possible irreducible representation~of~$G$:
\setcounter{MaxMatrixCols}{15}
\newcommand{\blockmatrixofrep}{
    \ensuremath{
        \scalebox{0.5}{$
            \begin{pNiceMatrix}[columns-width=.25cm, cell-space-bottom-limit=.25cm]
                & \Block[draw]{3-3}{\scalebox{1.8}{$\varrho_1(g)$}} & & & \null & \null & \null & \null & \null & \null & \null & \null & \null & \null \\
                &                                 & & & \null & \null & \null & \null & \null & \null & \null & \null & \null & \null \\
                &                                 & & & \null & \null & \null & \null & \null & \null & \null & \null & \null & \null \\
                & \null & \null & \null & \Block[draw]{3-3}{\scalebox{1.8}{$\varrho_1(g)$}} & & & \null & \null & \null & \null & \null & \null & \null \\
                & \null & \null & \null &                                 & & & \null & \null & \null & \null & \null & \null & \null \\
                & \null & \null & \null &                                 & & & \null & \null & \null & \null & \null & \null & \null \\
                & \null & \null & \null & \null & \null & \null & \Block[draw]{2-2}{\scalebox{1.2}{$\varrho_2(g)$}} & & \null & \null & \null & \null & \null \\
                & \null & \null & \null & \null & \null & \null &                                 & & \null & \null & \null & \null & \null \\
                & \null & \null & \null & \null & \null & \null & \null & \null & \Block[draw]{2-2}{\scalebox{1.2}{$\varrho_2(g)$}} & & \null & \null & \null \\
                & \null & \null & \null & \null & \null & \null & \null & \null &                                 & & \null & \null & \null \\
                & \null & \null & \null & \null & \null & \null & \null & \null & \null & \null & \ddots & \null & \null \\
                & \null & \null & \null & \null & \null & \null & \null & \null & \null & \null & \null & \Block[draw]{2-2}{\scalebox{1.4}{$\varrho_k(g)$}} & & \\
                & \null & \null & \null & \null & \null & \null & \null & \null & \null & \null & \null & & & \\
            \end{pNiceMatrix}
        $}
    }
}
\ifstocproceedings  
\begin{align}
    \label{eq:block-diagonalization}
    \rep(g) 
    &
    = 
    V^{\dagger} 
    \left(
        \bigoplus_{\substack{\lambda \in \widehat{G}\\ i \in [\multlambda]}} 
              \varrho_{\lambda}(g)
    \right) 
    V
    \\
    &
    =
    V^\dagger
    \;
    \blockmatrixofrep
    \;
    V
    \,,
\end{align}
\else
\begin{align}
    \label{eq:block-diagonalization}
    \rep(g) 
    = 
    V^{\dagger} 
    \left(
        \bigoplus_{\substack{\lambda \in \widehat{G}\\ i \in [\multlambda]}} 
              \varrho_{\lambda}(g)
    \right) 
    V
    =
    V^\dagger
    \;
    \blockmatrixofrep
    \;
    V
    \,,
\end{align}
\fi
where, $\multlambda$ denotes the ``multiplicity'' of the irreducible representation $\irrep_{\lambda}$ in $\rep$, that is, the number of times it appears repeated in the block-diagonalization, and each individual block has dimension $\dimlambda$.  We note while that the block decomposition into the larger subspaces%
\footnote{
    These are called the \emph{isotypic components} of $\lambda$, which group together all the blocks that behave similarly under $\rep$. 
}
$\reg{W}^{\lambda} = \mathrm{span}\{\ket{\lambda, i, j}: i \in [\multlambda], j \in [\dimlambda]\}$ is unique, for non-Abelian groups, there is flexibility in the actual choice of the basis vectors $\ket{\lambda, i, j}$ of $\reg{W}^{\lambda}$.

This unitary transformation, $V$, which optimally block-dia\-gona\-lizes $\rep$, is the \emph{Fourier transform} $\mathrm{QFT}_{\rep}$ of the representation.%
\footnote{
    Even though this unitary may or may not be efficient, we write it as a \emph{quantum} Fourier transform, $\mathrm{QFT}_{\rep}$, instead of $\mathrm{FT}_{\rep}$ for notational consistency.
}
Although it is not common to think of this transformation as the Fourier transform, it is indeed a generalization of the familiar quantum Fourier transform:  
the plain Fourier transform of a group $G$ (without any representation specified) is simply the Fourier transform of its left-regular representation, $\regrep[G]$ (the representation defined by left multiplication on the group algebra, i.e. $\regrep[G](g) \ket{h} = \ket{gh}$). In this case, the quantum Fourier transform can be written as
\begin{equation*}
    \mathrm{QFT}_{G} 
    =
    \mathrm{QFT}_{\regrep[G]} 
    = 
    \sum_{g \in G} 
    \sum_{
        \substack{
            \lambda \in \widehat{G} \\ 
            i, j \in [\dimlambda]
        }
    } 
    \sqrt{\frac{\dimlambda}{\abs{G}}}
    \,
    \irrep_{\lambda}(g)_{j, i} 
    \ket{\lambda, i, j}
    \!\!
    \bra{g}
    \,.
\end{equation*}
We of course recover the most familiar form of the quantum Fourier transform by 
specializing further to the cyclic group $\Z_N$, 
and thus 
setting $\dimlambda = 1$ (since $\Z_N$ is Abelian) and $\irrep_{\lambda}(g) = e^{2\pi i (g \cdot \lambda) / N}$.
The basis 
$
    \{
        \ket{\lambda, i, j}
    \}_{
        \lambda \in \widehat{G}, 
        \; 
        i \in [\multlambda],
        \;
        j \in [\dimlambda]
    }
$
is the Fourier basis of the representation, with classical strings $\lambda$, $i$, and $j$ labeling the irreducible representation, multiplicity index, and basis state within each irreducible invariant subspace, respectively, corresponding to the blocks in~\Cref{eq:block-diagonalization}. 

The plain quantum Fourier transform, $\mathrm{QFT}_{G}$ of many groups can be performed efficiently (including for all constant size or polynomial-sized groups, Abelian groups~\cite{coppersmith1994fourier}, and many important non-Abelian groups such as the dihedral and symmetric groups~\cite{hoyer1997efficient,beals1997quantum,moore2006generic}). 
However, even for such groups, the \emph{representation} Fourier transform $\mathrm{QFT}_{\rep}$ of many of their representations may still be inefficient. For instance, as we will see, the Fourier transform of a \emph{cryptographic group action} is inefficient even if the Fourier transform of the underlying group is efficient.
Therefore, in order to access information about the Fourier basis of a general representation, we need more clever techniques like Fourier sampling and \fse.

\subsection{Fourier Sampling and \FSE}

Despite the lack of an efficient quantum Fourier transform for general representations, you might nevertheless want to access information about the Fourier basis (moreover, the transformed Fourier basis, 
$
    \{
        \mathrm{QFT}_{\rep}^\dagger
        \ket{\lambda, i, j}
    \}_{
        \lambda \in \widehat{G}, 
        \; 
        i \in [\multlambda],
        \;
        j \in [\dimlambda]
    }
$,
which we refer to as the Fourier basis when there is no ambiguity).
To demonstrate that this is possible and build intuition, let us start with the simple case of a representation, $\rep : G \to U(\mathcal{H})$, of an Abelian group $G$.
Say you want to perform a projective measurement in the Fourier basis of $\rep$ on some state $\ket{\psi} \in \mathcal{H}$.  Without computational considerations, we aim to implement the circuit in~\Cref{fig:fourier-sampling-ideal}, where we perform a representation Fourier transform $\mathrm{QFT}_{\rep}$ on $\ket{\psi}$, to get a superposition over classical strings encoding the Fourier basis information (in this case, just the irreducible representation label, $\lambda$), copy the classical string $\lambda$ onto an ancilla register,%
\footnote{
    While we draw this copy operation as a CNOT gate for simplicity, for Abelian groups, this copy can equivalently be seen as the group operation of the group's Pontryagin dual $\widehat{G}$.
}
and then undo the representation Fourier transform to recover the projected state.

\begin{figure}[H]
    \centering
    \begin{quantikz}
        \lstick{$\ket{0^n}$} & \qw & \targ{} & \qw  & \meter{} \\
        \lstick{$\ket{\psi}$} & \gate{\mathrm{QFT}_{\rep}} & \ctrl{-1} & \gate{\mathrm{QFT}_{\rep}^\dagger} & \qw \\
    \end{quantikz}
    \caption{The ideal functionality of a measurement in the Fourier basis of a group representation,~$\rep$. Here,~$\mathrm{QFT}_{\rep}$ is not the Fourier transform of the group, but rather the (generally inefficient) Fourier transform of the \emph{representation}.}
    \Description{Quantum circuit with an ancilla register in the state 0 to the n at the top and a state ket psi at the bottom. A representation Fourier transform is applied to the bottom wire, followed by an upward CNOT gate and an inverse representation Fourier transform. Then the top wire (the ancilla) is measured.}
    \label{fig:fourier-sampling-ideal}
\end{figure}

As is probably clear, however, the circuit in \Cref{fig:fourier-sampling-ideal} is not actually an explicit algorithm, since there is in general no efficient circuit for the representation Fourier transform $\mathrm{QFT}_{\rep}$.

However, if we look at both registers in their respective Fourier bases---the top register in the Fourier basis of the group (or equivalently, of the left-regular representation), and the bottom register in the Fourier basis of the representation $\rep$---it turns out that the upwards copy (an upwards group operation of $\widehat{G}$) looks like a downwards copy: a controlled group representation onto the bottom register. Thus, we can pull out Fourier transforms on both wires and flip the circuit diagram to get the equivalent circuit depicted in \Cref{fig:fourier-sampling}. (This is, of course, similar to, and a proper generalization of, the well-known identity
$
\scalebox{0.55}{$
\begin{quantikz}[row sep=.1cm, column sep=.3cm]
    & & \targ{} & & \midstick[2,brackets=none]{=} & \gate{H} & \ctrl{1} & \gate{H} & \\
    & \gate{H} & \ctrl{-1} & \gate{H} & & & \targ{} & &
\end{quantikz}
$}
$
on the group $\Z_2$ and its regular representation.)

\begin{figure}[H]
    \centering
    \begin{quantikz}
        \lstick{$\ket{0^n}$} & \gate{\mathrm{QFT}^\dagger} & \ctrl{1} & \gate{\mathrm{QFT}}  & \meter{} \\
        \lstick{$\ket{\psi}$} & \qw & \gate{\rep(g)^{\dagger}} & \qw & \qw \\
    \end{quantikz}
    \caption{Fourier sampling circuit for the representation $\rep$. For Abelian representations, the functionality of this circuit is identical to that of \Cref{fig:fourier-sampling-ideal}, but unlike the circuit drawn there, this circuit depicts an efficient algorithm.}
    \label{fig:fourier-sampling}
\end{figure}

This is called the Fourier sampling circuit, and for Abelian groups, it acts identically to the circuit of \Cref{fig:fourier-sampling-ideal}, but now we have moved from having gates for a \emph{representation} Fourier transform on the bottom wire in \Cref{fig:fourier-sampling-ideal} to having only gates for the Fourier transform \emph{of the plain group} on the top wire in \Cref{fig:fourier-sampling}. Now, if our group has an efficient quantum Fourier transform, and we can efficiently implement the controlled representation, then the circuit in \Cref{fig:fourier-sampling} is efficiently implementable, even while the original circuit in \Cref{fig:fourier-sampling-ideal} may not be.
Thus, while we may not be able to implement the representation's Fourier transform, we can still recover and measure some information about the Fourier basis of the representation.

\begin{remark}
    A minor comment is that general representations---even representations of Abelian groups---could have multiplicity (the generalization of the concept of degeneracy in the eigenspaces of unitaries). In this case, the same irreducible representation $\irrep_{\lambda}$ could appear $\multlambda$ different times in the block-diagonal decomposition of $\rep$. Therefore, besides the label $\lambda \in \widehat{G}$ of the irreducible representation, we have a multiplicity index $i \in [\multlambda]$ which indexes which of the irreducible invariant subspaces corresponding to $\irrep_{\lambda}$ we have. Since the representation $\rep$ does not distinguish between subspaces on which it acts identically, there can be no way to generically use it to measure $i$. This is captured by the fact that the circuit we get when we do the flipping trick above on \Cref{fig:fourier-sampling} is not in fact the circuit in \Cref{fig:fourier-sampling-ideal}, but rather the one in \Cref{fig:fourier-sampling-ideal-mult}, where the information about the label $\lambda$ of the irreducible representation is copied to the ancilla register, but the information about the multiplicity index $i$ is left untouched.

    \begin{figure}[H]
        \centering
        \begin{quantikz}[wire types={q,n,b,n},classical                 gap=0.08cm,row sep={0.5cm,between origins}]
            \lstick{$\ket{0^n}$} & \qw & & \targ{} &  & \qw  & \meter{} \\[0.5cm]
            & \gate[3]{\mathrm{QFT}_{\rep}} & \qw{\lambda} \setwiretype{q} & \ctrl{-1} &  & \gate[3]{\mathrm{QFT}_{\rep}^\dagger} & \setwiretype{n} \\
            \lstick{$\ket{\psi}$} &  & \setwiretype{n} &  &  &  & \setwiretype{b} \\
            &  & \qw{i} \setwiretype{q} &  &  &  & \setwiretype{n} \\
        \end{quantikz}
        \caption{The ideal functionality of a measurement in the Fourier basis of an Abelian representation $\rep$ that has multiplicity. The label $\lambda$ of the irreducible representation is copied up to the ancilla register, while the multiplicity index $i$ is left inaccessible. While this circuit is not efficiently implementable, its functionality is equivalent to that of the efficient Fourier sampling circuit of \Cref{fig:fourier-sampling}.}
        \label{fig:fourier-sampling-ideal-mult}
    \end{figure}
\end{remark}

\paragraph{Generalizing to Non-Abelian Groups}
When generalizing to non-Abelian groups, it must be taken into account that there is now not only the possibility of multiplicity, but also the presence of irreducible subspaces of dimensions larger than one. We therefore have a third index, $j \in [\dimlambda]$, which runs over the dimension $\dimlambda$ of the irreducible subspace. We now have a choice between two kinds of Fourier measurements: a coarser-grained version known as weak Fourier sampling, and a finer-grained version known as strong Fourier sampling (\Cref{fig:fourier-sampling-ideal-non-abelian}).

\begin{figure}[htbp]
    \centering
    \begin{subfigure}{0.45\textwidth}
        \centering
        \resizebox{\linewidth}{!}{
            \begin{quantikz}[wire types={q,n,b,n},classical                     gap=0.08cm]
                \lstick{$\ket{0^n}$} &  &  & \targ{} &    &   & \meter{} \\
                & \gate[3]{\mathrm{QFT}_{\rep}} & \qw{\lambda} \setwiretype{q} &  \ctrl{-1} &    & \gate[3]{\mathrm{QFT}_{\rep}^\dagger} & \setwiretype{n} \\
                \lstick{$\ket{\psi}$} &  & \qw{i} \setwiretype{q} &  &    &  & \setwiretype{b} \\
                &  & \qw{j} \setwiretype{q} &  &    &  & \setwiretype{n} \\
            \end{quantikz}
        }
        \caption{Ideal functionality of weak Fourier sampling}
        \label{fig:weak-fourier-sampling-ideal}
    \end{subfigure}
    \hfill
    \begin{subfigure}{0.45\textwidth}
        \centering
        \resizebox{\linewidth}{!}{
            \begin{quantikz}[wire types={q,q,n,b,n},classical                     gap=0.08cm]
                \lstick{$\ket{0^n}$} &  &  & \targ{} &  &  &   & \meter{} \\
                \lstick{$\ket{0^n}$} &  &  &  & \targ{} &  &   & \meter{} \\
                & \gate[3]{\mathrm{QFT}_{\rep}} & \qw{\lambda} \setwiretype{q} &  \ctrl{-2} &  &  & \gate[3]{\mathrm{QFT}_{\rep}^\dagger} & \setwiretype{n} \\
                \lstick{$\ket{\psi}$} &  & \qw{i} \setwiretype{q} &  &  &  &  & \setwiretype{b} \\
                &  & \qw{j} \setwiretype{q} &  &  \ctrl{-3} &  &  & \setwiretype{n} \\
            \end{quantikz}
        }
        \caption{Ideal functionality of strong Fourier sampling}
        \label{fig:strong-fourier-sampling-ideal}
    \end{subfigure}
    \caption{The ideal functionality of two kinds of measurements in the Fourier basis of a non-Abelian representation $\rep$. In both of these, the label $\lambda$ of the irreducible representation is copied up to an ancilla register, and in the strong case, the state index $j$ is copied up as well. The multiplicity index $i$ is of course always left inaccessible. Neither circuit can be instantiated efficiently, but either functionality can be implemented using the efficient Fourier sampling circuit of \Cref{fig:fourier-sampling}.}
    \label{fig:fourier-sampling-ideal-non-abelian}
\end{figure}

Either type of non-Abelian Fourier sampling is known to be implementable from the representation by using the Fourier sampling circuit of \Cref{fig:fourier-sampling}.
For many purposes, this ends up being sufficient. 
Indeed, it has long been known that an efficient implementation of the representation (and the group's quantum Fourier transform) is enough to get efficient Fourier sampling.

However, if we want an actual \emph{duality}, the reverse direction must hold as well, which unfortunately it does not: in certain cases, there may very well be a way to get even a strong Fourier sampling functionality for a representation that has no efficient implementation. 

Getting a generalized duality is therefore quite a bit more subtle.
If we take a closer look at what actually occurs in the Fourier basis of both wires of \Cref{fig:fourier-sampling} in the non-Abelian case (\Cref{fig:fourier_extraction_from_rep_ideal}), we notice something interesting. As expected, the information flow is reversed---going from the bottom to the top wire instead of top to bottom. However, while the information about the irrep label $\lambda$ is \emph{copied} up to the appropriate ancilla, the information about the state index $j$ moves up, but it is in fact not copied at all---it is extracted!

\begin{figure}[H]
    \centering
\begin{quantikz}[wire types={q,q,q,n,b,n},classical
gap=0.08cm]
\lstick{$\ket{0^n}$} & \qw & \qw & \targ{} & \permute{1, 2, 6, 4, 5, 3} & \qw & \qw & \qw \\
\makeebit{$\ket{\text{EPR}_k}$}& \qw & \qw & \qw & & \qw & \qw & \qw \\
 & \qw & \qw & \qw & \qw & \qw& \qw & \qw \\
& \gate[3]{\mathrm{QFT}_{\rep}} & \qw{\lambda}  \setwiretype{q} & \ctrl{-3} & \qw & \qw & \gate[3]{\mathrm{QFT}_{\rep}^\dagger} & \setwiretype{n} & \\
\lstick{$\ket{\psi}$} & & \qw{i} \setwiretype{q} & \qw & \qw & \qw & \qw & \qw \setwiretype{b} \\
& & \qw{j} \setwiretype{q} & \qw & \qw & \qw & \qw &  \setwiretype{n}
\end{quantikz}
    \caption{Ideal functionality that occurs for a non-Abelian representation when the circuit of \Cref{fig:fourier-sampling}, is viewed in the Fourier-transformed bases of both wires (i.e. the group Fourier transform on the top wire and the representation Fourier transform on the bottom wire).\protect\footnotemark
    { }Note that the wire corresponding to the state index $j$ is swapped out completely, such that after the second representation Fourier transform, no information about it is left over at the bottom. Thus it is not copied out, but rather \emph{extracted}.}
    \label{fig:fourier_extraction_from_rep_ideal}
\end{figure}
\footnotetext{
    For simplicity, we are only depicting the functionality of the circuit when the ancilla registers start in all zeros state (which is interpreted as the trivial irrep of the left regular representation). The effect of the circuit for general states of the ancilla registers is considerably more complicated. Also, we depict the entanglement created between the second ancilla register and the working register as a $k$-qubit EPR pair. This is just for visualization, since in reality, the EPR pair we get has dimension which depends on the dimension of the measured irrep. We omit this for simplicity.
}

\label{sec:fourier_extraction_tech_over}
The ideal functionality of \fse in general is shown in \Cref{fig:fourier_extraction_ideal}, where the irrep label $\lambda$ is copied out to an ancilla, and the state index $j$ is extracted.

\begin{figure}[H]
    \centering
\begin{quantikz}[wire types={q,q,n,b,n},classical
gap=0.08cm]
\lstick{$\ket{0^n}$} & \qw & \qw & \targ{} & \permute{1, 5, 3, 4, 2} & \qw & \qw & \qw \\
\lstick{$\ket{0^n}$} & \qw & \qw & \qw & \qw & \qw& \qw & \qw \\
& \gate[3]{\mathrm{QFT}_{\rep}} & \qw{\lambda}  \setwiretype{q} & \ctrl{-2} & \qw & \qw & \gate[3]{\mathrm{QFT}_{\rep}^\dagger} & \setwiretype{n} & \\
\lstick{$\ket{\psi}$} & & \qw{i} \setwiretype{q} & \qw & \qw & \qw & \qw & \qw \setwiretype{b} \\
& & \qw{j} \setwiretype{q} & \qw & \qw & \qw & \qw &  \setwiretype{n}
\end{quantikz}
    \caption{Ideal functionality of \fse.}
    \label{fig:fourier_extraction_ideal}
\end{figure}

\subsection{Duality between \FSE and the Representation of a Group}


The duality theorem 
\ifstocproceedings
    (\Cref{thm:duality_informal})
\else
    (\Cref{thm:duality_exact}) 
\fi
states that for all groups where the quantum Fourier transform for the left-regular representation is efficient, implementing a group representation is efficient if and only if the corresponding \fse is efficient.  In some sense, this means that quantum \fse is the ``right'' way to generalize the distinguishing task of~\cite{aaronson2020hardness} and Fourier sampling to non-Abelian groups. It is computationally equivalent to implementing the representation.

In order to prove the duality theorem, we need to provide an efficient implementation of a \fse given a representation, and vice versa.  From an implementation view-point, both of the directions of our duality can be viewed as generalizations of \cite{aaronson2020hardness} to the case of non-Abelian groups.  

The forward direction is the same as the Fourier sampling circuit we considered above, with the observation that it performs a \fse for general non-Abelian groups.
That is, given a representation $\rep: G \mapsto U(\mathcal{H})$ of a finite group $G$, we can implement a \fse of $G$ in $\mathcal{H}$ via two applications of the quantum Fourier transform on the plain group and a single application of the controlled representation, via the circuit in \Cref{fig:duality_1}.

\begin{figure}[H]
    \centering
    \begin{quantikz}
        \lstick{$\ket{0}$} & \gate{\mathrm{QFT}^{\dagger}} & \ctrl{1} & \gate{\mathrm{QFT}} & \qw \\
        \lstick{$\ket{\psi}$} & \qw & \gate{\rep(g)^{\dagger}} & \qw & \qw \\
    \end{quantikz}
    \caption{Implementing a \fse, given the ability to implement a representation $\rep$ of a group $G$.}
    \label{fig:duality_1}
\end{figure}

In the other direction, we show that given access to a \fse $\mathcal{M}_{\rep}$ (and its inverse $\mathcal{M}_{\rep}^{\dagger}$), we can implement the corresponding representation of $G$ in $\mathcal{H}$ 
via the circuit in \Cref{fig:duality_2_full}.  Here we assume that the \fse outputs three registers: one containing a label of an irreducible representation $\lambda \in \widehat{G}$, one containing the \arch $\ket{\phi^{\lambda}_{i}}$ of a \manifestation $i \in [\multlambda]$, and one containing a state index label $j \in [\dimlambda]$.  


\begin{figure}[H]
    \centering
    \ifstocproceedings
    \resizebox{\linewidth}{!}{
    \fi
        \begin{quantikz}[wire types={q,q,n,b,n},classical gap=0.08cm]
            \lstick{$\ket{g}$} & \qw & \qw & \qw & \ctrl{2} &\qw &\qw & \qw &\qw \rstick{$\ket{g}$} \\
            \lstick{$\ket{0}$} & \qw & \qw{i=0} & \gate[3]{\mathrm{QFT}^{\dagger}} & \setwiretype{n} & \gate[3]{\mathrm{QFT}} & \setwiretype{q} & \qw & \qw  \rstick{$\ket{0}$} \\
            &\gate[3]{\mathcal{M}_{\rep}} & \qw{\lambda} \setwiretype{q} &\qw &  \gate{G} &  &  &\gate[3]{\mathcal{M}_{\rep}^{\dagger}} & \setwiretype{n} \\
            \lstick{$\ket{\psi}$} & & \qw{j} \setwiretype{q} &\qw & \setwiretype{n} &  & \setwiretype{q} & \qw & \setwiretype{b} \rstick{$\rep(g) \ket{\psi}$}\\
            & & \qw{\phi_{i}^{\lambda}}\setwiretype{b} & \qw & \qw & \qw & \qw & \qw & \setwiretype{n} \\
        \end{quantikz}
    \ifstocproceedings
    }
    \fi
    \caption{Implementing a representation $\rep$ of a group $G$, given a \fse $\mathcal{M}_{\rep}$. The controlled $G$ gate is a left group operation in the group $G$.}
    \label{fig:duality_2_full}
\end{figure}

We also prove an approximate version of this duality, where $\rep$ and $\mathcal{M}$ are approximate implementations of a representation or \fse. In this case we can analyze the performance of our circuits by relating them to ``nearby'' exact representations and \fses, which are guaranteed to exist due to the work of Gowers and Hatami~\cite{gowers2016inverse}.

\subsection{Quantum Money, Lightning, and Fire from Non-Abelian Group Actions}

Next, we describe a protocol for constructing quantum money and lightning from any non-Abelian group action.  Both the minting and verification procedure mimic the implementation of \fse and Fourier sampling, except that we use the specific representation that comes from applying the group action in superposition.  For both, we first generate a uniform superposition over group elements (this is essentially the same as the first quantum Fourier transform), apply the group action in superposition, and measure the first register in the Fourier basis.  

For minting (\Cref{fig:lightning_minting}), we start with a fixed set element of the set acted on by the group. The serial number will be the measured label $\lambda$ of an irreducible representation of $G$.  We show that this produces a serial number $\lambda$ with probability proportional to the Plancherel measure of $\lambda$.

\begin{figure}[H]
    \centering
    \begin{quantikz}[classical gap=0.05cm]
    \lstick{$\sum_{g}\ket{g}$} & \ctrl{1}&  \gate{\mathrm{inv}}& \gate{\mathrm{QFT}}& \meter{}& \setwiretype{n} \cw & \cw{s = \lambda, i, j} \\
    \lstick{$\ket{x}$}& \gate{\rep(g)}& \qw& \qw& \qw & \qw & \qw{\ket{\$^{\lambda}_{i, j}}}
    \end{quantikz}
    \caption{Minting procedure in the quantum lightning scheme from non-Abelian group actions.  Here $\rep$ is a group action of group $G$ with starting element $x$, and $\mathrm{inv}$ maps $g$ to $g^{-1}$ (this inversion is simply in order to give us the desired $\rep(g)^{\dagger}$ instead of $\rep(g)$).}
    \label{fig:lightning_minting}
\end{figure}

The verification algorithm (\Cref{fig:lightning_verification}) implements the same measurement as mint, except that it starts with the claimed banknote $\ket{\pounds}$, and it checks if the measured irreducible representation label $\lambda'$ matches the claimed serial number.  

\begin{figure}[H]
    \centering
    \begin{quantikz}
    \lstick{$\sum_{g}\ket{g}$} & \ctrl{1}&  \gate{\mathrm{inv}}& \gate{\mathrm{QFT}}& \meter{}& \setwiretype{n} \cw & \cw{\lambda'} \\
    \lstick{$\ket{\pounds}$}& \gate{\rep(g)}& \qw& \qw& \qw & \qw&
    \end{quantikz}
    \caption{Verification procedure in the quantum lightning scheme from non-Abelian group actions.  The verification algorithm measures an irreducible representation label $\lambda'$ and accepts if $\lambda'$ is equal to $\lambda$ from the serial number. Note that this circuit may modify the banknote even if it is valid. However, the resulting modified banknote will also be valid for the same serial number, and furthermore, the circuit can be uncomputed to return the original banknote.}
    \label{fig:lightning_verification}
\end{figure}

\subsubsection{Security of the Scheme from Preactions}
\label{sec:preaction_sec_tech_over}

The security of the scheme is based on a new assumption we call ``preaction security'' 
\ifstocproceedings
(see the full version of the paper for the formal statement).
\else
(see \Cref{assum:preaction-indist} for the formal statement).
\fi
Given a group action of $G$ on set $X$, and a starting set element $x \in X$, the preaction of a group element $h \in G$ maps $g * x$ to $gh^{-1} * x$, and the preaction security assumption is that it is hard to distinguish the case when a random action has been applied to a register from the case when a random preaction and action (which we call a biaction) have been applied to the register.  

Quantum lightning security requires it to be inefficient for any adversary to produce two banknote states corresponding to the same serial number. 
To gain intuition for why preaction security implies the security of our quantum lightning scheme, we need to think about what a preaction does to a state in the Fourier basis.  Intuitively, in the Fourier basis of the left-regular representation, acting by a group element $g$ moves a basis state $\ket{\lambda, i, j}$ to a different vector in the same \manifestation, i.e. $\sum_{k} \alpha_{k} \ket{\lambda, i, k}$.  On the other hand, a preaction acts by moving a vector in a particular \manifestation to a superposition over \manifestations, i.e. $\sum_{k} \alpha_{k} \ket{\lambda, k, j}$.  
Because (in general) preactions are hard to implement, an adversary cannot simply measure the multiplicity label $i$.
This means that with only a single copy of the quantum money state (which is a Fourier basis state), it is difficult to determine if this multiplicity label changes.  However, an adversary \emph{can} perform \fse, which allows them to extract an ``\arch'' state $\ket{\phi^{\lambda}_i}$ that depends only on $i$.  

Imagine a quantum lightning adversary that is given two copies of a quantum lightning state with the \emph{same} serial number. 
Assume for simplicity that they are also in the same \manifestation (of course, we do not make this assumption in our proof, but the intuition is much clearer for this case).
The adversary can apply the random action or biaction to the first of the two registers, perform \fse on both, and do a swap test on the archetype states to see if they are equal.  A random action, which does not change the \manifestation, will not change the archetype state that results from the \fse, and thus the swap test will succeed with probability $1$.  A random preaction, on the other hand, will move the first register to a random \manifestation, and therefore pass the swap test with probability close to $\frac{1}{2}$ (the exact probability depends on the structure of the group's irreps).  More care is needed when the quantum lightning states produced by the adversary fall into different \manifestations, but we leave out the details here.  Thus, this adversary breaks preaction security, and, by contradiction, if we start with a preaction secure group action, our quantum lightning scheme is secure.

\subsubsection{Instantiation with the Group Action of the McEliece Cryptosystem}

While we are able to construct quantum money and lightning from any group action satisfying preaction security, we also propose a few concrete group actions from which to instantiate it.
A notable example that we believe may satisfy the necessary conditions is based on a group action implicit in the McEliece cryptosystem~\cite{Mceliece78}.

Our group action will represent the symmetric group $S_n$ over the set $X$ of $n \times t$ (for $t = \poly(n)$) matrices with entries from $\mathbb{F}_{q^m}$.  Given a matrix $M$ and permutation $\sigma$, $\sigma * M$ is the matrix where we first permute the columns of $M$ according to $\sigma$, and then row-reduce the resulting matrix.  The starting element of the group action will be the parity check matrix of a Goppa code~\cite{goppa1970new}.  In particular, we sample elements $g_1, \ldots, g_t \in \mathbb{F}_{q^{m}}$, which defines a degree-$t$ univariate polynomial $g(z) = \sum_{i = 1}^{t} g_t z^t$, and $\alpha_1, \ldots, \alpha_n \in \mathbb{F}_{q^{m}}$ satisfying $g(\alpha_{j}) \neq 0$.  Then the starting element corresponding to the choice of $g_1, \ldots, g_t$ and $\alpha_1, \ldots, \alpha_n$ is the following
\begin{equation*}
    \ifstocproceedings
    \scalebox{0.65}{$
    \fi
    \begin{pmatrix}
        g_t g(\alpha_1)^{-1} & g_t g(\alpha_2)^{-1} & \ldots & g_t g(\alpha_n)^{-1}\\
        (g_t \alpha_1 + g_{t-1})g(\alpha_1)^{-1} & (g_t \alpha_2 + g_{t-1})g(\alpha_2)^{-1} & \ldots & (g_t \alpha_n + g_{t-1})g(\alpha_n)^{-1}\\
        \vdots & \vdots & \ddots & \vdots\\
        (g_t \alpha_1^{t-1} + g_{t-1} \alpha_{1}^{t-2} \ldots + g_1) g(\alpha_1)^{-1} & \ldots & \ldots & (g_t \alpha_n^{t-1} + \ldots + g_1) g(\alpha_n)^{-1}
    \end{pmatrix}
    \ifstocproceedings
    $}
    \fi
    \,.
\end{equation*}
It is commonly assumed that for certain distributions over $g$ and $\alpha$, these matrices are indistinguishable from a uniformly random matrix even to quantum algorithms~\cite{singh2019code}, and that applying a random column permutation \emph{and} a random $n \times n$ matrix from the right yields a matrix that is indistinguishable from a uniformly random matrix even to adversaries who know the starting element, however these do not directly imply the preaction security of this group action.  
We conjecture that preaction security holds for this group action, and show that under this assumption, instantiating our framework with this group action yields a concrete secure quantum lightning in the plain model.

\subsubsection{Quantum Fire}

As part of this work, we also provide a framework for constructing quantum fire in the plain model. We show conditions under which these states are efficiently clonable, but we leave it as an open question to prove that the scheme satisfies untelegraphability.  The construction of quantum fire is similar to the construction of quantum lightning and money from non-Abelian group actions, however we make use of the insight that applying a controlled group operation between registers can be used to ``copy'' information about the irreducible representation to allow for efficient cloning of the states in an inaccessible Fourier basis.

We now describe the construction.  Let $G$ and $H$ be two groups and $f: G \mapsto H$ be a one-way injective homomorphism between the two (i.e. $f(g) f(h) = f(gh)$).%
\footnote{We require that $f$ is one-way because otherwise there is a way to telegraph the quantum fire states we describe, so the one-wayness is necessary for security of the scheme, although it is not clear that it is sufficient. }  
We can define a representation of $G$ by letting $\rep(g)\ket{h} = \ket{f(g) \cdot h}$.  
Then we can spark (the fire equivalent of minting) and verify quantum fire using the same circuits as in the quantum lightning construction, but with this representation.
\begin{figure}[H]
    \centering
    \begin{quantikz}[classical gap=0.05cm]
    \lstick{$\sum_{g}\ket{g}$} & \ctrl{1}&  \qw & \gate{\mathrm{QFT}}& \meter{}& \setwiretype{n} \cw & \cw{s = \lambda, i, j} \\
    \lstick{$\ket{\id_{H}}$}& \gate{\rep(g)}& \qw& \qw& \qw & \qw & \qw{\ket{\phi^{\lambda}_{i, j}}}
    \end{quantikz}
    \caption{Sparking procedure in the quantum fire scheme from a one-way injective homomorphism.  Here $\rep(g)$ acts on elements of $H$ by mapping $h \mapsto f(g) \cdot h$, where $f$ is the homomorphism, and $\id_{H}$ is the identity element of $H$.  Verification only checks that the label $\lambda$ is the same, as in the quantum lightning scheme.}
    \label{fig:fire_sparking}
\end{figure}
To clone states of this form, we apply the following circuit, which will produce two quantum fire registers $\sum_{k \in [\dimlambda]} \ket{\phi^{\lambda}_{i, k}} \ket{\phi^{\lambda}_{k, j}}$ with an entangled value of $k \in [\dimlambda]$.
\begin{figure}[H]
    \centering
    \begin{quantikz}[classical gap=0.05cm]
        \lstick{$\ket{\phi^{\lambda}_{i, j}}$} & & \gate{*_{H}} & & \rstick{$\ket{\phi^{\lambda}_{i, k}}$}\\
        \lstick{$\sum_{h \in H_f} \ket{h}$} & \gate{\mathrm{inv}} & \ctrl{-1} & \gate{\mathrm{inv}} & \rstick{$\ket{\phi^{\lambda}_{k, j}}$}
    \end{quantikz}
    \label{fig:quantum_fire_cloning}
    \caption{Cloning a quantum fire state.  Here $H_f$ is the image of the one-way injective homomorphism, $\mathrm{inv}$ is an operation that maps $h \in H$ to $h^{-1} \in H$, and $*_{H}$ is the group operation of $H$.}
    \label{fig:fire_cloning}
\end{figure}

While this process produces an entangled state between the two registers, we note that it is not difficult to then unentangle the registers if we wish. This is done by performing a \fse on both states to extract the entangled indices $\sum_{k \in [\dimlambda]} \ket{k} \ket{k}$. We can then perform the reverse operation to inject a new pair of registers in the state $\ket{j} \ket{i}$, to get the tensor product state $\ket{\phi^{\lambda}_{i, j}} \otimes \ket{\phi^{\lambda}_{i, j}}$, that is, two identical copies of the original state.

Thus, the registers first become entangled, before then becoming unentangled again. 
This is thus an efficient cloning operation, that surprisingly does not go through any classical description of the state in question. 
While we leave open the problem of showing untelegraphability for this scheme under certain assumptions about $G$, $H$, or $f$, the fact that this method of cloning these states is inherently an entangling operation justifies that it should not be possible to do across a classical channel.
\ifstocproceedings
We encourage the reader to view the full version of the paper for more details.
\fi

\section{Open Problems and Future Directions}

We discuss directions for future work here.

\begin{enumerate}
    \item Can \fse be used as a subroutine for other problems?  For example, Fourier measurements play an important role in algorithms for the Abelian hidden subgroup problem.  For the dihedral hidden subgroup problem it is known that the outcome of Fourier \emph{measurements} on coset states does not yield sufficient information about the hidden subgroup, indicating that this is not the right way to generalize algorithms for the Abelian hidden subgroup problem.  Perhaps another generalization, such as \fse, is the right way to generalize such algorithms to non-Abelian groups.
    \item Can we construct other fully-quantum primitives using preaction secure group actions?  There are many other primitives, like quantum copy-protection for certain classes of functions~\cite{CCC:Aaronson09, aaronson2021new, coladangelo2024quantum}, one-shot signatures~\cite{STOC:AGKZ20}, or quantum state obfuscation~\cite{bartusek2023obfuscation, coladangelo2024use, bartusek2024quantum} for which we have constructions relative to oracles and sometimes post-quantum indistinguishability obfuscation, but no plausible candidates in the plain model otherwise.  Can we use non-Abelian group actions to find plausible candidate constructions for these primitives, and preaction security to prove their security?
    \item As written, the preaction security game is not efficiently verifiable.  In particular, the challenger in the game must themselves apply a preaction, which should be inefficient if the security holds.  One could define a notion of trap-door preaction security, where the challenger has a secret key that allows them to efficiently perform a preaction, while any adversary without the key cannot distinguish between a random biaction and random action. However, we leave open the task of identifying good candidates for such trapdoor preaction security.
    \item Can the untelegraphability of our quantum fire scheme be proven under any computational assumption about the groups $G$ and $H$.  We note that any such proof may also provide (under the same assumptions) a separation between $\mathsf{QMA}$ and $\mathsf{QCMA}$, as noted in \cite{ITCS:NZ23}.  
    \item The forward direction (going from the representation to Fourier extraction) of our approximate generalized duality theorem, when applied to $\mathbb{Z}_2$, yields a worse distinguishing advantage than the duality theorem of \cite{aaronson2020hardness}.  This is because the implementation of an approximation of a representation might be off of the true representation by a phase that can depend on the input state.  In the case of $\mathbb{Z}_2$, since there are only two states that the representation acts on, the phase can be corrected by applying a global phase corresponding to the difference of the two phases.  However, it is not clear how to generalize this strategy past $\mathbb{Z}_2$, leading to our approximate duality not being tight.  Can a tight version of our approximate duality theorem be proven?
\end{enumerate}

\section{Preliminaries}

\subsection{Quantum Preliminaries}

A register $\repspace$ is a named finite-dimensional Hilbert space.  When two registers appear next to each other, as in $\reg{AB}$, this refers to the tensor product space of $\reg{A}$ and $\reg{B}$.  We write $\tr(\cdot)$ to denote the trace, and $\tr_{\reg{B}}(\cdot)$ to denote the partial trace over a register $\reg{B}$.  We denote by $\|X\|_1 = \tr(|X|)$ the trace norm, where $|X| = \sqrt{XX^{\dagger}}$.  For a vector space $V$, we write $\mathrm{GL}(V)$ to denote the general linear group from $V$ to itself, i.e. invertible square matrices, and $U(V)$ to denote the unitary group. For two matrices in $\mathrm{GL}(V)$, we define the Hilbert-Schmidt inner product as follows.

\begin{definition}[Hilbert-Schmidt inner product]
    Let $A, B \in \mathrm{GL}(\repspace)$, then we define the \emph{Hilbert-Schmidt inner product} between $A$ and $B$ to be
    \begin{equation*}
        \langle A, B \rangle = \frac{1}{\dim(\repspace)} \tr\left[AB^{\dagger} \right]\,. 
    \end{equation*}
    This implies a norm in the natural way: $\|A\| = \sqrt{\langle A, A \rangle}$.
\end{definition}

\subsection{Representation Theory}

\begin{definition}[Representation]
Let $G$ be a finite group.  Then a function $\rep: G \mapsto U(\repspace)$ is a representation of $G$ if the following holds for all group elements $g, h \in G$:
\begin{equation*}
    \rep(g) \rep(h) = \rep(gh)\,.
\end{equation*}
The vector space $\repspace$ is called a representation space of $G$.  We note that representations need not be defined over Hilbert spaces (they can be defined over any vector space), but we will only ever consider representations that output unitaries in Hilbert spaces. We use the notation $\dim(\rep)$ to denote the dimension of the representation space $\repspace$.  
\end{definition}

We will also need a notion of a function being ``almost'' a representation.  The following is the definition of an $\epsilon$-approximate representation (in Hilbert-Schmidt norm), taken from \cite{gowers2016inverse}.
\begin{definition}[$\epsilon$-approximate representation~\cite{gowers2016inverse}]\label{def:eps-approx-rep}
    Let $G$ be a group, and $\rep: G \mapsto U(\repspace)$ be a function taking group elements to unitaries over $\repspace$.  $\rep$ is a $\epsilon$-approximate representation if the following holds:
    \begin{equation*}
        \avg_{g, h \in G}\left[\Re \left\langle\rep(g)^{\dagger}\rep(h), \rep(g^{-1}h)^{\dagger}\right\rangle\right] \geq 1 - \epsilon\,.
    \end{equation*}
    Here $\Re$ denotes the real component of a complex number. 
\end{definition}

We use the following additional definition of an $\epsilon$-close representation, which is the notation of being close to an exact representation of a group, up to an isometry $V$.
\begin{definition}[$\epsilon$-close representation]
    Let $G$ be a group and $\rep: G \mapsto U(\repspace)$ be a function taking group elements to unitaries over $\repspace$.  We say that $\rep$ is $\epsilon$-close to a representation of $G$ if there exists another representation of $G$, $\otherrep: G \mapsto U(\otherrepspace)$ and an isometry $V: \repspace \mapsto \otherrepspace$ such that
    \begin{equation*}
        \avg_{g \in G} \left\| \rep(g) - V^{\dagger} \otherrep(g) V \right\|^2 \leq \epsilon\,.
    \end{equation*}
\end{definition}

We will also need some definitions and facts from character theory.  A reference for these can be found in, e.g.~\cite{serre1977linear}.
\begin{definition}[Irreducible representation]
    A representation $\rep: G \mapsto \mathrm{GL}(\repspace)$ is an \emph{irreducible representation} of $G$ if for all subspaces $\reg{W} \subset \repspace$, $\rep(g)\reg{W} \not\subseteq \reg{W}$.  We sometimes refer to $\repspace$ as the irreducible representation of $G$. Irreducible representations are often called ``\emph{irreps}''.
\end{definition}

We will use the notation $\irrep$ to denote irreducible representations (when it is clear from context), and $\rep$ to denote representations in general.

\begin{definition}[Dual of a group]
    The \emph{dual} of a group $G$, denoted $\widehat{G}$, is the set of all irreducible representations of $G$, up to equivalence by a unitary transformation. For an Abelian group, $\widehat{G}$ will itself have a group structure, but this is not generally the case for non-Abelian groups.
\end{definition}

We will use the notation $\lambda \in \widehat{G}$ to denote a ``label'' of an irreducible representation (think: a string that uniquely determines the identity of a particular irrep), and $\varrho_{\lambda}$ denote the corresponding representation (the function from group elements to unitaries), although we may refer to both as an irreducible representation when it is clear from context.

\begin{lemma}[Size of the dual]
    The size of the dual, $\widehat{G}$, of a group is equal to the number of conjugacy classes of $G$. In particular, for a finite group $G$, $\widehat{G}$ is also finite.
\end{lemma}

\begin{definition}[Character]    
    Let $\rep: G \mapsto \mathrm{GL}(\repspace)$ be a representation of $G$.  We define the \emph{character} of $\rep$ to be 
    \begin{equation*}
        \chi_{\rep}(g) = \tr[\rep(g)]\,.
    \end{equation*}
\end{definition}

\begin{definition}[Inner product of characters]
    Let $\chi_{\rep}$ and $\chi_{\otherrep}$ be two characters, then we define their inner product to be
    \begin{equation*}
        \braket{\chi_{\rep} | \chi_{\otherrep}} = \frac{1}{|G|}\sum_{g \in G} \chi_{\rep}(g) \chi_{\otherrep}^{\dagger}(g)\,.
    \end{equation*}
\end{definition}

\begin{lemma}[Irreps are norm $1$]
    For every irreducible representation of a group $G$, the following holds
    \begin{equation*}
        \braket{\chi_{\varrho} | \chi_{\varrho}} = 1\,.
    \end{equation*}
\end{lemma}

\begin{lemma}[Decomposition into irreps]
\label{lem:decomposition_into_irreps}
    Let $\rep$ be a representation of a group $G$ with representation space $\repspace$, and let $\dimlambda$ be $\braket{\chi_{\rep}, \chi_{\varrho_{\lambda}}}$.  Then the following holds: 
    \begin{equation*}
        \repspace \simeq \bigoplus_{\lambda \in \widehat{G}} \reg{W}_{\lambda}^{\oplus \dimlambda}\,.
    \end{equation*}
    Where $\reg{W}_{\varrho_{\lambda}}$ is the irreducible representation space of $\varrho_{\lambda}$.  Furthermore, the decomposition into $\reg{W}^{\oplus \dimlambda}$ is unique, the decomposition into further subspaces depends on the choice of basis. In the basis of $\bigoplus_{\lambda} \reg{W}_{\lambda}^{\oplus \dimlambda}$, $\rep(g)$ looks like:
    \begin{equation*}
        \sum_{\lambda \in \widehat{G}} \Pi_{\reg{W}_{\lambda}} \varrho_{\lambda}(g) \Pi_{\reg{W}_{\lambda}}\,.
    \end{equation*}
    Here $\Pi_{\reg{W}_{\lambda}}$ is the projector onto $\reg{W}_{\lambda}$.
\end{lemma}

\begin{definition}[\ifmanifestation Manifestations and multiplicity
        \else
            Multiplicity 
        \fi
        of an irreducible representation]
    The number of irreducible representation spaces corresponding to the same irrep $\lambda$ in $\rep$
    is its \emph{multiplicity}, which we denote as $\multlambda^{\rep}$ (or $\multlambda$ when $\rep$ is clear from context). 
    The 
    \ifmanifestation
    \manifestations%
    \footnote{It is common to refer to the different irreducible representation subspaces $\reg{W}^{\lambda}_i$ on which the representation $\rep$ acts as irrep $\lambda$ as ``copies'' of of the irreducible representation. We prefer the word ``\manifestations'' to avoid confusion later with the notion of copies of a state due to cloning.}
    \else
    \manifestations
    \fi
    of an irreducible representation $\lambda$ within a representation $\rep$, are the subspaces $\reg{W}_{\lambda}$ that are irreducible representation spaces of $\lambda$.
    The direct sum of all of the \manifestations of $\lambda$ is the course Fourier subspace, also known as the \emph{isotypic component} of $\lambda$. In fact, the breakdown of the isotypic component of $\lambda$ into its \manifestations is basis-dependent and not unique.
\end{definition}

\begin{definition}[Right/left regular representation]
    The left regular representation of a group $G$ is the following function.
    \begin{equation*}
        \mathcal{L}(h) = \sum_{g \in G} \ket{hg}\!\!\bra{g}\,.
    \end{equation*}
    The right regular representation of a group $G$ is the following function.
    \begin{equation*}
        \mathcal{R}(h) = \sum_{g \in G} \ket{gh^{-1}}\!\!\bra{g}\,.
    \end{equation*}
\end{definition}

\begin{lemma}
    For all groups $G$ and all irreducible representations of $G$, the following holds
    \begin{equation*}
        \braket{\chi_{\mathcal{L}} | \chi_{\varrho}} = \dim(\varrho)\,.
    \end{equation*}
    and similarly for the right regular representation.
\end{lemma}

Using this fact, together with the fact that the character of the right (or left) regular representation is equal to $|G|$ at the identity, and $0$ elsewhere, we have:
\begin{lemma}
\label{lem:sum_square_dimension}
Let $G$ be a finite group, then the following holds.
\begin{equation*}
    \sum_{\lambda \in \hat{G}} \dim(\varrho_{\lambda})^2 = |G|\,.
\end{equation*}
\end{lemma}

\begin{definition}[Plancherel measure]
    The Plancherel measure is a probability distribution over irreducible representations of a group $G$.  The Plancherel measure of an irreducible representation with label $\lambda$ is given by 
    \begin{equation*}
        \mu(\lambda) = \frac{\dim(\varrho_{\lambda})^2}{|G|}\,.
    \end{equation*}
\end{definition}

We can see that this corresponds to selecting an irreducible representation according to its ``weight'' in the sum of \cref{lem:sum_square_dimension}. 
A concept we will be interested in is the maximum Plancherel measure of any irreducible representation of the group. 
For example, for the symmetric group, upper and lower bounds are given by the following lemma.
\begin{lemma}[Plancherel measure of the symmetric group~\cite{vershik1985asymptotic}]
The following inequalities hold for constants $c_0 = 0.2313$ and $c_1 = 2.5651$
\begin{equation*}
    e^{-\frac{c_1}{2}\sqrt{n}} \sqrt{n!} \leq \max_{\lambda \in \widehat{S_n}} \dim(\varrho_{\lambda}) \leq e^{-\frac{c_0}{2}\sqrt{n}} \sqrt{n!}\,.
\end{equation*}
\end{lemma}
Thus, the maximum Plancherel measure of an irreducible representation of the symmetric group, $S_n$, is $e^{-c_0 \sqrt{n}}$, which is negligible in $n$.  

\begin{lemma}[Schur orthogonality relations~\cite{issai1905neue}]
    \label{sec:schur-orthogonality}
    Let $\varrho, \sigma \in \widehat{G}$ be irreducible representations of $G$. Then we have that:
    \begin{align*}
           \sum_{g\in \G} \; \varrho(g)_{i, j}^*\;\sigma(g)_{k, \ell} =
           \frac{\abs{G}}{\dim(\varrho)}
           \delta_{\varrho, \sigma} \delta_{i, k}\delta_{j, \ell}
           \,.
    \end{align*}
\end{lemma}

\subsubsection{Quantum Fourier Transform}

Now we define the quantum Fourier transform for non-Abelian groups. 
\begin{definition}[Quantum Fourier transform]
    Let $\dimlambda$ be the dimension of $\varrho_{\lambda}$ for every irreducible representation of a group $G$.  The quantum Fourier transform over a general group $G$ is the following unitary transformation 
    \begin{equation*}
        \mathrm{QFT}_{G} = \sum_{g \in G}\sum_{\substack{\lambda \in \widehat{G}, \\i, j \in [\dimlambda]}} \sqrt{\frac{\dimlambda}{|G|}} \; \varrho_{\lambda}(g)_{j, i} \ket{\lambda, i, j}\!\!\bra{g}\,.
    \end{equation*}
    Its inverse is 
    \begin{equation*}
        \mathrm{QFT}_{G}^{\dag} = \sum_{g \in G}\sum_{\substack{\lambda \in \widehat{G}, \\i, j \in [\dimlambda]}} \sqrt{\frac{\dimlambda}{|G|}} \; \varrho_{\lambda}(g^{-1})_{i, j} \ket{g}\!\!\bra{\lambda, i, j}\,.
    \end{equation*}
    We will often refer to either one as the quantum Fourier transform over $\G$, and it will be clear from context which one we mean.
\end{definition}

We note that for Abelian groups, every irreducible representation is dimension $1$, so the sum over $i, j$ goes away, and we recover the usual Abelian quantum Fourier transform.

\begin{definition}[Fourier basis states]
    For a group $\G$, let $\{\ket{\mathcal{L}^{\lambda}_{ij}}\}_{\substack{\lambda \in \widehat{G}, \; i, j \in [\dim(\varrho_{\lambda})]}}$, where $\ket{\mathcal{L}^{\lambda}_{ij}} := \sqrt{\frac{\dimlambda}{|G|}} \sum_{g \in G} \varrho_{\lambda}(g^{-1})_{i, j} \ket{g}$, be the basis recovered by applying $\mathrm{QFT}_{G}^{\dag}$ to $\{ \ket{\lambda, i, j} \}_{\substack{\lambda \in \widehat{G}, \; i, j \in [\dim(\varrho_{\lambda})]}}$. We call this the \emph{(left-regular) Fourier basis} of $\G$.
\end{definition}

\subsection{Fourier Measurements}

\subsubsection{Coarse Fourier Measurement}
\begin{definition}[Coarse Fourier measurement]
    The coarse Fourier measurement\footnotemark{} is the measurement of the subspaces corresponding to the irreducible representations, but not the basis of the subspaces. Formally, for a group $G$ and representation $\rep$, the coarse Fourier measurement is given by the POVM
    \begin{equation*}
        \left\{\Pi_{\reg{W}_{\lambda}^{\oplus \multlambda}}\right\}_{\lambda \in \widehat{G}}\,.
    \end{equation*}
    Here the decomposition into unique subspaces $\reg{W}_{\lambda}^{\oplus \multlambda}$ is given by \cref{lem:decomposition_into_irreps}.
\end{definition}

Performing a measurement using the coarse Fourier measurement to produce a random irreducible representation label is known in the literature as \emph{weak Fourier sampling}.  

\subsubsection{Fine Fourier Measurement}

\begin{definition}[Fine Fourier measurement]
    Let $G$ be a finite group and $\rep$ be a representation of that group.  Let $\reg{W}_{\lambda}$ be an irreducible representation space of $G$, let $\multlambda = \braket{\chi_{\rep}|\chi_{\varrho_{\lambda}}}$, and $\{\ket{\psi^{\lambda}_{i, j}}\}_{i \in [\dim(\varrho_{\lambda})], j \in [\multlambda]}$ be a basis for the subspace $\reg{W}_{\lambda}^{\oplus \multlambda}$.  Then the fine Fourier measurement\hyperref[foot:fourier-meas-samp]{\footnotemark[\value{footnote}]} is given by the POVM
    \begin{equation*}
        \{\proj{\psi^{\lambda}_{i, j}}\}_{\lambda, i, j}\,.
    \end{equation*}
\end{definition}
Performing a measurement using the fine Fourier measurement (for any choice of basis) is known in the literature as \emph{strong Fourier sampling}.
\footnotetext{
    \label{foot:fourier-meas-samp}
    The literature often refers to course and fine Fourier measurements as \emph{weak} and \emph{strong} Fourier sampling, respectively. We prefer to use the course and fine terminology, capturing how coarse- or fine-grained the decomposition of the space, but we will use the terms interchangeably.
}

\subsubsection{\FSE}

For our purposes, we require a stronger notion than Fourier \emph{measurement}. We introduce a stronger notion called \emph{\fse}. Unlike a Fourier measurement which measures and outputs a classical value for each irreducible representation space $\reg{W}^{\lambda}_{i}$, \fse extracts a coherent quantum state out of each $\reg{W}^{\lambda}_{i}$, maintaining the original superposition within $\reg{W}^{\lambda}_{i}$ but expressing it in the standard basis.

\begin{definition}[\fse]
    Let $G$ be a finite group and $\rep$ be a representation of that group.
    A \emph{\fse} is a coarse projective measurement $\{\Pi_\lambda\}_{\lambda \in \widehat{G}}$\textemdash 
    where each $\Pi_\lambda$ projects onto 
    $\reg{W}^{\oplus \dimlambda}_{\lambda} := \bigoplus_{i \in [\multlambda]} \reg{W}^{\lambda}_{i}$,
    the union of the \manifestations 
    of $\lambda$%
    \textemdash 
    and a \se within each subspace $\reg{W}^{\lambda}_i$.
    Specifically, let each $\reg{W}^{\lambda}_i$ have basis $\{\ket{\psi^{\lambda}_{i, j}}\}_{
    j \in [\dim(\varrho_{\lambda})]}$. Then a \fse implements a unitary
    \begin{equation*}
        \mathcal{M}: \ket{\psi^{\lambda}_{i, j}} \ket{0} \mapsto \ket{\phi^{\lambda}_{i}} \ket{\lambda, j}\,,
    \end{equation*}
    for some orthonormal set of ``\emph{archetype}'' states $\{\ket{\phi^{\lambda}_{i}}\}_{\lambda \in \widehat{G}, \, i \in [\multlambda]}$.%
    \footnote{Note that the form of the archetype states does not matter. The only requirement is that they are orthonormal so that $\mathcal{M}$ is an isometry. That is, $\braket{\psi^{\varrho}_{ij} \;|\; \psi^{\sigma}_{k\ell}} = \delta_{\varrho\sigma} \delta_{ik} \delta_{j\ell}$.}
\end{definition}

\subsection{Group Actions}
A group action is a representation of a group that appears often in the field of cryptography.  Formally, it is define as follows
\begin{definition}[Group action]
    A group action consists of a family of groups $G = (G_{n})_{n}$, a family of sets $\mathcal{X} = (\mathcal{X}_{n})_{n}$, and a binary operation $*: G_{n} \times \mathcal{X}_{n} \mapsto \mathcal{X}_{n}$ satisfying the following properties
    \begin{itemize}
        \item \textbf{Identity:} Let $\id \in G$ be the identity element, then $0 * x = x$ for all $x \in \mathcal{X}_{n}$.
        \item \textbf{Representation:} For all $g, h \in G_{n}$ and $x \in \mathcal{X}_{n}$, $gh * x = g * (h * x)$.
    \end{itemize}
    We sometimes require the following additional properties.
    \begin{itemize}
        \item \textbf{Efficiently computable:} There is a quantum polynomial-time algorithm that on input $1^{n}$ outputs a description of $G_{n}$ and an element $x_{n} \in \mathcal{X}_{n}$.  The binary operation $*$ is also computable by a quantum polynomial-time algorithm.
        \item \textbf{Efficiently recognizable:} There is a quantum polynomial-time algorithm such that for any $n$ and string $y$, the algorithm accepts with probability at least $2/3$ if $y \in \mathcal{X}_{n}$ and rejects with probability at least $2/3$ if $y \not \in \mathcal{X}_{n}$. 
        \item \textbf{Transitive:} There is exactly one orbit. That is, for any two elements $x, y \in \mathcal{X}_{n}$, exists a $g \in G_{n}$ such that $y = g * x$.
        \item \textbf{Semiregular:} (also called ``free'') There are no fixed points. That is, for every $g \in G_{n}$ and $y \in \mathcal{X}_{n}$, $g * x = x$ implies that $g = \id$.
        \item \textbf{Regular:} Regular group actions are both transitive and semiregular. That is, for every $y \in \mathcal{X}_{n}$, there is exactly one $g \in G_{n}$ such that $y = g * x_{n}$.
    \end{itemize}
\end{definition}

Later in the paper, we will describe additional properties of group actions that will be useful in proving security of cryptographic primitives constructed from group actions.

\begin{definition}[Orbits of a group action]
    The \emph{orbit} of an element $x\in\mathcal{X}$ is the set of elements accessible from $x$ by acting with $\G$:
    \begin{align*}
        \mathsf{Orb}(x) = \{y \;|\; \exists g \in \G \; \textbf{ s.t. } y = g*x \}
        \,.
    \end{align*}
\end{definition}

One important property of group actions is they are representations on the Hilbert space spanned by the elements of $\mathcal{X}$.
\begin{definition}[Group Action Representation]
    A group action of $G$ defines a representation of $G$ by the following unitary:
    \begin{equation*}
        \rep(h) \ket{g * x} = \ket{hg * x}\,.
    \end{equation*}
    Note that this representation is isomorphic to a direct sum of left-regular representations on the different orbits of the group action.
\end{definition}

\subsection{Quantum Money and Quantum Lightning}

Now we define public-key quantum money and quantum lightning.  Both primitives have the similar syntax, with differences in how their key generation works.

\begin{definition}[Public-key quantum money~\cite{CCC:Aaronson09}]
A public-key quantum money scheme is a triple of efficient quantum algorithms $\mathcal{S} = (\keygen,\mint,\ver)$ where 
\begin{itemize}
    \item $\keygen$ takes as input the security parameter $1^n$ and outputs a private/public key pair $(\sk,\pk)$, 
    \item $\mint(\sk)$ outputs a pair $(s,\ket{\$^s})$ where $s$ is a string representing a serial number and $\ket{\$^s}$ is a quantum state representing a banknote,%
    \footnote{
        We will refer to these states interchangeably as either quantum money states or banknotes.
    }
    and
    \item $\ver$ takes as input the public key $\pk$, a serial number $s$, and an alleged banknote $\sigma$, and either accepts or rejects.
\end{itemize}
A public-key quantum money scheme $\mathcal{S}$ satisfies correctness if for all $n \in \mathbb{N}$, 
\[
    \Pr \left [ \ver(\pk, s, \ket{\$^s}) \text{ accepts} : \begin{array}{c}
(\sk,\pk) \leftarrow  \keygen(1^n) \\
(s,\ket{\$^s}) \leftarrow \mint(\sk) 
\end{array}
   \right ] \geq 1 - \negl(n)\,.
\]

\end{definition}

\begin{definition}[Quantum money security]
    A public-key quantum money scheme $\mathcal{S}$ satisfies \emph{$\epsilon$-quantum-money security} if for all efficient adversaries $A$, the success probability of $A$ in the counterfeit security game (\Cref{prot:quantum_money_security}) is at most $\epsilon(n)$.
\end{definition}

\begin{longfbox}[breakable=false, padding=1em, margin-top=1em, margin-bottom=1em]
\begin{algo}[{\bf Public-key Quantum Money Counterfeit Security Game}]\label{prot:quantum_money_security}
\end{algo}
\begin{enumerate}
    \item Generate $(\sk,\pk) \leftarrow \keygen(1^{n})$, $(s, \ket{\$^s})\leftarrow \mint(\sk)$ and send $(\pk, s, \ket{\$^s})$ to the adversary.
    \item Adversary returns two registers $\reg{AB}$ in some potentially entangled state $\sigma_{\reg{AB}}$.
    \item Run $\ver(\pk, s, \sigma_{\reg{A}})$ and $\ver(\pk, s, \sigma_{\reg{B}})$.  If either check rejects, then reject, otherwise accept.
\end{enumerate}
\end{longfbox}

In place of full public-key quantum money schemes, we will often make use of quantum money \emph{mini-schemes}, simpler objects that can be upgraded to public-key quantum money schemes using digital signatures~\cite{STOC:AarChr12}. Because of this effective equivalence, when it is clear from context, we will also often refer to quantum money mini-schemes as public-key quantum money.

\begin{definition}[Quantum money mini-scheme~\cite{STOC:AarChr12}]
A quantum money scheme is a pair of efficient quantum algorithms $\mathcal{S} = (\mint,\ver)$ where 
\begin{itemize}
    \item $\mint(1^\lambda)$ outputs a pair $(s,\ket{\$^s})$ where $s$ is a string representing a serial number and $\ket{\$^s}$ is the banknote,
    and
    \item $\ver$ takes as input a serial number $s$ and an alleged banknote $\sigma$, and either accepts or rejects.
\end{itemize}
\end{definition}

The security is similar to that of full public-key quantum money setting:

\begin{longfbox}[breakable=false, padding=1em, margin-top=1em, margin-bottom=1em]
\begin{algo}[{\bf Quantum Money Mini-Scheme Counterfeit Security Game}]\label{prot:quantum_money_mini_scheme_security}
\end{algo}
\begin{enumerate}
    \item Run $(s, \ket{\$^s})\leftarrow \mint(1^{n})$ and send $(s, \ket{\$^s})$ to the adversary.
    \item Adversary returns two registers $\reg{AB}$ in some potentially entangled state $\sigma_{\reg{AB}}$.
    \item Run $\ver(s, \sigma_{\reg{A}})$ and $\ver(s, \sigma_{\reg{B}})$.  If either check rejects, then reject, otherwise accept.
\end{enumerate}
\end{longfbox}
\begin{definition}[Quantum money mini-scheme security]
    A quantum money mini-scheme scheme $\mathcal{S}$ satisfies \emph{$\epsilon$-quantum-money security} if for all efficient adversaries $A$, the success probability of $A$ in the counterfeit security game (\Cref{prot:quantum_money_mini_scheme_security}) is at most $\epsilon(n)$.
\end{definition}

Quantum lightning is a stronger security guarantee on quantum money in which \emph{not even the mint} can produce two banknotes for the same serial number~\cite{zhandry2021quantum}.

\begin{longfbox}[breakable=false, padding=1em, margin-top=1em, margin-bottom=1em]
\begin{algo}[{\bf Quantum Lightning Counterfeit Security Game}]\label{prot:quantum_lightning_security}
\end{algo}
\begin{enumerate}
    \item On input $1^n$, adversary returns a serial number $s$ and two registers $\reg{AB}$ in some potentially entangled state $\sigma_{\reg{AB}}$.
    \item Run $\ver(s, \sigma_{\reg{A}})$ and $\ver(s, \sigma_{\reg{B}})$.  If either check rejects, then reject, otherwise accept.
\end{enumerate}
\end{longfbox}
\begin{definition}[Quantum lightning security~\cite{zhandry2021quantum}]
    A quantum money mini-scheme scheme $\mathcal{S}$ satisfies \emph{$\epsilon$-quantum-lightning security} if for all efficient adversaries $A$, the success probability of $A$ in the counterfeit security game (\Cref{prot:quantum_lightning_security}) is at most $\epsilon(n)$.
\end{definition}

In each of the definitions, when $\epsilon$ is a negligible function in $n$, we say the scheme satisfies ``strong'' security.

\section{Duality Theorem}
In this section we present our main theorem, a computational duality between implementing a group representation and implmenting a \fse. We first present the exact case in \Cref{sec:duality_exact}. Then, in \Cref{sec:duality_approx}, we show how to generalize it to the case of approximate representations and \fse$\mskip-\thickmuskip$s.

\subsection{Exact Case}
\label{sec:duality_exact}

Now we prove an exact version of the duality theorem, where the implementation of the representation and the \fse are perfect.

\begin{theorem}[Quantum Duality for Representations of Groups, Exact Case]
\label{thm:duality_exact}
    Let $G$ be a finite group with an efficient quantum Fourier transform. Let $\rep: G \to U(\repspace)$ be a unitary representation of $G$, which decomposes into irreducible representations $\{(\lambda, W^{\lambda}_i)\}_{\lambda \in \widehat{G}, i \in [\multlambda]}$. 
    Then the following are equivalent:
    \begin{enumerate}
        \item There exists a quantum circuit, $C_{\rep}$, of size $s_{\rep}$, that implements the representation $\rep$. That is, 
        it implements the unitary
        \begin{align*}
            \ket{g} \otimes \ket{\psi} \mapsto \ket{g} \otimes \rep(g) \ket{\psi}
        \end{align*}
        for all $g \in G$ and all $\ket{\psi} \in \repspace$.
        \label{item:impl_rep}
        
        \item There exists a quantum circuit, $C_{\mathcal{M}}$, of size $s_{\mathcal{M}}$,
        that implements a \fse, $\mathcal{M}$, on the Fourier subspaces 
        $\{W^{\lambda}_i\}_{\lambda \in \widehat{G}, i \in [\multlambda]}$.
        Specifically, let each $W^{\lambda}_i$ have basis $\{\ket{\psi^{\lambda}_{i,j}}\}_{
        j \in [\dim(\varrho_{\lambda})]}$. Then $C_{\mathcal{M}}$ implements
        \begin{equation*}
            \mathcal{M}: \ket{\psi^{\lambda}_{i, j}} \ket{0} \mapsto \ket{\phi^{\lambda}_{i}} \ket{\lambda, j}\,,
        \end{equation*}
        for some orthonormal set of ``archetype'' states $\{\ket{\phi^{\lambda}_{i}}\}_{\lambda \in \widehat{G}, \, i \in [\multlambda]}$.%
        \footnote{Note that the form of the archetype states does not matter. The only requirement is that they are orthonormal so that $\mathcal{M}$ is an isometry.}

        \label{item:impl_meas}
    \end{enumerate}
    
    Going from \Cref{item:impl_rep} to \Cref{item:impl_meas}, we have that $s_{\mathcal{M}}$ 
    is
    $O(s_{\rep} + s_{\mathrm{QFT}})$, where $s_{\mathrm{QFT}}$ is the circuit complexity of implementing the quantum Fourier transform of $G$. In the other direction, we have that 
    $s_{\rep} = O(s_{\mathcal{M}} + s_{\mathrm{QFT}})$.
\end{theorem}

\begin{remark}
    In the special case in which the group is \emph{Abelian}, all the irreducible representations are 1-dimensional, so \Cref{item:impl_meas} above simplifies to a full projective measurement in the Fourier basis of the representation (the basis of states that are fixed by the representation). Moreover, the quantum Fourier transform for Abelian groups can always be implemented efficiently. Thus we get as a special case that for Abelian groups, the representation is directly dual to a Fourier measurement.
\end{remark}
\begin{remark}
    As an even more special case, the duality theorem of~\cite{aaronson2020hardness} is the case in which $G \cong \mathbb{Z}_2$.
\end{remark}
\begin{remark}
    This theorem is phrased in terms of explicit quantum circuits. One might wonder if it still applies in the black-box setting. And indeed, we do not make use of any non-relativizing properties of these circuits, except to assume that we can access the inverse of $C_{\mathcal{M}}$ (i.e. a Fourier injection) in order to uncompute it. (We similarly require inverse queries to the group's quantum Fourier transform.) Therefore, we can simply modify the statement of \Cref{item:impl_meas} to require access to both $C_{\mathcal{M}}$ and $C_{\mathcal{M}}^{\dagger}$, and then it would hold the same way in the black-box setting.
\end{remark}

\begin{proof}[Proof of \ref{thm:duality_exact}]~
We prove both implications separately.

\noindent\textbf{\ref{item:impl_rep} $\Rightarrow$ \ref{item:impl_meas}:}
Suppose that \Cref{item:impl_rep} is true. That is, we have a circuit of size $s$ that implements the representation $\rep$.
Let $\varrho_{\lambda}: G \to U(\repspace)$ be an irreducible representation of $G$ of dimension $\dimlambda$ and multiplicity $\multlambda$ in $\rep$, and let $W^{\lambda}_1, \dots, W^{\lambda}_{\multlambda}$ be the \manifestations
of $\varrho_{\lambda}$ in $\rep$.
For each subspace $W^{\lambda}_i$, take $\{\ket{\psi^{\lambda}_{i,j}}\}_{j \in [\dimlambda]}$ to be a basis for the subspace such that the corresponding  unitary $\varrho_{\lambda}(g)$ sends $\ket{\psi^{\lambda}_{i, j}}$ to $\sum_{k \in [\dimlambda]} \varrho_{\lambda}(g)_{k,j} \ket{\psi^{\lambda}_{i,k}}$.%
\footnote{Technically, any basis of $W^{\lambda}_i$ works fine, and we just need to unitarily transform the unitary $\varrho_{\lambda}(g)$ accordingly in our minds. However, it is convenient to consider a similar basis for all the subspaces $W^{\lambda}_i$ corresponding to $\lambda$, so that we can write $\varrho_{\lambda}(g)$ in terms of its matrix elements $\varrho_{\lambda}(g)_{i, j}$ in the same way across all of them.}

Suppose we have a basis state $\ket{\psi^{\lambda}_{i,j}}$ on which we want to perform \fse to produce $\ket{\phi^{\lambda}_i} \ket{\lambda, j}$ (for some set of ``archetype'' states $\ket{\phi^{\lambda}_i}$).%
\footnote{We consider only basis states $\ket{\psi^{\lambda}_{i,j}}$ without loss of generality because the general case follows from linearity.}
We begin by preparing the the uniform superposition over the group
$\frac{1}{\sqrt{|G|}}\sum_{g \in G} \ket{g}$ in an ancilla register and then, controlled on that register, apply the promised circuit $C_{\rep}$ for implementing $\rep$ to our state.
\begin{align*}
    &  \frac{1}{\sqrt{|G|}}\sum_{g \in G} \rep(g) \ket{\psi^{\lambda}_{i,j}} \otimes \ket{g}
    \\ &= 
    \frac{1}{\sqrt{|G|}}\sum_{g \in G} \sum_{k \in [\dimlambda]} \varrho_{\lambda}(g)_{k,j} \ket{\psi^{\lambda}_{i,k}} \otimes \ket{g}
    \\ &= 
    \sum_{k \in [\dimlambda]} \ket{\psi^{\lambda}_{i,k}} \otimes \frac{1}{\sqrt{|G|}}\sum_{g \in G} \varrho_{\lambda}(g)_{k,j} \ket{g}
\end{align*}
\noindent Inverting the group element in the last register gives
\begin{align*}
    &\rightarrow 
    \sum_{k \in [\dimlambda]} \ket{\psi^{\lambda}_{i, k}} \otimes \frac{1}{\sqrt{|G|}}\sum_{g \in G} \varrho_{\lambda}(g)_{k, j} \ket{g^{-1}}
    \\ 
    &= 
    \sum_{k \in [\dimlambda]} \ket{\psi^{\lambda}_{i,k}} \otimes \frac{1}{\sqrt{|G|}}\sum_{g \in G} \varrho_{\lambda}(g^{-1})_{k,j} \ket{g}
    \\ 
    &= 
    \frac{1}{\sqrt{\dimlambda}}\sum_{k \in [\dimlambda]} \ket{\psi^{\lambda}_{i,k}} \otimes \ket{\mathcal L^{\lambda}_{k,j}}
    \,,
\end{align*}
where $\ket{\mathcal L^{\lambda}_{kj}}$ is the $j$th basis vector of the $k$th \manifestation of the irrep $\lambda$ in the \emph{left regular} representation of $G$.
If we now perform a quantum Fourier transform on the second register, we get 
\begin{align*}
    \frac{1}{\sqrt{\dimlambda}}\sum_{k \in [\dimlambda]} \ket{\psi^{\lambda}_{i,k}} \otimes \ket{\lambda,k,j}
    \,.
\end{align*}

Reordering and regrouping the registers gives us 
\begin{align*}
    \left(\frac{1}{\sqrt{\dimlambda}}\sum_{k \in [\dimlambda]} \ket{\psi^{\lambda}_{i,k}} \otimes \ket{k}\right) \ket{\lambda, j}
    =
    \ket{\phi^{\lambda}_i} \ket{\lambda, j}
\end{align*}

If we wish, we can now measure the register containing $\lambda$ to get the label of the irreducible representation space containing our state.
Note that within subspace $W^{\lambda}_i$, this is a \se that extracts out $\ket{j}$\textemdash the state in the standard basis corresponding to whichever state inside $W^{\lambda}_i$ we started with%
\footnote{Note that the state that is extracted in the last register does not depend on which basis we chose for $W^{\lambda}_i$ before. In fact, our choice was only a mathematical choice and did not actually affect the computation in any way. What determined the basis we got at the output was really our choice of the vectors $\ket{\mathcal L^{\lambda}_{k, j}}$ for the left regular representation, and these are determined simply by which quantum Fourier transform we chose to implement. Interestingly, with a \fse, since \emph{all} the information about the original state within subspace $W^{\lambda}_i$ is extracted into a single register in the standard basis, we do not have to decide on a basis ahead of time! We can convert a \fse in one basis to one in another basis \emph{after the fact} by applying a unitary to the resulting extracted register.}
\textemdash and leaves behind the \arch state $\ket{\phi^{\lambda}_i} := \frac{1}{\sqrt{\dimlambda}}\sum_{k \in [\dimlambda]} \ket{\psi^{\lambda}_{i,k}} \otimes \ket{k}$. Interestingly, observe that in this case, the reduced state on the first register of the archetype state for subspace $W^{\lambda}_{i}$ is the fully mixed state on $W^{\lambda}_{i}$.
This is not necessarily the case, however, for a general \fse, which may have any form of archetype state (as long as the archetype states form an orthonormal basis, which ensures that the \fse is an isometry).

\paragraph{\ref{item:impl_meas} $\Rightarrow$ \ref{item:impl_rep}:}
Suppose that \Cref{item:impl_meas} is true. Then we have an circuit $C_{\mathcal{M}}$ implementing the \fse $\mathcal{M}$, which performs both a projective measurement $\{\Pi_{\lambda}\}_{\lambda \in \widehat{G}}$, where $\Pi_{\lambda}$ projects onto subspace $W^{\lambda}$, the (possibly empty) union of some set of subspaces $W^{\lambda}_1, \dots, W^{\lambda}_{\multlambda}$\textemdash where each $W^{\lambda}_i$ has dimension $\dimlambda$ and is spanned by some basis $\{\ket{\psi^{\lambda}_{i,j}}\}_{j \in [\dimlambda]}$\textemdash and a \se on each subspace $W^{\lambda}_i$:
\begin{align*}
    \mathcal M : \ket{\psi^{\lambda}_{i,j}} \mapsto \ket{\phi^{\lambda}_{i}}\ket{\lambda}\ket{j}
\end{align*}
for some orthonormal set of archetype states $\ket{\phi^{\lambda}_{i}}$.

We would like to perform $\rep$, the representation defined by each irrep $\lambda \in \widehat{G}$ corresponding to irrep subspaces $W^{\lambda}_1, \dots, W^{\lambda}_{\multlambda}$. 
We receive as input a state of the form $\ket{g} \ket{\psi}$, 
with
the first register containing a group element $g \in G$ for which we would like to implement its representation $\rep(g)$, 
and the second register containing a quantum state $\ket{\psi}$ on which we would like to perform the representation.

Write $\ket{\psi}$ in the basis of the $\ket{\psi^{\lambda}_{i,j}}$'s as 
$\ket{\psi} = \sum_{\lambda, i, j} \alpha^{\lambda}_{ij} \ket{\psi^{\lambda}_{i,j}}$.
We start by applying the promised circuit $C_{\mathcal{M}}$ for implementing $\mathcal M$ on $\ket{\psi}$ to get 
\begin{align*}
    \ket{g} 
    \otimes 
    \mathcal{M} 
    \ket{\psi} 
    = 
    \ket{g} 
    \otimes 
    \sum_{\substack{\lambda \in \widehat G\\ i \in [\multlambda]\\ j \in [\dimlambda]}} 
    \alpha^{\lambda}_{i,j} 
    \ket{\phi^{\lambda}_{i}} 
    \ket{\lambda}
    \ket{j}
    \,.
\end{align*}


Now, we claim that we can use two calls to the quantum Fourier transform of $G$ and a single group operation to implement the irrep $\irrep_{\lambda}$, but we will return to this later.
Assuming, for now, that we can do this, we apply it to this state
to perform $\varrho_{\lambda}(g)$ on the last register:
\begin{align*}
    &
    \ket{g}
    \otimes 
    \sum_{\lambda \in \widehat G, i \in [n_{\varrho}], j \in [\dimlambda]} 
    \alpha^{\lambda}_{i, j} 
    \ket{\phi^{\lambda}_{i}} 
    \ket{\lambda} 
    \otimes 
    \varrho_{\lambda}(g)
    \ket{j}
    \\ &=
    \ket{g}
    \otimes 
    \sum_{\lambda \in \widehat G, i \in [n_{\varrho}], j \in [\dimlambda]} 
    \alpha^{\lambda}_{i,j} 
    \ket{\phi^{\lambda}_{i}} 
    \sum_{k \in [d_\lambda]} 
    \ket{\lambda} 
    \otimes 
    \varrho_{\lambda}(g)_{k, j} 
    \ket{k}
    \\ &=
    \ket{g}
    \otimes 
    \sum_{\lambda \in \widehat G, i \in [\multlambda], j \in [\dimlambda]} 
    \alpha^{\lambda}_{i, j} 
    \sum_{k \in [d_\lambda]} 
    \varrho_{\lambda}(g)_{k, j} 
    \ket{\phi^{\lambda}_{i}} 
    \ket{\lambda} \ket{k}
    \,.
\end{align*}

We now use $C_{\mathcal{M}}^\dagger$ to un-compute the \fse on the last three registers, producing
\begin{align*}
    &
    \ket{g}
    \otimes 
    \sum_{\lambda \in \widehat G, i \in [\multlambda], j \in [\dimlambda]} 
    \alpha^{\lambda}_{i, j} \sum_{k \in [d_\lambda]} \varrho_{\lambda}(g)_{k, j} \ket{\psi^{\lambda}_{i, k}}
    \\ &=
    \ket{g}
    \otimes 
    \sum_{\lambda \in \widehat G, i \in [\multlambda], j \in [\dimlambda]} 
    \alpha^{\lambda}_{i, j} \;\; \rep(g) \ket{\psi^{\lambda}_{i, j}}
    \\ &=
    \ket{g}
    \otimes 
    \rep(g)
    \sum_{\lambda \in \widehat G, i \in [\multlambda], j \in [\dimlambda]} 
    \alpha^{\lambda}_{i, j} \ket{\psi^{\lambda}_{i, j}}
    \\ &=
    \ket{g}
    \otimes 
    \rep(g) \ket{\psi}
    \,.
\end{align*}

We can see that this successfully implements the representation $\rep(g)$.

It remains to show how to implement the irreps of $G$. It was observed by~\cite{jordan2008fast} that this is possible using just the group's quantum Fourier transform, and we present it here for completeness.
Given a state of the form $\ket{g} \ket{\lambda} \ket{j}$, where $g \in G$ is the group element for which we would like the perform the irrep, $\lambda \in \widehat{G}$ is the label of the irrep we would like to perform, and $j \in [\dimlambda]$ is a computational basis state within a Hilbert space of appropriate dimension, we would like to map it to 
$
    \ket{g}
    \otimes 
    \ket{\lambda} 
    \otimes 
    \sum_{k \in [d_\lambda]} 
    \varrho_{\lambda}(g)_{k, j} 
    \ket{k}
$.

We start by adding a multiplicity label, $i \in [\dimlambda]$. Technically any value will do, since it will have no effect and will be returned untouched, but since we will need to ensure that it fits within the size of of the corresponding irrep, we can always just take the first \manifestation, that is, $i = 0$. We then have a state of the form $\ket{g} \ket{\lambda} \ket{0} \ket{j}$, onto which we apply the inverse quantum Fourier transform of $G$, to get 
\begin{align*}
    \ket{g} 
    \otimes
    \mathrm{QFT}_{G}^{\dag} 
    \ket{\lambda} 
    \ket{0} 
    \ket{j}
    &
    =
    \ket{g} 
    \ket{\mathcal{L}^{\lambda}_{0,j}}
    \\
    &
    =
    \ket{g} 
    \sqrt{\frac{\dimlambda}{|G|}} 
    \sum_{h \in G} 
    \varrho_{\lambda}(h^{-1})_{0, j} 
    \ket{h}
    \,.
\end{align*}
Applying the group operation of $g$ onto $h$ from the left gives 
\begin{align*}
    &
    \rightarrow 
    \ket{g} 
    \sqrt{\frac{\dimlambda}{|G|}} 
    \sum_{h \in G} 
    \varrho_{\lambda}(h^{-1})_{0, j} 
    \ket{g h}
    \\
    &
    =
    \ket{g} 
    \sqrt{\frac{\dimlambda}{|G|}} 
    \sum_{h \in G} 
    \varrho_{\lambda}(h^{-1} g)_{0, j} 
    \ket{h}
    \\
    &
    =
    \ket{g} 
    \sqrt{\frac{\dimlambda}{|G|}} 
    \sum_{h \in G} 
    \sum_{k \in [\dimlambda]} 
    \varrho_{\lambda}(h^{-1})_{0, k} 
    \;
    \varrho_{\lambda}(g)_{k, j} 
    \ket{h}
    \\
    &
    =
    \ket{g} 
    \sum_{k \in [\dimlambda]} 
    \varrho_{\lambda}(g)_{k, j} 
    \sqrt{\frac{\dimlambda}{|G|}} 
    \sum_{h \in G} 
    \varrho_{\lambda}(h^{-1})_{0, k} 
    \ket{h}
    \\
    &
    =
    \ket{g} 
    \sum_{k \in [\dimlambda]} 
    \varrho_{\lambda}(g)_{k, j} 
    \ket{\mathcal{L}^{\lambda}_{0,k}}
    \,.
\end{align*}
Applying once again the quantum Fourier transform of $G$ produces
$
    \ket{g} 
    \ket{\lambda}
    \ket{0} 
    \sum_{k \in [\dimlambda]} 
    \varrho_{\lambda}(g)_{k, j} 
    \ket{k}
$,
from which we can discard the ancilla to recover the desired state.
\end{proof}

From \Cref{thm:duality_exact}, we get the following interesting corollary, which may be of independent interest. It allows us to implement a representation of a group $G$ by using a circuit for implementing a \emph{different} representation of the same group $G$, as long as the representation we want shares irrep subspaces with the representation we have.

\begin{corollary}
    Given an efficient implementation of a representation $\rep: \G \mapsto U(\repspace)$ that breaks $\repspace$ into some set of irreducible subspaces $\{W^{\lambda}_i\}_{\lambda \in \widehat{G}, i \in [\multlambda]}$, 
    as well an efficiently computable function $r : \widehat{G} \to \widehat{G}$ mapping irrep labels appearing in $\rep$ to other irreps of $G$ having the same or smaller size, 
    we can implement the new representation of the same group that acts on the same subspaces, but with each $W^{\lambda}_i$ acted on by irrep $r(\lambda)$ instead of $\lambda$.
    
    In other words, given the implementation of one representation, we can implement many different other representations just by being able to compute the new desired irrep labels.
\end{corollary}
\begin{proof}[Main Idea]
    This results from a double application of \Cref{thm:duality_exact}. Once in the forward direction on the first representation to get a \fse on the subspaces, and then once in the backwards direction to get an implementation of the second representation.
\end{proof}

\subsection{Approximate Case}
\label{sec:duality_approx}
Here we present an approximate duality theorem, in that the conditions of \Cref{item:impl_rep} and \Cref{item:impl_meas} have to hold approximately.    
In order to prove an approximate version of the duality theorem, we will need the following theorem from \cite{gowers2016inverse} regarding approximate representations.  

\begin{theorem}[Gowers-Hatami~\cite{gowers2016inverse}]\label{thm:gowers-hatami}
    Let $G$ be a finite group, $\epsilon \geq 0$ and $\rep: G \mapsto U(\repspace)$ be an $\epsilon$-approximate representation of $G$.  Then there exists a register $\otherrepspace$ of dimension $d' = (1 + O(\epsilon)) \dim(\repspace)$, an isometry $V: \repspace \mapsto \otherrepspace$ and an exact representation $\otherrep: G \mapsto U(\otherrepspace)$ such that
    \begin{equation*}
        \avg_{x \in G} \left\|\rep(x) - V^{\dagger} \otherrep(x) V\right\|^2\leq 2\epsilon\,.
    \end{equation*}
    Where the norm $\|\cdot \|$ is implied by the dimension-normalized Hilbert-Schmidt inner product $\braket{A, B} = \tr[AB^{\dagger}] / \dim(\repspace)$.
\end{theorem}

\begin{remark}
    While the matter of approximate representations has been extensively studied in mathematics and quantum computer science, the idea of an approximate measurement into irreducible representations has not been studied as much.  In particular, the idea of weak (or strong) Fourier sampling is typically used in algorithms for solving problems in groups.  For these kinds of problems, there is a well defined measurement that one can try to approximate.  However in our case, as works from representation theory note~\cite{gowers2016inverse,Kazhdan1982On}, there may be vector spaces that admit approximate representations, but for which no exact representation exists.  This raises the question of what a measurement into an invariant subspace \emph{should} look like.  \cite{gowers2016inverse} proposes a lemma pertaining to ``approximately invariant subspaces'', but it uses a notion of Fourier transform that is different from the quantum Fourier transform that we often consider.
    Here we propose a notion of approximate measurement onto an invariant subspace inspired by the result of \cite{gowers2016inverse}, and use it in our duality result.
\end{remark}

Consider the following approximate versions of \cref{thm:duality_exact}.

\begin{theorem}[Approximate duality, forward direction]
\label{thm:approx_duality_1}
    Let $G$ be a finite group with a Fourier transform that can be implemented with a circuit of size $s_{\mathrm{QFT}}$.  Let $\rep: G \mapsto U(\repspace)$ be an $\epsilon$-approximate representation of $G$ with a circuit implementation of size $s_{\rep}$, with $\repspace$ being an $n$-qubit register.  Then there exists a register $\otherrepspace$, an exact group representation $\otherrep: G \mapsto U(\otherrepspace)$
    and an isometry $V$ mapping $\repspace$ to $\otherrepspace$ such that for the Fourier decomposition $\otherrepspace = \bigoplus_{\lambda\in\widehat{\G}, i \in [\multlambda]} W^{\lambda}_i$ and basis $\{\ket{\psi_{i,j}^{\lambda}}\}_{\lambda\in\widehat{\G}, i \in [\multlambda],\, j \in [\dimlambda]}$ of $\otherrep$ as in \cref{item:impl_meas} from \cref{thm:duality_exact}, there is a circuit $C$ of size $O(s_{\rep} + s_{\mathrm{QFT}})$, and a set of archetype states $\{\ket{\phi_{i}^{\lambda}}\}_{\lambda\in\widehat{\G}, i \in [\multlambda]}$, such that
    \begin{equation*}
        \frac{1}{\dim(\otherrepspace)}\sum_{\substack{\lambda\in \widehat{G}, i \in [n_\lambda]\\ j \in [\dimlambda]}} \Re \bra{\phi_{i}^{\lambda}} \otimes \bra{\lambda, j} V C V^{\dagger} \ket{\psi^{\lambda}_{i, j}} \geq 1 - \epsilon\,.
    \end{equation*}
    That is, we get an $\epsilon$-approximate \fse.
\end{theorem}
\begin{proof}
    We let $C$ be the same circuit as in the exact case, first preparing the state $\frac{1}{\sqrt{|G|}} \sum_{g \in G} \ket{g}$, and applying the representation $\rep(g)$ controlled on that register to the state. We also let the archetype states be $\ket{\phi_{i}^{\lambda}} = \frac{1}{\sqrt{\dimlambda}} \sum_{k \in [\dimlambda]} \ket{\psi^{\lambda}_{i, k}} \otimes \ket{k}$ as in the exact case, where $\dimlambda := \dim(\varrho_{\lambda})$.  We can compute the quantity from the theorem statement as:
    
    \begin{align*}
        &
            \frac{1}{\dim(\otherrepspace)}
            \sum_{\lambda, i, j} 
            \Re 
            \bra{\phi_{i}^{\lambda}} 
            \bra{\lambda, j} 
            V C V^{\dagger} 
            \ket{\psi^{\lambda}_{i, j}} 
        \\ 
        &= 
            \frac{1}{\dim(\otherrepspace)}
            \sum_{\lambda, i, j} 
            \Re 
            \bra{\phi_{i}^{\lambda}} 
            \bra{\lambda, j} 
            (V \otimes \id) 
            (\id \otimes \mathrm{QFT}) 
            \cdot
            \sum_{g \in G}\rep(g) \otimes \proj{g^{-1}} 
            \cdot
            (\id \otimes \mathrm{QFT}^{\dagger}) 
            (V \otimes \id)^{\dagger} 
            \ket{\psi^{\lambda}_{i, j}} 
            \ket{0}
        \\ 
        &= 
            \frac{1}{\dim(\otherrepspace)}
            \sum_{\lambda, i, j} 
            \Re 
            \frac{1}{\sqrt{\dimlambda}} 
            \sum_{k}
            \bra{\psi_{i, k}^{\lambda}} 
            \bra{\mathcal{L}^{\lambda}_{k, j}} 
            (V \otimes \id) 
            \cdot
            \sum_{g \in G}\rep(g) \otimes \proj{g^{-1}} 
            \cdot
            (\id \otimes \mathrm{QFT}^{\dagger}) 
            (V \otimes \id)^{\dagger} 
            \ket{\psi^{\lambda}_{i, j}} 
            \ket{0}
        \\ 
        &= 
            \frac{1}{\dim(\otherrepspace)}
            \sum_{\lambda, i, j} 
            \Re 
            \frac{1}{\sqrt{\dimlambda}} 
            \sum_{k}
            \bra{\psi_{i, k}^{\lambda}} 
            \bra{\mathcal{L}^{\lambda}_{k, j}} 
            (V \otimes \id) 
            \cdot 
            \sum_{g \in G}\rep(g) \otimes \proj{g^{-1}} 
            \cdot 
            (V \otimes \id)^{\dagger} 
            \ket{\psi^{\lambda}_{i, j}} 
            \frac{1}{\sqrt{\abs{G}}}
            \sum_{g' \in G}
            \ket{g'}
        \\
        &= 
            \frac{1}{\dim(\otherrepspace)}
            \sum_{\lambda, i, j} 
            \Re 
            \frac{1}{\sqrt{\dimlambda}} 
            \sum_{k}
            \bra{\psi_{i, k}^{\lambda}} 
            \bra{\mathcal{L}^{\lambda}_{k, j}} 
            (V \otimes \id) 
            \cdot 
            \left(
                \frac{1}{\sqrt{\abs{G}}}
                \sum_{g \in G}
                \rep(g) 
                \otimes \id
            \right)
            \cdot 
            (V \otimes \id)^{\dagger} 
            \ket{\psi^{\lambda}_{i, j}} 
            \ket{g^{-1}}
        \\
        &= 
            \frac{1}{\dim(\otherrepspace)}
            \sum_{\lambda, i, j} 
            \Re 
            \frac{1}{\sqrt{\abs{G}}}
            \sum_{g \in G}
            \frac{1}{\sqrt{\dimlambda}} 
            \sum_{k}
            \bra{\psi_{i, k}^{\lambda}} 
            \bra{\mathcal{L}^{\lambda}_{k, j}} 
            (V \otimes \id) 
            \cdot 
            \left(
                \rep(g) 
                \otimes \id
            \right)
            \cdot 
            (V \otimes \id)^{\dagger} 
            \ket{\psi^{\lambda}_{i, j}} 
            \ket{g^{-1}}
        \\
        &= 
            \frac{1}{\dim(\otherrepspace)}
            \sum_{\lambda, i, j} 
            \Re 
            \frac{1}{\sqrt{\abs{G}}}
            \sum_{g \in G}
            \frac{1}{\sqrt{\dimlambda}} 
            \sum_{k}
            \bra{\psi_{i, k}^{\lambda}} 
            \bra{\mathcal{L}^{\lambda}_{k, j}} 
            (
                V 
                \rep(g) 
                V^{\dagger}
                \otimes 
                \id
            ) 
            \ket{\psi^{\lambda}_{i, j}} 
            \ket{g^{-1}}
        \\
        &= 
            \frac{1}{\dim(\otherrepspace)}
            \sum_{\lambda, i, j} 
            \Re 
            \frac{1}{\sqrt{\abs{G}}}
            \sum_{g \in G}
            \frac{1}{\sqrt{\dimlambda}} 
            \sum_{k}
            \bra{\psi_{i, k}^{\lambda}} 
            V \rep(g) V^{\dagger} 
            \ket{\psi^{\lambda}_{i, j}} 
            \braket{
                \mathcal{L}^{\lambda}_{k, j}
                |
                g^{-1}
            }
        \\
        &= 
            \frac{1}{\dim(\otherrepspace)}
            \sum_{\lambda, i, j} 
            \Re 
            \frac{1}{\sqrt{\abs{G}}}
            \sum_{g \in G}
            \frac{1}{\sqrt{\dimlambda}} 
            \sum_{k}
            \bra{\psi_{i, k}^{\lambda}} 
            V \rep(g) V^{\dagger} 
            \ket{\psi^{\lambda}_{i, j}} 
            \sqrt{
                \frac
                {\dimlambda}
                {\abs{G}}
            }
            \sum_{h \in \G}
            \varrho(h^{-1})^*_{k, j}
            \braket{
                h
                |
                g^{-1}
            }
            \addtag\label{eq:approx_forw_Lrkj}
        \\
        &= 
            \frac{1}{\dim(\otherrepspace)}
            \sum_{\lambda, i, j} 
            \Re 
            \frac{1}{\abs{G}}
            \sum_{g \in \G}
            \sum_{k}
            \varrho_{\lambda}(g)^*_{k, j}
            \bra{\psi_{i, k}^{\lambda}} 
            V \rep(g) V^{\dagger} 
            \ket{\psi^{\lambda}_{i, j}} 
        \\
        &= 
            \frac{1}{\dim(\otherrepspace)}
            \sum_{\lambda, i, j} 
            \Re 
            \frac{1}{\abs{G}}
            \sum_{g \in \G}
            \bra{\psi_{i, j}^{\lambda}} 
            \otherrep(g)^{\dagger} 
            V 
            \rep(g) 
            V^{\dagger} 
            \ket{\psi^{\lambda}_{i, j}} 
            \addtag\label{eq:approx_forw_mcG}
        \\
        &= 
            \avg_{g \in \G}
            \frac{1}{\dim(\otherrepspace)}
            \sum_{\lambda, i, j} 
            \Re 
            \bra{\psi_{i, j}^{\lambda}} 
            \otherrep(g)^{\dagger} 
            V 
            \rep(g) 
            V^{\dagger} 
            \ket{\psi^{\lambda}_{i, j}} 
        \\
        &= 
            \avg_{g \in \G}
            \frac{1}{\dim(\otherrepspace)}
            \Re 
            \;
            \tr
            \left[
                \otherrep(g)^{\dagger} 
                V 
                \rep(g) 
                V^{\dagger} 
            \right]
            \addtag\label{eq:approx_forw_trace}
        \\
        &= 
            \avg_{g \in \G}
            \Re 
            \;
            \left\langle
                \otherrep(g)
                ,
                V^{\dagger} 
                \rep(g) 
                V 
            \right\rangle
            \addtag\label{eq:approx_forw_HSinner}
        \\
        &= 
            1 -
            \frac{1}{2}
            \avg_{g \in \G}
            \left\|
                \rep(g)
                -
                V^{\dagger} 
                \otherrep(g) 
                V 
            \right\|^2
            \addtag\label{eq:approx_forw_diff}
        \\
        &\geq 
            1 - \epsilon
        \,.
    \end{align*}

    \needspace{4\baselineskip}
    Here the first line is expanding out the definition of the circuit as a quantum Fourier transform, controlled $\rep$,%
    \footnote{
        Technically, we control on $g^{-1}$, but this is just so that we can use the left-regular Fourier transform, rather than the right-regular one. This is not essential, but it slightly simplifies the notation.
    }
    and then an inverse quantum Fourier transform.  
    (There is also rearrangement of registers, but this is implicit in order to simplify notation.)
    The second and third lines applies the inverse Fourier transform to the $\ket{0}$ state, which represents the trivial irrep of $\G$, as well as to the $\bra{\lambda, k, j}$ (commuting it past the $V$, which acts only on the first register). The line labeled~\ref{eq:approx_forw_Lrkj} expands the definition of $\ket{\mathcal{L}^{\lambda}_{k,j}}$, and line~\ref{eq:approx_forw_mcG} uses the fact that $\otherrep$ exactly performs the representation on the basis of the states $\ket{\psi_{i,j}^{\lambda}}$.
    Line~\ref{eq:approx_forw_trace} uses the fact that the states $\ket{\psi_{i,j}^{\lambda}}$ form a complete basis for $\otherrepspace$.
    Line~\ref{eq:approx_forw_HSinner} uses the definition of the Hilbert Schmidt inner product, 
    line~\ref{eq:approx_forw_diff} uses the fact that $\|A - B\| = \sqrt{2 - 2\, \Re \braket{A, B}}$,
    and the last line uses the bound from \Cref{thm:gowers-hatami}.
\end{proof}

We note that this part of the duality theorem preserves the error between the representation and the measurement.  

\begin{remark}
    The forward direction could equivalently be phrased as follows: Let $\rep$ be $2\epsilon$-close to an exact representation $\otherrep$ under isometry $V$, then there is an implementation of $\epsilon$-approximate \fse up to $V$ with a circuit whose size of $O(s_{\rep} + s_{\mathrm{QFT}})$.
\end{remark}

We can also show the other direction, albeit with (we believe) sub-optimal error scaling.


\begin{theorem}[Approximate duality, reverse direction]
\label{thm:approx_duality_2}
    Let $G$ be a finite group with a Fourier transform that can be implemented with a circuit of size $s_{\mathrm{QFT}}$.
    Let $\repspace$ and $\otherrepspace$ be two registers with an isometry $V$ mapping $\repspace$ to $\otherrepspace$, and let $\otherrep$ be an exact representation on $\otherrepspace$. 
    Say that we have a circuit $C_{\mathcal{M}}$ of size $s_{\mathcal{M}}$ which implements an $\epsilon$-approximate \fse in $\repspace$, satisfying
    \begin{equation*}
        \frac{1}{\dim(\otherrepspace)} \sum_{\lambda, i, j} \Re \bra{\phi_i^{\lambda}} \otimes \bra{\lambda, j} (V \otimes \id) \mathcal{M} (V^{\dagger} \otimes \id) \ket{\psi_{i, j}^{\lambda}} \otimes \ket{0} \geq 1 - \epsilon\,.
    \end{equation*}
    Then there exists a circuit of size $O(s_{\mathcal{M}} + s_{\mathrm{QFT}})$ which implements a map $\rep$ of $\G$ on $\repspace$, that is $2\epsilon$-close to $\otherrep$, i.e. one satisfying 
    \begin{equation*}
        \mathop{\mathbb{E}}_{g \in G} \left\|V \rep(g) V^{\dagger} - \otherrep(g) \right\|^2 \leq 2\epsilon\,.
    \end{equation*}
\end{theorem}

\begin{proof}
    The implementation of $\rep$ will be identical to the one from \Cref{thm:duality_exact}.  In particular, $\rep(g)$ will first apply $\mathcal{M}$ to measure $\lambda$ and extract $j$, then apply $\varrho_{\lambda}(g)$ to the register containing $j$, and finally it will un-compute $\mathcal{M}$.  

    We can proceed by evaluating the average difference between $V\rep(g)V^{\dagger}$ and $\otherrep(g)$ under the Hilbert-Schmidt norm.
    \begin{align*}
        &\avg_{g \in G} \left\|V \rep(g) V^{\dagger} - \otherrep(g)\right\|^2\\
        &\hspace{10mm}= \avg_{g \in G} \braket{V \rep(g) V^{\dagger} - \otherrep(g), V \rep(g) V^{\dagger} - \otherrep(g)}\\
        &\hspace{10mm}= \avg_{g \in G}\frac{1}{\dim(\otherrepspace)}\tr\left[V\rep(g) V^{\dagger} V\rep(g)^{\dagger} V^{\dagger} + \otherrep(g)\otherrep(g)^{\dagger} - V \rep (g)^{\dagger}  V^{\dagger} \otherrep^{\dagger}(g) - \otherrep(g) V \rep(g) V^{\dagger}\right] \\
        &\hspace{10mm}= 2 - \avg_{g \in G} \frac{1}{\dim(\otherrepspace)} \tr\left[V \rep(g) V^{\dagger} \otherrep^{\dagger} (g) + \otherrep(g) V \rep(g)^{\dagger}  V^{\dagger}\right]\numberthis \label{eq:bound_on_closeness}
    \end{align*}
    Here we note that the implementation of $\rep(g)$ is always unitary, and $V^{\dagger} V = \id$, so the first two terms are the the identity on $\otherrepspace$.  
    Now we lower bound the second term.  We begin by writing it as two times the real component of a trace, and expand the definitions of $\otherrep$ and $\rep$.  

    \begin{align*}
        \avg_{g \in G} \Re \frac{2}{\dim(\otherrepspace)} \tr\left[V \rep(g) V^{\dagger} \otherrep^{\dagger} (g)\right] &= \avg_{g \in G} \frac{2}{\dim(\otherrepspace)} \Re \sum_{\lambda, i, j} \bra{\psi^{\lambda}_{i, j}} V \rep(g) V^{\dagger} \otherrep(g)\ket{\psi^{\lambda}_{i, j}}\\
        &= \avg_{g \in G} \frac{2}{\dim(\otherrepspace)} \Re \sum_{\lambda, i, j, k} \varrho_{\lambda}(g)^{\dagger}_{k, j} \bra{\psi^{\lambda}_{i, j}} V \rep(g) V^{\dagger}\ket{\psi^{\lambda}_{i, k}}\,.
    \end{align*}
    Now, we expand out the definition of $\rep$.  This yields the following state.
    \begin{align*}
        &\avg_{g \in G} \frac{2}{\dim(\otherrepspace)} \Re \sum_{\lambda, i, j, k} \varrho_{\lambda}(g)^{\dagger}_{k, j} \bra{\psi^{\lambda}_{i, j}} V \mathcal{M}^{\dagger}\varrho_{\lambda}(g) \mathcal{M} V^{\dagger}\ket{\psi^{\lambda}_{i, k}} \\
        &\hspace{5mm}= \avg_{g \in G} \frac{2}{\dim(\otherrepspace)} \Re \sum_{\substack{\lambda, i, j, k\\ \lambda', a, b}} \varrho_{\lambda}(g)^{\dagger}_{k, j}(\bra{\phi^{\lambda'}_{a}} \otimes \bra{\lambda', b}) (V \otimes \id) \mathcal{M} V^{\dagger}\ket{\psi^{\lambda}_{i, k}} \bra{\psi^{\lambda}_{i, j}} V \mathcal{M}^{\dagger} \varrho_{\lambda}(g) (V^{\dagger} \otimes \id) \ket{\phi^{\lambda'}_{a}} \ket{\lambda', b}\\
        &\hspace{5mm}= \avg_{g \in G} \frac{2}{\dim(\otherrepspace)} \Re \sum_{\substack{\lambda, i, j, k\\ \lambda', a, b, c}} \varrho_{\lambda}(g)^{\dagger}_{k, j} \varrho_{\lambda'}(g)_{c, b}(\bra{\phi^{\lambda'}_{a}} \otimes \bra{\lambda', b}) (V \otimes \id) \mathcal{M} V^{\dagger}\ket{\psi^{\lambda}_{i, k}} \bra{\psi^{\lambda}_{i, j}} V \mathcal{M}^{\dagger} (V^{\dagger} \otimes \id) \ket{\phi^{\lambda'}_{a}} \ket{\lambda', c}\\
        &\hspace{5mm}= \avg_{g \in G} \frac{2}{\dim(\otherrepspace)} \Re \sum_{\substack{\lambda, i, j, k\\ \lambda', a, b, c}} \varrho_{\lambda}(g)^{\dagger}_{k, j} \varrho_{\lambda'}(g)_{c, b}(\bra{\phi^{\lambda'}_{a}} \otimes \bra{\lambda', b}) (V \otimes \id) \mathcal{M} V^{\dagger}\ket{\psi^{\lambda}_{i, k}} \bra{\psi^{\lambda}_{i, j}} V \mathcal{M}^{\dagger} (V^{\dagger} \otimes \id) \ket{\phi^{\lambda'}_{a}} \ket{\lambda', c}\\
        &\hspace{5mm}= \frac{2}{\dim(\otherrepspace)} \Re \sum_{\substack{\lambda, i, j, k\\ \lambda', a, b, c}}\avg_{g \in G} \left[\varrho_{\lambda}(g)^{\dagger}_{k, j} \varrho_{\lambda'}(g)_{c, b}\right](\bra{\phi^{\lambda'}_{a}} \otimes \bra{\lambda', b}) (V \otimes \id) \mathcal{M} V^{\dagger}\ket{\psi^{\lambda}_{i, k}} \bra{\psi^{\lambda}_{i, j}} V \mathcal{M}^{\dagger} (V^{\dagger} \otimes \id) \ket{\phi^{\lambda'}_{a}} \ket{\lambda', c}\\
        &\hspace{5mm}= \frac{2}{\dim(\otherrepspace)} \Re \sum_{\lambda, a, i, j, k} \frac{1}{\dimlambda}(\bra{\phi^{\lambda'}_{a}} \otimes \bra{\lambda, j}) (V \otimes \id) \mathcal{M} V^{\dagger}\ket{\psi^{\lambda}_{i, k}} \bra{\psi^{\lambda}_{i, j}} V \mathcal{M}^{\dagger} (V^{\dagger} \otimes \id) \ket{\phi^{\lambda}_{a}} \ket{\lambda, k}\\
        &\hspace{5mm}= \Re\frac{2}{\dim(\otherrepspace)}\tr\left[\sum_{\lambda, i, j, k} \frac{1}{\dimlambda}\id \otimes \ket{\lambda, j}\!\!\bra{\lambda, k} (V \otimes \id) \mathcal{M} V^{\dagger} \ket{\psi^{\lambda}_{i, k}}\!\!\bra{\psi^{\lambda}_{i, j}} V \mathcal{M}^{\dagger}(V^{\dagger} \otimes \id)\right]\\
        &\hspace{5mm}= \Re\frac{2}{\dim(\otherrepspace)}\tr\left[\sum_{\lambda, i, j, k}  \frac{1}{\dimlambda}\id \otimes \ket{\lambda, j}\!\!\bra{\lambda, k} \mathcal{M} V^{\dagger} \ket{\psi^{\lambda}_{i, k}}\!\!\bra{\psi^{\lambda}_{i, j}} V \mathcal{M}^{\dagger}\right]\\
        &\hspace{5mm}=\Re\frac{2}{\dim(\otherrepspace)}\sum_{\lambda, i, j, k} \frac{1}{\dimlambda}\bra{\phi^{\lambda}_{i}} \otimes \bra{\lambda, k} (V \otimes \id)\mathcal{M} V^{\dagger} \ket{\psi^{\lambda}_{i, k}}\!\!\bra{\psi^{\lambda}_{i, j}} V \mathcal{M}^{\dagger} (V^{\dagger} \otimes \id)\ket{\phi^{\lambda}_{i} }\otimes \ket{\lambda, j}\\
        &\hspace{5mm}\geq \Re\frac{2}{\dim(\otherrepspace)}\sum_{\lambda, i} \frac{1}{\dimlambda}\left(\sum_{j}\bra{\phi^{\lambda}_{i}} \otimes \bra{\lambda, j} (V \otimes \id)\mathcal{M} V^{\dagger} \ket{\psi^{\lambda}_{i, j}}\right)^2\\
        &\hspace{5mm}\geq \Re\frac{2}{\dim(\otherrepspace)}\sum_{\lambda, i, j} \bra{\phi^{\lambda}_{i}} \otimes \bra{\lambda, j} (V \otimes \id)\mathcal{M} V^{\dagger} \ket{\psi^{\lambda}_{i, j}}\\
        &\hspace{5mm}\geq 2(1 - \epsilon)\\
        &\hspace{5mm}\geq 2 - 2\epsilon\,.
    \end{align*}
    Here, we insert identity matrices between $\varrho_{\lambda}$ and $\mathcal{M}$, and we use the definition of the inner product.  Then, we use the Schur orthogonality relations to cancel the terms where $\varrho_{\lambda} \neq \varrho_{\lambda'}$ or $(k, j) \neq (c, b)$.
    Then we use the definition of the trace, and the cyclic property.  Finally, since $(V \otimes \id)$ commutes with $\id \otimes \ket{\lambda, j}\!\!\bra{\lambda, k}$, we can move it to the other side using the cyclic property.  Then we use the fact that $\proj{\phi^{\lambda}_{i}} \otimes \ket{\lambda, j}\!\!\bra{\lambda, k} \preceq \id \otimes \ket{\lambda, j}\!\!\bra{\lambda, k}$, together with the cyclic property of the trace.  Finally, we apply Cauchy-Schwarz twice on the sum over $j$ and $k$, and the assumption about the performance of $\mathcal{M}$ on an average state from $V^{\dagger} \otherrepspace$.

    Plugging this back into \Cref{eq:bound_on_closeness}, we get the following upper bound on the average distance between $\otherrep$ and $\rep$:
    \begin{align*}
        \avg_{g \in G} \left\|V \rep(g) V^{\dagger} - \otherrep(g)\right\|^2 &\leq 2 - \avg_{g \in G} \frac{1}{\dim(\otherrepspace)} \tr\left[V \rep(g) V^{\dagger} \otherrep^{\dagger} (g) + \otherrep(g) V \rep(g)^{\dagger}  V^{\dagger}\right]\\
        &\leq 2 - (2 - 2\epsilon)\\
        &\leq 2\epsilon
        \,,
    \end{align*}
    as desired.
\end{proof}

We note that while in the forward direction (\cref{thm:approx_duality_1}), our duality theorem preserves the inner product error from the approximate representation, we are not able to prove a perfectly tight approximate duality because the reverse direction (\cref{thm:approx_duality_2}) yields a different notion of approximate representation, i.e. being close (up to an isometry) to a real representation.  Applying the definition of $\epsilon$-approximate representation directly would not yield the same $\epsilon$ as we started with in the reverse direction.  Note that if we had defined the forward direction in the same way, using the result of \cite{gowers2016inverse}, we would get a perfect duality, but the notion of approximate representation from \cref{def:eps-approx-rep} is more widely used.  We leave it as an open question whether an $\epsilon$-approximate representation can be recovered in the reverse direction.

\paragraph{Comparison with \cite{aaronson2020hardness}.}

We comment on how our approximate duality (\Cref{thm:approx_duality_1,thm:approx_duality_2}) relates to the approximate duality theorem from \cite[Theorem $2$]{aaronson2020hardness}.  Let $\ket{x}$ and $\ket{y}$ be two orthogonal quantum states and $U$ be a unitary such that
\begin{align*}
    \bra{y} U \ket{x} &= a\\ 
    \bra{x} U \ket{y} &= b\,.
\end{align*}
Unlike in the general case of \Cref{thm:approx_duality_1}, in this case, the fact that $\ket{x}$ and $\ket{y}$ are orthogonal implies that there exists a unitary $\widehat{U}$ in the \emph{same} register such that $\widehat{U}$ exactly swaps $\ket{x}$ and $\ket{y}$.  
As a representation of $\mathbb{Z}_2$, we thus have the efficient $\epsilon$-close representation $\rep: g \mapsto U^{g}$ and an exact representation $\otherrep: g \mapsto \widehat{U}^g$.  We then have the following:
\begin{align*}
    \avg_{g \in \mathbb{Z}_2} \left\|U^{g} - \widehat{U}^{g}\right\|^2 
    &= 
        \frac{1}{2} \left(\left\|\id - \id\right\|^2 + \left\|U - \widehat{U}\right\|^2\right)
    \\
    &= 
        \frac{1}{2} \langle U - \widehat{U}, U - \widehat{U}\rangle
    \\
    &= 
        \frac{1}{4} \tr[(U - \widehat{U}) (U - \widehat{U})^{\dagger}]
    \\
    &= 
        \frac{1}{4} \left(4 - \tr[U\widehat{U}^{\dagger}] - \tr[\widehat{U}U^{\dagger}]\right)
    \\
    &= 
        1 - \frac{1}{4} \Re\left(2a + 2b\right)
    \\
    &= 
        1 - \frac{\Re\left(a + b\right)}{2} 
    =: 2\epsilon
    \,.
\end{align*}
Here we use the fact that $\tr[U\widehat{U}^{\dagger}] = \bra{x} U \ket{y} + \bra{y} U \ket{x} = a + b$ since $\widehat{U}$ exactly swaps $\ket{x}$ and $\ket{y}$, and similarly for $\tr[\widehat{U}U^{\dagger}]$. Let $\mathcal{M}$ be the measurement implied by the forward direction of the approximate duality (the \fse simplifies to a binary projective measurement for the case of $\Z_2$).  This is an approximate distinguishing measurement between the states $\ket{\psi} = \frac{\ket{x} + \ket{y}}{\sqrt{2}}$ and $\ket{\phi} = \frac{\ket{x} - \ket{y}}{\sqrt{2}}$ (i.e. the states corresponding to the two one-dimensional irreducible representations of $\otherrep$), and we calculate the bias below. We assume here without loss of generality that the probability of accepting $\ket{\psi}$ is higher than the probability of accepting $\ket{\phi}$, and that $\proj{0}$ corresponds to the accept outcome.  If these are not the case, then the roles of $\ket{\psi}$ and $\ket{\phi}$ or $0$ and $1$ can be swapped.  
\begin{align*}
    \left|\left(\id \otimes \bra{0}\right) \mathcal{M} \ket{\psi} \right|^2- \left| \left( \id \otimes \bra{0}\right)\mathcal{M} \ket{\phi}\right|^2 
    &= 
        \left| \left(\id \otimes \bra{0}\right) \mathcal{M} \ket{\psi} \right|^2 - \left(1 - \left|\left(\id \otimes \bra{1}\right) \mathcal{M} \ket{\phi}\right|^2\right)
    \\
    &=
        \left| \left(\id \otimes \bra{0}\right) \mathcal{M} \ket{\psi} \right|^2+ \left|\left(\id \otimes \bra{1}\right) \mathcal{M} \ket{\phi} \right|^2 - 1 
    \\
    &\geq 
        \frac{1}{2}( \left|\left(\id \otimes \bra{0}\right) \mathcal{M} \ket{\psi} \right| + \left|\left(\id \otimes \bra{1}\right) \mathcal{M} \ket{\phi}\right|)^2 - 1
    \\
    &\geq 
        2 \left(\frac{1}{2} \left( \Re \left(\id \otimes \bra{0}\right) \mathcal{M} \ket{\psi} + \Re\left(\id \otimes \bra{1}\right) \mathcal{M} \ket{\phi}\right)\right)^2 - 1
    \\
    &\geq 
        2\left(\frac{1}{2} + \frac{\Re(a+b)}{4}\right)^2 - 1
    \\
    &= 
        2\left(\frac{1}{4} + \frac{\Re(a+b)}{4} + \frac{\Re(a+b)^2}{16}\right) - 1
    \\
    &= 
        \frac{\Re(a+b)}{2} + \frac{\Re(a+b)^2}{8} - \frac{1}{2}
    \,.
\end{align*}
Here we note that the error bound is much weaker than the tight bound proved in \cite{aaronson2020hardness}.  While our approximate duality theorem is tight with respect to the Hilbert-Schmidt inner product, it does not necessarily recover an optimal \emph{distinguishing} measurement.  The bound in~\cite{aaronson2020hardness} in fact modifies modifies the circuit to get a tighter bound, and we comment on this more later.

In the other direction, assume that we have a measurement that accepts $\ket{\psi}$ with probability $p$ and $\ket{\phi}$ probability $p - \Delta$.  Then we can first construct a measurement that applies the original measurement, copies the result over, and un-computes the measurement.  For this measurement, we have the following:
\begin{equation*}
    \Re (\id \otimes \bra{0}) \mathcal{M} \ket{\psi} = \sqrt{p}\,.
\end{equation*}
and similarly 
\begin{equation*}
    \Re (\id \otimes \bra{1}) \mathcal{M} \ket{\phi} = \sqrt{1 - (p - \Delta)}\,.
\end{equation*}
Note that in this case \Cref{thm:approx_duality_2} works up to any unitary applied to $\ket{\psi}$ and $\ket{\phi}$, since they are still orthogonal and thus are a basis for some exact representation of $\mathbb{Z}_2$.   So we can always pick a unitary on the first register such that the archetype states are exactly the residual states of $\mathcal{M}$ after measuring.  Then we have the following bound on the condition of \Cref{thm:approx_duality_2}:
\begin{align*}
    1 - \epsilon &= \frac{1}{2} \left(\sqrt{p} + \sqrt{1 - p + \Delta}\right)\\
    &\geq \sqrt{\frac{1 + \Delta}{2}}\,.
\end{align*}
Here we minimize this expression over $p$ by setting $p = \frac{1 + \Delta}{2}$.  Let $U$ be the unitary we implement when applying \Cref{thm:approx_duality_2} and $\widehat{U}$ be the unitary that swaps $\ket{x}$ and $\ket{y}$.  Combined with our calculation before, we have the following
\begin{equation*}
    \avg_{g \in \mathbb{Z}_2} \left\|U^{g} - \widehat{U}^{g}\right\|^2 = 1 - \frac{\Re(a+b)}{2} \leq 2\left(1 - \sqrt{\frac{1+\Delta}{2}}\right)\,.
\end{equation*}
Here $a$ and $b$ are $\bra{x}U\ket{y}$ and $\bra{y}U\ket{x}$ respectively. Rearranging terms, we have that 
\begin{equation*}
    \frac{\Re(a+b)}{2} \geq 2\sqrt{\frac{1+\Delta}{2}} - 1 \geq \Delta\,.
\end{equation*}

The reason we are able to get a tighter duality in this direction is because we can alter the measurement \emph{before hand} so that the real component becomes the same as the absolute value, where as to do the same in the forward direction requires modifying the unitary in a way that depends on the group element, and thus would need to be written into the implementation of the duality theorem itself.  

Thus, we recover a non-tight version of the approximate duality from~\cite{aaronson2020hardness}. 
As noted before, in order to get a tighter bound, the approximate duality of~\cite{aaronson2020hardness} analyzes a slightly different algorithm, in which instead of controlling the swap on the positive superposition between $\ket{0}$ and $\ket{1}$, the control qubit is initialized as $\frac{1}{\sqrt{2}}(\ket{0} + e^{i\theta} \ket{1})$, with an arbitrary phase that depends on $a$ and $b$. In our case, this corresponds to initializing the control register with a state that differs from the uniform positive superposition on the group (i.e. the trivial irrep). Specifically, each group element would receive a phase that depends on the Hilbert Schmidt inner product between the $\epsilon$-close representation and some exact representation on $g$ (since we have the freedom to alter the isometry and unitary, we can take any exact representation).  
This does not work na\"ively, in part because it would seem to require computing an exponential number of complex phases (in the size of the binary representation of the group), but we suspect that such a strategy may be possible in order to get a tighter bound.  We leave it to future work to prove a tighter version of the generalized duality theorem.

\begin{remark}
    One might wonder what would happen if we proved a similar theorem, but instead starting from the result of \cite{Kazhdan1982On}.  Here, the definition of $\epsilon$-approximate is with respect to the operator norm, but there is no need for an isometry in the resulting exact representation.  However, the stricter requirements on this approximate representation make it hard to apply to ``approximate adversaries'' in the way that we would want.  In particular, an adversary that breaks some game with inverse polynomial probability might succeed with very high probability in some cases, but $0$ in others.  This means that the result of \cite{Kazhdan1982On} does not help us transform these adversaries into other useful adversaries.
\end{remark}

\section{Quantum Lightning From Non-Abelian Group Actions}
We generalize the construction of quantum money / lightning of~\cite{ITCS:Zhandry24a} to general group actions. This allows us to instantiate the construction from a potentially much wider class of group action instantiations. Generalizing to non-Abelian groups, specifically, also allows us to show a security reduction from a concrete computational assumption \emph{in the plain model}.%
\footnote{By contrast,~\cite{ITCS:Zhandry24a} is only able to show a security reduction in the black-box setting of generic group actions.} (See \Cref{sec:preaction-security} for a discussion of the assumption.)
Below, we present a quantum money construction from non-Abelian group actions.

\subsection{The Quantum Lightning Construction}
\label{sec:qm-construction}

Let $G$ be a group with an efficient quantum Fourier transform and a negligible maximum Plancherel measure (that is, each irrep $\lambda$ of $G$ has dimension at most $\dimlambda := \dim(\varrho_{\lambda}) \le \sqrt{|G|} \cdot \negl(\log|G|)$). 
For example, we can take $G$ to be the dihedral group $D_{2^n}$ or the symmetric group $S_n$.
Let $* : G \times X \to X$ be a semiregular group action of $G$ on some set $X$, and let $x \in X$ be a fixed starting element in the set.
We build a our quantum lightning scheme as follows:

\paragraph{Mint:} To mint a quantum bank note, the mint begins with a copy of the starting element of the group $x \in X$ in a quantum register $\reg{B}$, in tensor product with the uniform superposition of all elements of the group.%
\footnote{
This can be attained by performing the inverse quantum Fourier transform on the trivial irrep label of the group.
}
\begin{equation*}
    \frac{1}{\sqrt{|G|}}\sum_{g \in G} \ket{g}_{\reg{A}} \ket{x}_{\reg{B}}\,.
\end{equation*}
The mint then applies the group action, controlled on register $\reg{A}$, yielding the following quantum state:
\begin{equation*}
    \frac{1}{\sqrt{|G|}}\sum_{g \in G} \ket{g}_{\reg{A}} \ket{g * x}_{\reg{B}}\,.
\end{equation*}
The mint inverts the group element in register $\reg{A}$ to get:
\begin{align*}
    \frac{1}{\sqrt{|G|}}
    \sum_{g \in G} 
    \ket{g^{-1}}_{\reg{A}} 
    \ket{g * x}_{\reg{B}}
    &=
    \frac{1}{\sqrt{|G|}}
    \sum_{g \in G} 
    \sum_{\substack{\lambda \in \hat G \\ i, j \in [\dimlambda]}} 
    \sqrt{\frac{\dimlambda}{|G|}}
    \;
    \varrho_{\lambda}(g^{-1})_{i, j}
    \ket{\mathcal{L}^{\lambda}_{j, i}}_{\reg{A}} 
    \ket{g * x}_{\reg{B}}
    \\
    &=
    \frac{1}{\sqrt{|G|}}
    \sum_{\substack{\lambda \in \hat G \\ i, j \in [\dimlambda]}} 
    \ket{\mathcal{L}^{\lambda}_{j, i}}_{\reg{A}} 
    \sqrt{\frac{\dimlambda}{|G|}}
    \sum_{g \in G} 
    \varrho_{\lambda}(g^{-1})_{i, j}
    \ket{g * x}_{\reg{B}}
    \,,
\end{align*}
where 
$
    \ket{\mathcal{L}^{\lambda}_{a, b}} 
    :=
    \sqrt{\frac{\dimlambda}{|G|}}
    \sum_{h \in G} 
    \;
    \varrho_{\lambda}(h^{-1})_{a, b}
    \ket{h}
$
is the Fourier basis state of the left-regular representation.

The mint then applies the quantum Fourier transform on $\reg{A}$, yielding the following state:
\begin{align}
    \label{eq:pre_measurement_money_state}
    \frac{1}{\sqrt{|G|}}
    \sum_{\substack{\lambda \in \hat G \\ i, j \in [\dimlambda]}} 
    \ket{\lambda, j, i}_{\reg{A}} 
    \sqrt{\frac{\dimlambda}{|G|}}
    \sum_{g \in G} 
    \varrho_{\lambda}(g^{-1})_{i, j}
    \ket{g * x}_{\reg{B}}
    \,.
\end{align}
The mint then measures $\reg{A}$ in the computational basis to get an label $\lambda \in \widehat{G}$, as well as two Fourier indices $i, j \in [\dimlambda]$. 
The residual state on register $\reg{B}$ becomes:
\begin{align*}
    \ket{\mathcal{\$}^{\lambda}_{i, j}} 
    &:=
    \sqrt{\frac{\dimlambda}{|G|}}
    \sum_{g \in G} 
    \varrho_{\lambda}(g^{-1})_{i, j}
    \ket{g * x}
    \,.
\end{align*}
Output $\lambda$ as the serial number, and $\ket{\mathcal{\$}^{\lambda}_{i, j}}$ as the quantum money state.
This completes the description of~$\mathsf{Mint}$.

\begin{lemma}
    \label{lemma:money-orthonormal}
    The set possible money states $\{\ket{\mathcal{\$}^{\lambda}_{i, j}}\}_{\lambda \in \widehat{G}, \, i,j \in [\dimlambda]}$ is orthonormal. That is $\braket{\mathcal{\$}^{\lambda}_{i, j} \,|\, \mathcal{\$}^{\sigma}_{k, \ell}} = \delta_{\lambda, \sigma} \delta_{i, k} \delta_{j, \ell}$.
\end{lemma}
\begin{proof}
    This follows straightforwardly from Schur orthogonality relations (\Cref{sec:schur-orthogonality}) and the fact that the group action is semiregular (that is, $g*x = h*x$ only if $g = x$). We have:
    \begin{align*}
        \braket{\mathcal{\$}^{\lambda}_{i, j} \,|\, \mathcal{\$}^{\sigma}_{k, \ell}}
        &=
        \frac{\sqrt{\dimlambda d_{\sigma}}}{|G|}
        \sum_{g, h \in G} 
        \varrho_{\lambda}(g^{-1})^{*}_{i, j}
        \varrho_{\sigma}(h^{-1})_{k, \ell}
        \braket{g * x \,|\, h * x}
        \\
        &=
        \frac{\sqrt{\dimlambda d_{\sigma}}}{|G|}
        \sum_{g \in G} 
        \varrho_{\lambda}(g^{-1})^{*}_{i, j}
        \varrho_{\sigma}(g^{-1})_{k, \ell}
        \\
        &=
        \frac{\sqrt{\dimlambda d_{\sigma}}}{|G|}
        \cdot
        \frac{|G|}{\dimlambda}
        \delta_{\varrho, \sigma}
        \delta_{i, k}
        \delta_{j, \ell}
        \\
        &=
        \delta_{\varrho, \sigma}
        \delta_{i, k}
        \delta_{j, \ell}
        \,.\qedhere
    \end{align*}
\end{proof}

\begin{lemma}
    \label{lemma:minting-plancherel}
    The serial number\textemdash that is, the irrep label $\lambda$\textemdash produced by the Minting is sampled according to the Plancherel measure of $\lambda$ in $G$. That is, for all $\lambda \in \widehat{G}$,
    \begin{align*}
        \Pr
        \left[
            \lambda = \sigma
            \;\middle|\;
            (\sigma, \ket{\$^{\sigma}_{ij}})
            \gets
            \mint()
        \right]
        =
        \frac
        {\dimlambda^2}
        {\abs{G}}
        \,.
    \end{align*}
\end{lemma}
\begin{proof}
    We note that can write \Cref{eq:pre_measurement_money_state} as:
    \begin{align*}
        \frac{1}{\sqrt{|G|}}
        \sum_{\substack{\lambda \in \widehat G \\ i, j \in [\dimlambda]}} 
        \ket{\lambda, j, i}_{\reg{A}} 
        \ket{\$^{\lambda}_{i, j}}_{\reg{B}}
        \,.
    \end{align*}
    where the $\ket{\$^{\lambda}_{i, j}}$'s are orthonormal by \Cref{lemma:money-orthonormal}.
    We can see directly that the probability of measuring any triplet of 
    $(\lambda, j, i)$ in register $\reg{A}$
    is exactly $\frac{1}{\abs{G}}$.
    Furthermore, since for each $\lambda \in \widehat{G}$, $i$ and $j$ both run over $[\dimlambda]$, $\lambda$ appears in $\dimlambda^2$ such triplets. The total probability of the mint outputting serial number $\lambda$ is therefore $\frac{\dimlambda^2}{\abs{G}}$, which is the Plancherel measure of $\lambda$.
\end{proof}

\begin{lemma}
    For each $\lambda \in \widehat{G}$ and each $i \in [\dimlambda]$, the set $\{\ket{\mathcal{\$}^{\lambda}_{i, j}}\}_{j \in [\dimlambda]}$ spans a \manifestation, $W^{\lambda}_{i,x}$, of irreducible representation $\lambda$ in the group action representation $\rep(h) = \sum_{g \in G} \ketbra{hg * x}{g * x}$.
\end{lemma}
\begin{proof}
    Applying $\rep(h)$ to $\ket{\mathcal{\$}^{\lambda}_{i, j}}$ gives us:
    \begin{align*}
        \rep(h) 
        \ket{\$^{\lambda}_{i, j}} 
        &= 
        \sqrt{\frac{\dimlambda}{|G|}} 
        \sum_{g \in G} 
        \varrho_{\lambda}\left(g^{-1}\right)_{i, j}
        \ket{hg * x}
        \\
        &= 
        \sqrt{\frac{\dimlambda}{|G|}} 
        \sum_{g' = hg \in G} \varrho_{\lambda}\left((g')^{-1}h\right)_{i, j} \ket{g' * x}
        \\
        &= 
        \sqrt{\frac{\dimlambda}{|G|}} 
        \sum_{g \in G} \left( \varrho_{\lambda}\left(g^{-1}\right) \varrho_{\lambda}\left(h\right)\right)_{i, j} \ket{g * x}
        \\
        &= 
        \sqrt{\frac{\dimlambda}{|G|}} 
        \sum_{g \in G} \sum_{k \in [\dimlambda]} \varrho_{\lambda}\left(g^{-1}\right)_{i, k} \varrho_{\lambda}\left(h\right)_{k, j} \ket{g * x}
        \\
        &= 
        \sum_{k \in [\dimlambda]} \varrho_{\lambda}(h)_{k, j} \sum_{g \in G} \varrho_{\lambda}(g^{-1})_{i, k} \ket{g * x}
        \\
        &= 
        \sum_{k \in [\dimlambda]} \varrho_{\lambda}(h)_{k, j} \ket{\$^{\lambda}_{i, k}}
        \,.
    \end{align*}
    We can see that $\rep(h)$ acts exactly as the irreducible representation $\varrho_{\lambda}$ on the space spanned by the money states $\{\ket{\mathcal{\$}^{\lambda}_{i, j}}\}_{j \in [\dimlambda]}$. They must therefore span the same \manifestation $W^{\lambda}_{i,x}$ of irreducible representation $\varrho_{\lambda}$ in $\rep(h)$.
\end{proof}

\begin{corollary}
    For all $\lambda \in \widehat{G}$ and $i, j \in [\dimlambda]$, the money state
    $\ket{\$^{\lambda}_{i, j}}$ is in the subspace 
    $W^{\lambda}_{x} = \bigoplus_{i \in [\dimlambda]} W^{\lambda}_{i,x}$ 
    corresponding to irreducible representation $\varrho_{\lambda}$ of the group action representation starting from $x$.
    Moreover, if the group action has multiple orbits, then the full isotypic component of $\lambda$ is
    $W^{\lambda} = \bigoplus_{y \in \mathsf{Orb}(\rep)} W^{\lambda}_{y}$ where $y$ runs over the set of orbits of the group action, choosing an element from each orbit arbitrarily. 
\end{corollary}

\paragraph{Ver:} To verify, we begin by measuring that the state has support only on the set $X$. We then repeat essentially the same process as for minting, but starting with the claimed banknote in the second register, rather than $\ket{x}$. Suppose we want to verify a state $\ket{\yen^{\lambda}}$ with claimed serial number $\lambda$, we prepare the uniform superposition over group elements, perform the group action on $\ket{\yen^{\lambda}}$ in superposition, and then measure the control register in the Fourier basis. That is, we perform a course Fourier measurement on $\ket{\yen^{\lambda}}$ and check if it has the claimed label.

Suppose that $\ket{\yen^{\lambda}}$ is a valid state for label $\lambda$. That is $\ket{\yen^{\lambda}} = \sum_{i,j \in [\dimlambda]} \alpha_{i, j} \ket{\$^{\lambda}_{i, j}}$ for some coefficients $\alpha_{i, j}$. This gives the following:

\begin{align*}
    \frac{1}{\sqrt{|G|}}\sum_{g \in G} \ket{g} \ket{\yen^{\lambda}}
    =& \frac{1}{\sqrt{|G|}}\sum_{g \in G} \ket{g} \sum_{i,j \in [\dimlambda]} \alpha_{i, j} \ket{\$^{\lambda}_{i, j}}
    \\
    \xrightarrow[\text{action}]{\text{group}}&
    \frac{1}{\sqrt{|G|}}\sum_{g \in G} \ket{g} \sum_{i,j,k \in [\dimlambda]} \alpha_{i, j} \; \varrho_{\lambda}(g)_{k, j} \ket{\$^{\lambda}_{i, k}}
    \\
    =&
    \sum_{i,j,k \in [\dimlambda]} \alpha_{i, j} \frac{1}{\sqrt{|G|}}\sum_{g \in G} \varrho_{\lambda}(g)_{k, j} \ket{g} \ket{\$^{\lambda}_{i, k}}
    \\
    \xrightarrow[g]{\text{invert}}&
    \sum_{i,j,k \in [\dimlambda]} \alpha_{i, j} \frac{1}{\sqrt{|G|}}\sum_{g \in G} \varrho_{\lambda}(g)_{k, j} \ket{g^{-1}} \ket{\$^{\lambda}_{i, k}}
    \\
    =&
    \sum_{i,j,k \in [\dimlambda]} \alpha_{i, j} \frac{1}{\sqrt{|G|}}\sum_{g \in G} \varrho_{\lambda}(g^{-1})_{k, j} \ket{g} \ket{\$^{\lambda}_{i, k}}
    \\
    =&
    \sum_{i,j,k \in [\dimlambda]} \frac{\alpha_{i, j}}{\sqrt{\dimlambda}} \ket{\mathcal{L}^{\lambda}_{k, j}} \ket{\$^{\lambda}_{i, k}}
\end{align*}

Now if we perform a course (left-regular) Fourier basis measurement on the first register (perform a Fourier transform and measure the irrep label) we get the correct serial number $\lambda$. Now we have two options: we can further perform a fine Fourier basis measurement to get $k$ and $j$, collapsing the quantum money state to a new state $\ket{\$'^\lambda} = \sum_{i \in [\dimlambda]} \alpha'_{i, k} \ket{\$^{\lambda}_{ik}}$ with different weights on the same set $\{\ket{\$^{\lambda}_{a, b}}\}_{a,b \in [\dimlambda]}$ of basis states, but nevertheless still a valid quantum money state. Or, alternatively, we can refrain from measuring $k$ and $j$ and simply un-compute the whole process, which, in the case that verification passed with certainty, recovers the original state~$\ket{\yen^{\lambda}}$.


\begin{remark}
    Note that states of the form $\sum_{i,j \in [\dimlambda]} \alpha_{i, j} \ket{\$^{\lambda}_{i, j}}$ are not the only states that pass verification. If we denote 
    $\ket{\$^{\lambda}_{i, j} * x} := \sqrt{\frac{\dimlambda}{|G|}}\sum_{g \in G} \varrho_{\lambda}\left(g^{-1}\right)_{i, j} \ket{g * x}$ as the quantum money state produced by beginning with starting set element $x \in X$, then states of the form $\ket{\$^{\lambda}_{i, j} * y}$ for all $y \in X$ and $i,j \in [\dimlambda]$ (and their superpositions) also pass verification. Thus they are also valid quantum money states, despite not being the result of the minting process, and must be considered in the security arguments.
\end{remark}

\subsection{Variations on the Construction}

Here, we describe some possible variations of the above scheme.

\paragraph{Membership check.} The set $X$ may be a collection of sparse strings. In this case, a quantum money adversary may try to forge with fake banknotes that have support outside of $X$. If the group action supports membership testing for $X$, it is natural to also have the verifier check that a supposed banknote has support on $X$. Such a check is used in~\cite{ITCS:Zhandry24a} to analyze the security of their construction.  For may group actions, however, such a membership check is not efficiently feasible. 
In the case in which the group still acts compatibly (or approximately so) on elements outside of $X$ which cannot be distinguished from $X$, then this can be treated as the group action having additional orbits.
In \Cref{sec:intransitive}, we show an example of how to handle such a group action.

\paragraph{Irrep check.} It may be useful to insist that the serial number of the banknote corresponds to an irrep with certain properties. Notably, we will consider adding checks on the \emph{dimension} of the irrep, assuming the dimension is efficiently computable. For example, we may insist that banknotes come from irreps of dimension at least 2.

For such irrep checks, in order to ensure correctness, we need to ensure that the mint always produces irreps with the given property. If such irreps are at least inverse-polynomially dense according to the Plancherel measure, we can have the mint keep minting banknotes until it produces one with the given property. 

The following lemma shows that the irreps of size at least 2 are dense for all non-Abelian groups. Thus, for any non-Abelian group action, we can insist on valid banknotes having irrep dimension at least 2. We will make this assumption in the security analysis of our scheme in the following subsections.

\begin{lemma}
\label{lemma:prob-of-1d-irrep}
For any non-Abelian $G$, let $d$ be the dimension of a random irrep sampled according to the Plancherel measure. Then $\Pr[d\geq 2]\geq 1/2$.
\end{lemma}
\begin{proof}The 1-dimensional irreps are in bijection with the quotient of the commutator subgroup $G/[G,G]$. Since $|[G,G]|\geq 2$ for non-Abelian groups, $|G/[G,G]|\leq |G|/2$. The probability of sampling any given 1-dimensional irrep according to the Plancherel measure is $1/|G|$. Over all $\leq |G|/2$ such irreps, the probability of sampling \emph{any} 1-dimensional irrep is at most $1/2$. This means the probability sampling an irrep of dimension 2 or higher is at least $1/2$.
\end{proof}

\subsection{Security from Preaction Secure Group Actions}
\label{sec:preaction-security}

In this subsection we give a security proof from cryptographic group actions that are preaction-secure, which we define here. A \emph{preaction} on a group action is an operation that on input $x, g*x \in X$ computes $h \circ (g*x) := gh^{-1} * x$ for some $h \in \G$.%
\footnote{We assume here, as before, that the action is semiregular, so that this is well-defined.}
That is, it prepends a group element $h^{-1}$ on the right of $g$, as if $h^{-1}$ had acted \emph{before} $g$ had acted.
While for Abelian group actions, this is equivalent to the group action itself (up to inverting $h$), for non-Abelian groups, this is not generically efficient. Note, however that a preaction \emph{is} itself a group action, as it satisfies compatibility---that is, $h_1 \circ (h_2 \circ (g*x)) = (h_1 h_2) \circ (g*x)$. And moreover, the preaction of a preaction is the original group action.

We introduce both a search-type assumption and a decision-type assumption, that constitute different levels of preaction security. The search-type assumption, preaction hardness, requires that it is computationally hard to perform a preaction. The decision-type assumption, preaction indistinguishability, requires that it is hard to tell when a preaction has been performed. In both cases, the preactions are defined relative to a predetermined and fixed starting element $x \in X$.


\begin{assumption}[$\epsilon$-Preaction Hardness]
    \label{assum:preaction-is-hard}
    Given $g*x$, and $h$ for random $g,h \gets \G$, it is hard to output $gh^{-1} * x$.  That is, for all QPT adversaries, $\mathcal{A}$,%
    \footnote{
        Note that an adversary can always trivially perform a preaction with probability $\frac{1}{\abs{G}}$ by performing an action with a random group element. This is because the orbits of the action and preaction are always the same. We thus define the advantage as the best that can be done beyond this trivial attack. Groups that are not too far from being Abelian may have a different specialized trivial attack, so it may be convenient to assume the difficulty beyond that, but for simplicity, we do not get into that here.    
    }
    \begin{align*}
        \Pr\left[ z = gh^{-1} * x \; : \; x \gets X, \; g, h \gets \G, \; z \gets \mathcal{A}(x, g*x, h) \right] \le \frac{1}{\abs{\G}} + \epsilon\,.
    \end{align*}
\end{assumption}

\begin{longfbox}[breakable=false, padding=1em, margin-top=1em, margin-bottom=1em]
\begin{algo}[{\bf Preaction Indistinguishability Security Game}]\label{prot:preaction-ind-security-game}
\end{algo}
\begin{enumerate}
    \item Challenger samples $b \in \bits$ and two uniformly random group elements $h_1, h_2 \gets \G$.
    \item Adversary sends a register $\reg{A}$ to the challenger.
    \label{step:preaction-ind-query}
    \item If $b = 0$, the challenger applies the action of $h_1$ to $\reg{A}$. Otherwise, the challenger applies both the action of $h_1$ and the preaction of $h_2$ to $\reg{A}$. Send $\reg{A}$ back to the adversary.
    \label{step:preaction-ind-chal}
    \item Adversary outputs $b'$ and wins if $b' = b$.
\end{enumerate}
\end{longfbox}

\begin{assumption}[$\epsilon$-Preaction Indistinguishability]
    \label{assum:preaction-indist}
    It is hard to distinguish whether a preaction has been performed.  
    Formally, no adversary can win at the preaction indistinguishability security game (\Cref{prot:preaction-ind-security-game}) with advantage greater than $\epsilon$.
    That is, if we write the action of the challenger in Step \ref{step:preaction-ind-chal} as $\mathsf{F}_b^{h_1, h_2} \; : \; g*x \;\mapsto\; h_1 \, g \, h_2^{-b} * x$. 
    Then for all QPT adversaries, 
    $\mathcal{A}$,
    that make a single query%
    \footnote{
        We could in general define multi-round security game for preaction indistinguishability, in which Steps \ref{step:preaction-ind-query} and \ref{step:preaction-ind-chal} are repeated. Preaction security defined this way would be a stronger assumption, and may be useful for other settings. However, we do not formally define the stronger version as we are able to prove security from this weaker assumption, which gives a stronger security guarantee, and is more likely to hold for a larger class of group actions.
    }
    to $\mathsf{F}$,
    \begin{align*}
        \left|
            \Pr
            \left[ 
                0 \gets \mathcal{A}^{\mathsf{F}_0^{h_1, h_2}} 
                \; : \; 
                h_1, h_2 \gets \G 
            \right] 
            -
            \Pr
            \left[ 
                0 \gets \mathcal{A}^{\mathsf{F}_1^{h_1, h_2}} 
                \; : \; 
                h_1, h_2 \gets \G 
            \right] 
        \right|
        \le 
        \epsilon
    \end{align*}
    Note that when $b = 0$, $\mathsf{F}_b^{h_1, h_2}$ performs a group action for a random group element $h_1$, and when $b = 1$, it performs both a random \emph{bi-action}---that is, a random group action with $h_1$ and a random group pre-action with $h_2$.  
\end{assumption}

\begin{definition}
    We say that a group action of group $\G_{\lambda}$ on set $X_{\lambda}$ with starting element $x$ is $\epsilon$-\emph{preaction secure} if both \Cref{assum:preaction-is-hard} and \Cref{assum:preaction-indist} hold for the group action against any QPT adversary with advantage $\epsilon$. We say that the group action is preaction secure if it is $\negl(\lambda)$-preaction secure for any negligible function $\negl$.
\end{definition}

\begin{remark}
    Classically, distinguishing preactions in one round is \emph{information-theoretically impossible}. This is because both cases\textemdash with or without a preaction\textemdash send the element to a uniformly random element in its orbit. Interestingly, as we will see, this is \emph{not} the case for quantum distinguishers, since they are allowed to query $\mathsf{F}_b^{h_1, h_2}$ on superpositions of elements.
\end{remark}

\begin{remark}\label{remark:preactions-trivial-for-Abelian}
    For \emph{Abelian} group actions, breaking preaction hardness is trivial, since the preaction is equal to the action. On the other hand, for this same reason, preaction indistinguishability is \emph{information-theoretically impossible}: since the preaction is equal to the action, both cases\textemdash with or without a preaction\textemdash end up performing a uniformly random group action.
\end{remark}

Because of \Cref{remark:preactions-trivial-for-Abelian}, preaction security is a security notion that only makes sense for non-Abelian group actions. Moreover, the security proof for our quantum money scheme makes explicit use of the properties of representations of non-Abelian groups to prove the reduction.

In fact, for quantum adversaries and non-Abelian group actions, preaction indistinguishability is a stronger assumption than preaction hardness:

\begin{theorem}\label{thm:preaction-ind-implies-hardness}
    Let $(\G, X, *, x)$ be a semiregular single-orbit group action of a non-Abelian group $\G$ acting on set $X$.
    Then if the group action satisfies preaction indistinguishability with advantage $\epsilon$, then it also satisfies preaction hardness with advantage $\theta\left(\sqrt{\epsilon + 2\avg_{\lambda \in \widehat{G}}\left[\frac{1}{\dimlambda}\right]}\right)$.
\end{theorem}
\begin{proof}[Proof sketch.]
    We defer the proof of \Cref{thm:preaction-ind-implies-hardness} to \Cref{proof:preaction-ind-implies-hardness} because it makes use of the quantum money construction of \Cref{sec:preaction-secure-construction}.  
    The main idea is that preactions are themselves a representation of the group, with the Fourier indices exchanging roles relative to their roles for the group action. Thus for semiregular single-orbit group actions, the ability to perform preactions allows us to measure in which \manifestation of the irrep a state lies via the duality theorem (\Cref{thm:approx_duality_1}), with some advantage. We can then distinguish if a preaction has occurred by testing if it has moved us to a different \manifestation of the irrep.
\end{proof}

Therefore, when working with non-Abelian group actions that are semiregular and single-orbit, it suffices for preaction security to consider only the decision-type assumption of the indistinguishability of preactions.

\subsubsection{Construction.}\label{sec:preaction-secure-construction}
Let $(\G, X, x)$ be a semiregular single-orbit
group action satisfying the requirements of \Cref{sec:qm-construction}. 
The quantum money/lightning construction follows the framework of \Cref{sec:qm-construction}.
We show below that if the group action is preaction secure, then the quantum money construction satisfies lightning security.

\subsubsection{Security}
Let $\ket{\$_{i, j}^{\lambda}} \propto \sum_{g \in \G} \varrho_{\lambda}(g^{-1})_{i, j} \ket{g * x}$ be the quantum money states minted by the scheme.
We show a tight connection between the ability to perform preaction (i.e. breaking preaction hardness, \Cref{assum:preaction-is-hard}) and performing a ``right representation'' on the quantum money state, that is, coherently mapping the quantum money state as $\ket{\$_{i, j}^{\lambda}} \mapsto \sum_{k \in [\dimlambda]} \varrho(h^{-1})_{i, k} \ket{\$_{k, j}^{\lambda}}$ on input $h \gets \G$. 
This right representation treats the span of the same vector across the different \manifestations of $\lambda$ as a single invariant subspace. In other words, it is to the standard group action representation what the right-regular representation is to the left-regular representation.
We say that an adversary can perform the right representation with advantage $\epsilon$ if it can perform a unitary with Hilbert-Schmidt inner product at least $\frac{1}{\abs{\G}} + \epsilon$ with the ideal right representation.

\begin{lemma}\label{claim:right-repr-hard}
    Any adversary that performs a preaction $x, g*x, h \mapsto gh^{-1} * x$ with advantage $\epsilon$ for a fixed starting element, $x \in X$, and random $g,h \gets \G$, can be used to perform a right representation, $\ket{\$_{i, j}^{\lambda}} \mapsto \sum_{k \in [\dimlambda]} \varrho_{\lambda}(h^{-1})_{i, k} \ket{\$_{k, j}^{\lambda}}$ 
    with the same advantage $\epsilon$.
    Similarly any adversary that performs 
    a right representation with advantage $\epsilon$.
    can be used to perform 
    a preaction 
    with the same advantage $\epsilon$.
\end{lemma}

\begin{proof}
    Consider an adversary that performs preactions with advantage $\epsilon$, and consider what happens in the ideal case, in which the preaction is performed exactly. We start with the money state
    \begin{align*}
        \ket{\$_{i, j}^{\lambda}} \propto \sum_{g \in \G} \varrho_{\lambda}(g^{-1})_{i, j} \ket{g * x}
        \,.
    \end{align*}
    We perform a preaction with $h$ to get
    \begin{align*}
        \rightarrow &
        \sum_{g \in \G} \varrho_{\lambda}(g^{-1})_{i, j} \ket{g h^{-1} * x}
        \\
        = &
        \sum_{g \in \G} \varrho_{\lambda}(h^{-1} g^{-1})_{i, j} \ket{g * x}
        \\
        = &
        \sum_{\substack{g \in \G \\ k \in [\dimlambda]}} \varrho_{\lambda}(h^{-1})_{i, k} \varrho_{\lambda}(g^{-1})_{k, j} \ket{g * x}
        \\
        = &
        \sum_{k \in [\dimlambda]}
        \varrho_{\lambda}(h^{-1})_{i, k} 
        \sum_{g \in \G}
        \varrho_{\lambda}(g^{-1})_{k, j} \ket{g * x}
        \\
        \propto &
        \sum_{k \in [\dimlambda]}
        \varrho_{\lambda}(h^{-1})_{i, k} 
        \ket{\$_{i, j}^{\lambda}}
        \,.
    \end{align*}

    Let $U = \sum_{h, \lambda, i, k, j} \proj{h} \otimes \varrho_{\lambda}(h^{-1})_{i, k}\ket{\$^{\lambda}_{k, j}}\!\!\bra{\$^{\lambda}_{i, j}}$ be the unitary that performs the right representation controlled on a group element $h$. We can therefore rewrite it as $U = \sum_{h, g} \proj{h} \otimes \ket{g h^{-1} * x}\!\!\bra{g * x}$.

    Now suppose that the adversary performs preactions with advantage $\epsilon$. That is, it performs some $\widetilde{U} = \sum_{h, g} \proj{h} \otimes \ket{\psi(h, g*x)}\!\!\bra{g * x}$ where the probability of success is $\frac{1}{\abs{\G}^2}\sum_{h,g} \braket{gh^{-1} * x \,|\, \psi(h, g*x)} = \frac{1}{\abs{\G}} + \epsilon$. Then we have that, by the definition of the Hilbert-Schmidt inner product,
    $$
        \left\langle U, \widetilde{U}\right\rangle 
        =
        \frac{1}{\abs{\G}^2}
        \sum_{h, g} 
        \braket{g h^{-1} * x \,|\, \psi(h,g*x)}
        =
        \frac{1}{\abs{\G}} + \epsilon
        \,.
    $$

    Conversely, consider an adversary that performs the right representation with advantage $\epsilon$. That is, it performs some operator $\widetilde{U}$ such that $\left\langle U, \widetilde{U}\right\rangle = \frac{1}{\abs{\G}} + \epsilon$, where $U = \sum_{h, \lambda, i, k, j} \proj{h} \otimes \varrho_{\lambda}(h^{-1})_{i, k}\ket{\$^{\lambda}_{k, j}}\!\!\bra{\$^{\lambda}_{i, j}}$ is the ideal unitary that performs the right representation.
    
    Consider what happens when the ideal unitary is run on $\ket{g * x}$.  We start by writing $\ket{g * x}$ in the basis of the quantum money states $\{\ket{\$_{i, j}^{\lambda}}\}_{{\lambda \in \widehat\G, \; i,j \in [\dimlambda]}}$:
    \begin{align*}
        \ket{g * x}
        \propto \sum_{\substack{\lambda \in \widehat\G \\ i,j \in [\dimlambda]}} \varrho_{\lambda}(g)_{j, i} 
        \ket{\$_{i, j}^{\lambda}}\,.
    \end{align*}
    Now we perform the right representation to get
    \begin{align*}
        \rightarrow &
        \sum_{\substack{\lambda \in \widehat\G \\ i,j \in [\dimlambda]}} \varrho_{\lambda}(g)_{j, i}  
        \sum_{i' \in [\dimlambda]} \varrho_{\lambda}(h^{-1})_{i,i'} \ket{\$_{i', j}^{\lambda}}
        \\
        = &
        \sum_{\substack{\lambda \in \widehat\G \\ i,i',j \in [\dimlambda]}} 
        \varrho_{\lambda}(g)_{j, i}  
        \varrho_{\lambda}(h^{-1})_{i,i'} 
        \ket{\$_{i', j}^{\lambda}}
        \\
        = &
        \sum_{\substack{\lambda \in \widehat\G \\ i',j \in [\dimlambda]}} 
        \varrho_{\lambda}(gh^{-1})_{j,i'} 
        \ket{\$_{i', j}^{\lambda}}
        \\
        \propto &
        \ket{gh^{-1} * x}\,.
    \end{align*}
    If instead we run $\widetilde{U}$ on $\ket{g * x}$, and measure in the computational basis, we get the correct preaction with probability
    \begin{align*}
        \frac{1}{\abs{\G}^2} 
        \sum_{h, g \in \G}
        \bra{h} \otimes 
        \bra{gh^{-1} * x} 
        \widetilde{U} 
        \ket{h} \otimes 
        \ket{g * x}
        &=
        \frac{1}{\abs{\G}^2} 
        \sum_{h,g \in \G}
        \bra{h} \otimes 
        \bra{g * x} 
        U^{\dagger} 
        \widetilde{U} 
        \ket{h} \otimes 
        \ket{g * x} 
        \\
        &=
        \left\langle \widetilde{U}, U \right\rangle = \frac{1}{\abs{\G}} + \epsilon
        \,.
        \qedhere
    \end{align*}
\end{proof}

Therefore, pre-action hardness of the group action (\Cref{assum:preaction-is-hard}) is equivalent to the hardness of performing the right representation on the money states to map one \manifestation of an irrep to another \manifestation of the same irrep.

\begin{corollary}\label{cor:preaction-secure-multidimensional}
    For a group action to be $\delta$-preaction secure, at most a fraction $\frac{1}{\abs{\G}} + \delta$ of the Plancherel measure of $\G$ can be on irreps of dimension 1.
\end{corollary}
\begin{proof}
    Consider an adversary that simply implements the left-regular representation of the group (that is, the original group action).  Let $\Pi_{\mathrm{triv}}$ be the projector onto the irreducible representation spaces corresponding to $1$-dimensional representations, and note that both the left- and right-regular representations commute with $\Pi_{\mathrm{triv}}$.  Then we have that for every pair of group elements $g$ and $h$,
    \begin{align*}
        \abs{\braket{gh^{-1} * x | hg * x}}^2 &= \abs{\braket{gh^{-1} * x | \Pi_{\mathrm{triv}} | hg * x}}^2 + \abs{\braket{gh^{-1} * x | (\id - \Pi_{\mathrm{triv}}) | hg * x}}^2\\
        &\geq \abs{\braket{gh^{-1} * x | \Pi_{\mathrm{triv}} | hg * x}}^2\\
        &= \|\Pi_{\mathrm{triv}} \ket{g * x}\|^2\,.
    \end{align*}
    Here, we first use the fact that the left- and right- regular representations are block diagonal in the decomposition into $\Pi_{\mathrm{triv}}$ and $\id - \Pi_{\mathrm{triv}}$.  Then we use the fact that on $\Pi_{\mathrm{triv}}$, the left- and right-regular representations are equal, and thus they cancel each other out in the inner product.  Noting that $\|\Pi_{\mathrm{triv}} \ket{g * x}\|^2$ is equal to the Plancherel measure of $G$ on irreps of dimension $1$, we have that the adversary that applies the left-regular representation and measures in the computational basis has probability at least the Plancherel measure of $G$ on irreps of dimension $1$ of measuring the element $\ket{gh^{-1} * x}$, and thus breaks preaction security unless it is less than $\frac{1}{|G|} + \delta$.  
\end{proof}

\begin{remark}
    By \Cref{cor:preaction-secure-multidimensional}, we see that 
    for any preaction-secure group action, 
    the event of sampling a multi-dimensional irrep from the Plancherel measure happens with overwhelming probability, strengthening \Cref{lemma:prob-of-1d-irrep}, which states that it happens with probability at least $\frac12$ for general non-Abelian group actions.
    We can therefore always assume that the quantum money state sampled by the minting algorithm lies in a multi-dimensional irrep.
    We will assume therefore for the rest of the section that the quantum money verification rejects such 1-dimensional irreps.
    
\end{remark}

\begin{corollary}
    \label{cor:preaction-hard-to-right-meas}
    An adversary for preaction hardness with advantage $\epsilon$ can be used to perform a \textbf{right}-Fourier measurement on the quantum money state with that outputs the correct index $i$ of the \manifestation of $\varrho$ with advantage at least $4\epsilon^2 - \avg_{\lambda}\left[\frac{1}{\dimlambda}\right]$.%
    \footnote{
        where the measurement advantage here is defined as $|\Pr_{\varrho, i, j}[i \gets \mathcal A(\ket{\$^{\varrho}_{i,j}})] - \avg_{k} \Pr_{\varrho, i, j}[i \gets \mathcal A(\ket{\$^{\varrho}_{k,j}})] \,|$
    }
    That is, it can be used to measure $i$ for quantum money state $\ket{\$_{i, j}^{\lambda}} \propto \sum_{g \in \G} \varrho_{\lambda}(g^{-1})_{i, j} \ket{g * x}$.
\end{corollary}
\begin{proof}
    Intuitively, if we can break pre-action hardness, from \Cref{claim:right-repr-hard}, we can implement an approximate representation of the group where the Fourier indices $i$ and $j$ exchange roles, and from \Cref{thm:approx_duality_1} we can then implement a measurement that correctly outputs the index $i$ of the money state with high probability.  To prove this formally, let $\mathcal{M}$ be the approximate measurement implied by \Cref{thm:approx_duality_1}, with $V$ being the identity (since we know there exists an exact representation in the space, given by performing the exact preaction), we first note that by the definition of breaking preaction hardness with advantage $\epsilon$, and applying \Cref{claim:right-repr-hard}, together with the fact that $\avg_{g} \left\|\rep(g) - \otherrep(g)\right\|^2 = 1 - 2\avg_{g}\braket{\rep(g), \otherrep(g)} = 1-2(\epsilon + \frac{1}{|G|})$ from before (where $\otherrep$ is the exact irrep, used to form $U$), we get the following lower bound
    \begin{align*}
        2\left(\epsilon + \frac{1}{|G|}\right) &\leq \frac{1}{|G|}\sum_{\substack{\lambda \in \widehat{G},\, i \in [n_\lambda]\\ j \in [\dimlambda]}} \Re \bra{\phi_{i}^{\lambda}} \otimes \bra{\lambda, j} \mathcal{M}  \ket{\psi^{\lambda}_{i, j}}\\
        &\leq \frac{1}{|G|}\sum_{\substack{\lambda \in \widehat{G},\, i \in [n_\lambda]\\ j \in [\dimlambda]}} \left| \bra{\phi_{i}^{\lambda}} \otimes \bra{\lambda, j} \mathcal{M}  \ket{\psi^{\lambda}_{i, j}}\right|\,.
    \end{align*}
    Note that this function is the expectation of $|\bra{\phi_{i}^{\lambda}} \otimes \bra{\lambda, j} V \mathcal{M}^2 V^{\dagger} \ket{\psi^{\lambda}_{i, j}}|$, so we can apply Jensen's inequality with the function $f(x) = x^2$ to get the following lower bound
    \begin{align*}
        4\left(\epsilon + \frac{1}{|G|}\right)^2 &= \left(\frac{1}{|G|}\sum_{\substack{\lambda \in \widehat{G},\, i \in [n_\lambda]\\ j \in [\dimlambda]}} \left| \bra{\phi_{i}^{\lambda}} \otimes \bra{\lambda, j} \mathcal{M} \ket{\psi^{\lambda}_{i, j}}\right|\right)^2\\
        &\leq \frac{1}{|G|}\sum_{\substack{\lambda \in \widehat{G},\, i \in [n_\lambda]\\ j \in [\dimlambda]}} \left| \bra{\phi_{i}^{\lambda}} \otimes \bra{\lambda, j} \mathcal{M} \ket{\psi^{\lambda}_{i, j}}\right|^2\,.
    \end{align*}
    Since this quantity is the average probability of measuring the correct outcome given a uniformly random $\varrho$, $i$ and $j$, the advantage is at least $4\epsilon^2 - \avg_{\lambda}\left[\frac{1}{\dimlambda}\right]$, as desired. 
\end{proof}

We now show the following lemma which completes the proof of \Cref{thm:preaction-ind-implies-hardness}, showing that preaction indistinguishability implies preaction hardness for non-Abelian group actions.

\begin{lemma}\label{claim:right-fourier-measurement-hard}
    An adversary that can perform a \textbf{right}-Fourier measurement on the quantum money state with advantage $4\epsilon^2 - \avg_{\lambda}\left[\frac{1}{\dimlambda}\right]$ can be used to break preaction indistinguishability (\Cref{assum:preaction-indist}) with advantage $2\left(2\epsilon^2 - \avg_{\lambda \in \widehat{G}} \left[\frac{1}{\dimlambda}\right]\right)$.
\end{lemma}
\begin{proof}
    Assume at first that we have a perfect such adversary for performing right-Fourier measurements. 
    We start by running the minting algorithm to produce a uniformly random quantum money state 
    $
        \ket{\mathcal{\$}^{\lambda}_{i, j}} 
        :=
        \sqrt{\frac{\dimlambda}{|G|}}
        \sum_{g \in G} 
        \varrho(g^{-1})_{ij}
        \ket{g * x}
    $, along with its classical descriptors, the irreducible representation $\lambda \in \widehat{G}$ sampled according to the Plancherel measure (\Cref{lemma:minting-plancherel}) and uniformly random $i,j \in [\dimlambda]$.%
    \footnote{
        Note that the minting algorithm of the quantum lightning scheme as described in \Cref{sec:qm-construction} does not output $i$ and $j$, since they are not useful for verification (and in fact $i$ is not even verifiable). But they do pop up as part of the minting process, so we can modify the minting algorithm to output them as well. In fact, it is an odd quirk of our quantum lightning scheme that the minting party knows a piece of secret information about the money state\textemdash the \manifestation, $i$, of irrep $\lambda$ that the money state actually lies in\textemdash but that this information is completely useless. Neither the minting party nor anyone else can ever even verify this information! That is, unless they can break preaction hardness, which is what we assume here in this proof. 
    }
    
    We then apply the challenger given by \Cref{assum:preaction-indist} to get 
    \begin{align*}
        \rightarrow &
        \sum_{g \in \G} \varrho_{\lambda}(g^{-1})_{i, j} \ket{h_1 g h_2^{-b} * x}
        \\
        = &
        \sum_{g \in \G} \varrho_{\lambda}\left(h_2^{-b} g^{-1} h_1\right)_{i, j} \ket{g * x}
        \\
        = &
        \sum_{\substack{g \in \G \\ k,\ell \in [\dimlambda]}} 
        \varrho_{\lambda}\left(
        h_2^{-b} 
        \right)_{i, k} 
        \varrho_{\lambda}\left(
        g^{-1} 
        \right)_{k, \ell} 
        \varrho_{\lambda}\left(
        h_1
        \right)_{\ell, j} 
        \ket{g * x}
        \\
        = &
        \sum_{k,\ell \in [\dimlambda]} 
        \varrho_{\lambda}\left(
        h_2^{-b} 
        \right)_{i, k} 
        \varrho_{\lambda}\left(
        h_1
        \right)_{\ell, j} 
        \sum_{g \in \G} 
        \varrho_{\lambda}\left(
        g^{-1} 
        \right)_{k, \ell} 
        \ket{g * x}
        \\
        = &
        \sum_{k,\ell \in [\dimlambda]} 
        \varrho_{\lambda}\left(
        h_2^{-b} 
        \right)_{i, k} 
        \varrho_{\lambda}\left(
        h_1
        \right)_{\ell, j}
        \ket{\$_{k, \ell}^{\lambda}}
        \addtag\label{eq:preaction-black-box-result}
    \end{align*}

Suppose that $b=1$. 
Then when averaged over all pairs of group elements, $h_1$ and $h_2$, this gives 
\begin{align*}
    &
        \frac{1}{\abs{\G}^2}
        \sum_{h_1, h_2 \in \G}
        \;
        \sum_{k, k', \ell, \ell' \in [\dimlambda]} 
        \varrho_{\lambda}\left(
        h_2
        \right)^*_{i, k} 
        \varrho_{\lambda}\left(
        h_1
        \right)^*_{\ell, j}
        \varrho_{\lambda}\left(
        h_2
        \right)_{i, k'} 
        \varrho_{\lambda}\left(
        h_1
        \right)_{\ell', j}
        \ket{\$_{k,\ell}^{\lambda}}
        \!\!
        \bra{\$_{k',\ell'}^{\lambda}}
    \\
    &\;\; =
        \frac{1}{\abs{\G}^2}
        \sum_{k, k', \ell, \ell' \in [\dimlambda]} 
        \;
        \sum_{h_1 \in \G}
        \varrho_{\lambda}\left(
        h_1
        \right)^*_{\ell, j}
        \varrho_{\lambda}\left(
        h_1
        \right)_{\ell', j}
        \;
        \sum_{h_2 \in \G}
        \varrho_{\lambda}\left(
        h_2
        \right)^*_{i,k} 
        \varrho_{\lambda}\left(
        h_2
        \right)_{i,k'} 
        \ket{\$_{k,\ell}^{\lambda}}
        \!\!
        \bra{\$_{k', \ell'}^{\lambda}}
    \\
    &\;\; =
        \frac{1}{\dimlambda^2}
        \sum_{k, k', \ell, \ell' \in [\dimlambda]} 
        \delta_{\ell, \ell'}
        \delta_{k, k'}
        \ket{\$_{k,\ell}^{\lambda}}
        \!\!
        \bra{\$_{k',\ell'}^{\lambda}}
    \\
    &\;\; =
        \frac{1}{\dimlambda^2}
        \sum_{k, \ell, \in [\dimlambda]} 
        \ket{\$_{k, \ell}^{\lambda}}
        \!\!
        \bra{\$_{k, \ell}^{\lambda}}
    \,,
\end{align*}
where the second equality follows from the Schur orthogonality relations (\Cref{sec:schur-orthogonality}).
This is the fully mixed state over the isotypic component of $\lambda$---that is, over the union of all of the \manifestations of irrep $\lambda$.

Now with probability $1 - \frac{1}{\dimlambda} \ge \frac{1}{2}$ (since $\dimlambda \ge 2$), 
we get that $k \ne i$. That is, with probability at least $\frac{1}{2}$, the quantum money state has moved to a different \manifestation of the irrep $\lambda$, and measuring it again will confirm this.

If instead $b=0$, then $k=i$ with certainty (as $\varrho_{\lambda}((h_2^{0}))_{i, k} = \varrho_{\lambda}(\id)_{i, k} = \delta_{i, k} \quad \forall \lambda \in \widehat\G, h_2 \in \G$), so we instead get a fully mixed state over the $i$\ith \manifestation of $\lambda$. So we output $b' = 1$ if $k \ne i$ and $0$ otherwise. This gives a distinguishing advantage of at least $\frac{1}{2}$, breaking \Cref{assum:preaction-indist}.

\needspace{4\baselineskip}
Now suppose that the right-Fourier measurement adversary, $\mathcal{M}$, has advantage $\epsilon$ of measuring the correct \manifestation, over a uniformly random $\lambda$, $i$ and $j$. As shown above, in the case where there is no preaction, then it will never change the \manifestation, while in the case where there is a preaction, it will change to a uniformly random \manifestation.

We therefore have that the distinguishing advantage is 
\begin{align*}
    &
        \left|
            \Pr
            \left[ 
                0 \gets \mathcal{A}^{\mathsf{F}_0^{h_1, h_2}} 
                \; : \; 
                h_1, h_2 \gets \G 
            \right] 
            -
            \Pr
            \left[ 
                0 \gets \mathcal{A}^{\mathsf{F}_1^{h_1, h_2}} 
                \; : \; 
                h_1, h_2 \gets \G 
            \right] 
        \right|
    \\
    &\;\; =
        \left|
            \avg_{\substack{\lambda \in \widehat{G} \\ i,j \in [\dimlambda]}}
            \left| 
                 \id
                 \otimes
                 \bra{i}
                 \mathcal{M}
                 \ket{\$^{\lambda}_{i, j}}
            \right|^2 
            -
            \avg_{\substack{\lambda \in \widehat{G} \\ i,j, i',j' \in [\dimlambda]}}
            \left| 
                 \id
                 \otimes
                 \bra{i}
                 \mathcal{M}
                 \ket{\$^{\lambda}_{i', j'}}
            \right|^2 
        \right|
    \\
    &\;\; =
        \left|
            \avg_{\substack{\lambda \in \widehat{G} \\ i,j \in [\dimlambda]}}
            \left| 
                 \id
                 \otimes
                 \bra{i}
                 \mathcal{M}
                 \ket{\$^{\lambda}_{i, j}}
            \right|^2 
            -
            \avg_{\substack{\lambda \in \widehat{G}}}
            \frac{1}{\dimlambda}
        \right|
    \\
    &\;\;\geq 
        4\epsilon^2 - 2\avg_{\lambda \in \widehat{G}} \left[\frac{1}{\dimlambda}\right]
    \\
    &\;\;= 
        2\left(2\epsilon^2 - \avg_{\lambda \in \widehat{G}} \left[\frac{1}{\dimlambda}\right]\right)
    \,.\qedhere
\end{align*}
\end{proof}


\label{proof:preaction-ind-implies-hardness}
We can now complete the proof of \Cref{thm:preaction-ind-implies-hardness} by combining \Cref{claim:right-repr-hard,cor:preaction-hard-to-right-meas,claim:right-fourier-measurement-hard}.\qed
\vspace{1em}

We now turn to the quantum lightning security of the scheme.
We argue that any adversary who has two copies of the quantum money state can use them to break preaction indistinguishability (\Cref{assum:preaction-indist}). We therefore get a secure quantum lightning scheme from any group action that satisfies the syntactic requirements and is preaction-secure.

\paragraph{Focusing on the archetype states.}
For the analysis, before we proceed, it will be useful to consider a proxy for the quantum money states. The money states lie in a potentially large subspace, which is harder to analyze, so it is useful to instead focus on the archetype state that appears after performing a \fse, which is a unique state that characterizes each such subspace.

Suppose we have a quantum money state $\ket{\$^{\lambda}_{i,j}}$. We perform a \fse using \Cref{thm:duality_exact}, and get
\begin{align*}
    \ket{\$^{\lambda}_{i, j}} 
    \xrightarrow[]{FSE} 
    \ket{\phi^{\lambda}_{i}} \ket{\lambda}\ket{j} 
    = 
    \left(\frac{1}{\sqrt{\dimlambda}}\sum_{k \in [\dimlambda]} \ket{\$^{\lambda}_{i, k}} \otimes \ket{k}\right) \ket{\lambda} \ket{j}
\end{align*}

\begin{observation}
We observe that the archetype state $\ket{\phi^{\lambda}_{i}}$ in the first register is unaffected by applying the group action:
\begin{align*}
    \ket{\$^{\lambda}_{ij}} 
    \xrightarrow[]{\text{action by } h}
    \sum_{\ell \in [\dimlambda]} \varrho_{\lambda}(h)_{\ell, j} \ket{\$^{\lambda}_{i,\ell}} 
    \xrightarrow[]{FSE} 
    \ket{\phi^{\lambda}_{i}} \ket{\lambda}
    \left(
    \sum_{\ell \in [\dimlambda]} \varrho_{\lambda}(h)_{\ell, j} 
    \ket{\ell}
    \right)
\end{align*}

On the other hand, applying the corresponding \emph{preaction} performs the (inverted) irreducible representation $\varrho_{\lambda}$ onto the set of archetype states $\{\ket{\phi^{\lambda}_{i}}\}_{i \in [\dimlambda]}$ for the different \manifestations of $\varrho_{\lambda}$:
\begin{align*}
    \ket{\$^{\lambda}_{i, j}} 
    \xrightarrow[]{\text{preaction by } h}
    \sum_{\ell \in [\dimlambda]} \varrho_{\lambda}(h^{-1})_{i, \ell} 
    \ket{\$^{\lambda}_{\ell, j}} 
    \xrightarrow[]{FSE} 
    \left(
    \sum_{\ell \in [\dimlambda]} \varrho_{\lambda}(h^{-1})_{i, \ell} 
    \ket{\phi^{\lambda}_{\ell}} 
    \right)
    \ket{\lambda}
    \ket{j}
\end{align*}
\end{observation}

\begin{proposition}\label{prop:unclonable_from_preaction_secure}
    Suppose that the group action used in the quantum money construction (\Cref{sec:preaction-secure-construction}) is $\epsilon$-preaction secure.
    Then no QPT adversary can produce a quantum state on two registers such that the probability of measuring both registers in the same irreducible representation subspace of $\varrho_{\lambda}$ is greater than $2\dimlambda \epsilon / (1 + \dimlambda)$.
\end{proposition}
\begin{proof}
    Assume for the sake of contradiction that an adversary for quantum lightning, $\mathcal A$, can prepare a quantum state on two registers, both of which pass verification. By definition, the verifier projects onto $W^{\lambda}$.  Since we have shown that $\ket{\$^{\lambda}_{i, j}}$ is a basis for $W^{\lambda}$, the states produced by $\mathcal A$ must be supported on states of the form
    \begin{align*}
        \ket{\$_{i, j}^{\lambda}} \otimes \ket{\$_{k, \ell}^{\lambda}} 
        \qquad
        \text{where } \ket{\$_{i, j}^{\lambda}} = \sqrt{\frac{\dimlambda}{\abs{\G}}} \sum_{g \in \G} \varrho_{\lambda}(g^{-1})_{i, j} \ket{g * x}
    \end{align*}
    for some $\lambda\in \widehat\G$ such that $1 < \dimlambda$.
    We show that this adversary can be used to break \Cref{assum:preaction-indist}.  Let $\mathsf{Chal}_{b}$ be the challenger given in the assumption, which either applies a random action and random pre-action ($b = 1$), or just applies a random action ($b = 0$).  
    
    To demonstrate the idea, we first assume that $i = k$, that is, that the two registers initially lie in the same Fourier subspace of $\lambda$. We will see later how to handle the more general case.
    Suppose that we take only one of the two registers and apply $\mathsf{Chal}_{b}$. We get (see \Cref{eq:preaction-black-box-result})
    \begin{align*}
        &
        \sum_{r,s \in [\dimlambda]} 
        \varrho_{\lambda}\left(
        h_2^{-b} 
        \right)_{i, r} 
        \varrho_{\lambda}\left(
        h_1
        \right)_{s,  j}
        \ket{\$_{r, s}^{\lambda}}
        = 
        \begin{cases}
            \sum_{s \in [\dimlambda]} 
            \varrho_{\lambda}\left(
            h_1
            \right)_{s, j}
            \ket{\$_{i,s}^{\lambda}}
            &
            b = 0
            \\
            \sum_{r,s \in [\dimlambda]} 
            \varrho_{\lambda}\left(
            h_2^{-1} 
            \right)_{i,r} 
            \varrho_{\lambda}\left(
            h_1
            \right)_{s, j}
            \ket{\$_{r,s}^{\lambda}}
            &
            b = 1
        \end{cases}
    \end{align*}
    
    Then if $b = 0$ (i.e. the challenger did not apply a pre-action), the state remains in the same Fourier subspace with certainty, and so a swap test between the \arch states produced by performing a \fse on both registers will succeed with probability 1, and we output $b' = 0$.

    If $b = 1$, then with probability $1 - \frac{1}{\dimlambda} \ge \frac12$ (since $\dimlambda \ge 2$), the resulting state is in a different Fourier subspace. In this case, the swap test between the \arch states fails with probability $\frac{1}{2}$, in which case we output $b' = 1$.  Thus, in this case, we output $1$ with probability at least $\frac{1}{4}$.  The overall success probability is therefore $\frac{1}{2} + \frac{1}{8} = \frac{5}{8}$, breaking \Cref{assum:preaction-indist}.

    However, the initial states need not lie in the same initial Fourier subspace, so instead we give the following algorithm that sandwiches an application of $\mathsf{Chal}_{b}$ between two applications of the symmetric subspace projector.  Formally, consider the following algorithm. 
    
    \begin{longfbox}[breakable=false, padding=1em, margin-top=1em, margin-bottom=1em]
    \begin{algo}\label{alg:preaction_distinguishing_given_cloning}
        Adversary for pre-action indistingiuishability given a two-register state, both with support on the same irreducible representation subspace $W^{\lambda}$.
    \end{algo}
    \noindent \textbf{Input}: Two quantum registers that are in valid money states for $\lambda$ and a query to the blackbox $\mathsf{Chal}_b^{h_1, h_2}$ given by \Cref{assum:preaction-indist}. 
    \begin{enumerate}
        \item Perform \textbf{\fse} on the two halves of the input.
        \item Perform a \textbf{swap test} between the two registers containing the archetype states produced.
        \item Uncompute the \textbf{\fse} on both halves of the state.
        \item Query $\mathsf{Chal}_b^{h_1, h_2}$ on the first register.
        \item Perform \textbf{\fse} on both halves of the state.
        \item Perform a second \textbf{swap test} between the two registers containing the archetype states produced.
        \item If the results of both the first and second swap tests agree, output $b' = 0$ (``no preaction'').
        \item If the results of the two swap tests disagree, output $b' = 1$ (``preaction'').
    \end{enumerate}
    \end{longfbox}


    \paragraph{Case 1: $b=0$ (there is no preaction).}
    We first claim that in the case where there is no preaction, the algorithm outputs ``no preaction'' with probability $1$.  In order to argue this, we analyze the case when the adversary measures the symmetric subspace in the first measurement \emph{after} performing the \fse, and argue that un-computing the \se, applying $\mathsf{Chal}^{h_1, h_2}_0$, and then performing \fse \emph{always} maps us back into the symmetric subspace on the first register.  In this part of the proof, all sums are over $[\dimlambda]$.

    Recall that the symmetric subspace is equal to the span of $\ket{\psi}^{\otimes 2}$, for $\ket{\psi} = \sum_{i, j} \alpha_{i} \ket{\phi^{\lambda}_{i}}$, so we can write the state after measuring the symmetric subspace as being in the span of:
    \begin{equation*}
        \left(\sum_{i, k} \alpha_{i} \alpha_{k} \ket{\phi^{\lambda}_{i}} \otimes \ket{\phi^{\lambda}_{k}}\right) \otimes \sum_{j, \ell} \beta_{j, \ell} \ket{\lambda, \lambda, j, \ell}\,.
    \end{equation*}
    Inverting the \fse, we get the following state
    \begin{equation*}
        \sum_{i, k, j, \ell} \alpha_i\alpha_k \beta_{j, \ell} \ket{\$^{\lambda}_{i, j}} \otimes \ket{\$^{\lambda}_{k, \ell}}\,.
    \end{equation*} 
     Then, after applying $\mathsf{Chal}^{h_1, h_2}_{0}$ to the first register of this state, we have the following.
     \begin{equation*}
        \sum_{i, j, k, \ell} \sum_{s} \alpha_{i}\alpha_{k}\beta_{j, \ell} \varrho_{\lambda}(h_1)_{s, j}\ket{\$^{\lambda}_{i, s}} \otimes \ket{\$^{\lambda}_{k, \ell}}\,.
     \end{equation*}
     Then after performing \se on both registers, we end up with the following state
     \begin{align*}
         &\sum_{i, j, k, \ell}\sum_{s} \alpha_i \alpha_k \beta_{j, \ell} \varrho_{\lambda}(h_1)_{s, j} \left(\ket{\phi^{\lambda}_i} \otimes \ket{\phi^{\lambda}_k}\right) \otimes \ket{\lambda, \lambda, s, \ell}\\
         &\hspace{15mm}= \sum_{i, k} \alpha_i \alpha_k \left(\ket{\phi^{\lambda}_i} \otimes \ket{\phi^{\lambda}_k}\right) \otimes \sum_{j, s, \ell} \beta_{j, \ell} \varrho_{\lambda}(h_1)_{s, j}\ket{\lambda, \lambda, s, \ell}\\
         &\hspace{15mm} = \sum_{i, k} \alpha_i \alpha_k \left(\ket{\phi^{\lambda}_i} \otimes \ket{\phi^{\lambda}_k}\right) \otimes \sum_{s, \ell} \sum_{j}(\beta_{j, \ell} \varrho_{\lambda}(h_1)_{s, j}) \ket{\lambda, \lambda, s, \ell}\,.
     \end{align*}
     Setting $\beta'_{s, \ell} = \sum_{j} \beta_{j, \ell} \varrho_{\lambda}(h)_{s, j}$, we get that we are still in the symmetric subspace within the first register. Since this applied to any setting of coefficients, the unitary transformation that composes steps 2, 3 and 4 preserves the symmetric and anti-symmetric subspaces.  Thus, if the first measurement has either outcome, the second measurement on step 6 will have the same outcome with probability $1$, and the adversary will output `no preaction' with probability $1$.
   
    \paragraph{Case 2: $b=1$ (there is a preaction).}
    We perform a similar analysis in the case where there is a pre-action, but now we will need to consider both subspaces.  This is because we need to prove that the unitary that the adversary implements in steps 3 through 5 maps \emph{every} vector from the symmetric subspace to something with high overlap with the anti-symmetric subspace, and vice versa.  Starting with the symmetric subspace, we have the same starting state after inverting the \fse.
    \begin{equation*}
        \sum_{i, k, j, \ell} \alpha_i\alpha_k \beta_{j, \ell} \ket{\$^{\lambda}_{i, j}} \otimes \ket{\$^{\lambda}_{k, \ell}}\,.
    \end{equation*} 
    After applying $\mathsf{Chal}^{h_{1}, h_{2}}_{1}$, we will end up with the following state
    \begin{equation*}
        \sum_{i, k, j, \ell, r, s} \alpha_i\alpha_k \beta_{j, \ell} \varrho_{\lambda}(h_2^{-1})_{i, r} \varrho_{\lambda}(h_1)_{s, j} \ket{\$^{\lambda}_{r, s}} \otimes \ket{\$^{\lambda}_{k, \ell}} \,.
    \end{equation*}
    After performing \fse, we end up with the following state.
    \begin{align*}
        &\sum_{i, k, j, \ell, r, s} \alpha_i\alpha_k \beta_{j, \ell} \varrho_{\lambda}(h_2^{-1})_{i, r} \varrho_{\lambda}(h_1)_{s, j} \left(\ket{\phi^{\lambda}_{r}} \otimes \ket{\phi^{\lambda}_{k}}\right) \ket{\lambda, \lambda, s, \ell}\\
        &\hspace{15mm}= \left(\sum_{r, k} \left(\sum_{i}\alpha_i \varrho_{\lambda}(h_2^{-1})_{i, r}\right)\alpha_k\ket{\phi_r^{\lambda}} \otimes \ket{\phi_{k}^{\lambda}}\right) \otimes \sum_{s, \ell} \left(\sum_{j} \beta_{j, \ell} \varrho_{\lambda}(h_1)_{s, j}\right)\ket{\lambda, \lambda, s, \ell}\\
        &\hspace{15mm}= \left(\sum_{r} \alpha'_r \ket{\phi^{\lambda}_{r}}\right) \otimes \left(\sum_{k} \alpha_k \ket{\phi^{\lambda}_{k}}\right) \otimes \sum_{s, \ell} \left(\sum_{j} \beta_{j, \ell} \varrho_{\lambda}(h_1)_{s, j}\right)\ket{\lambda, \lambda, s, \ell}\,.
    \end{align*}
    Here in the final line we define $\alpha'_r = \sum_{i} \alpha_i \varrho(h_2^{-1})_{i, r}$.  We can then write out the following expression for the inner product of the first two registers with their swap.
    \begin{align*}
        \mathrm{F}_{\mathrm{SWAP}}&=\left(\sum_{r} (\alpha'_r)^{\dagger} \bra{\phi^{\lambda}_{r}}\right) \otimes \left(\sum_{k} \alpha_k^{\dagger} \bra{\phi^{\lambda}_{k}}\right) \mathrm{SWAP}\left(\sum_{r'} \alpha'_{r'} \ket{\phi^{\lambda}_{r'}}\right) \otimes \left(\sum_{k'} \alpha_{k'} \ket{\phi^{\lambda}_{k'}}\right) \\
        &= \sum_{r, k}\left( (\alpha'_r)^{\dagger} \bra{\phi^{\lambda}_{r}}\right) \otimes \left(\alpha_k^{\dagger} \bra{\phi^{\lambda}_{k}}\right) \mathrm{SWAP}\left(\alpha'_{k} \ket{\phi^{\lambda}_{k}}\right) \otimes \left(\alpha_{r} \ket{\phi^{\lambda}_{r}}\right)\\ 
        &= \sum_{r, k} \left((\alpha'_r)^{\dagger}\alpha_k^{\dagger} \alpha_r \alpha'_{k}\right)\,.
    \end{align*}
    Now we analyze a single term in the sum.  Since $\alpha'$ itself is a sum of more elements, this will make the equations more manageable.
    \begin{align*}
        (\alpha'_r)^{\dagger}\alpha_k^{\dagger} \alpha_r^{\dagger} \alpha'_{k} &= \sum_{i, i'} \alpha_i^{\dagger} \varrho_{\lambda}(h_2^{-1})_{i, r}^{\dagger} \alpha_k^{\dagger} \alpha_r \alpha_i' \varrho_{\lambda}(h_2)^{-1}_{i', k}\\
        &= \alpha_k^{\dagger} \alpha_r \sum_{i, i'} \alpha_i^{\dagger} \alpha_{i'}  \varrho_{\lambda}(h_2^{-1})_{i, r}^{\dagger} \varrho_{\lambda}(h_2^{-1})_{i', k}\,.
    \end{align*}
    Now, computing an average over group elements and adding back in the sum over $r$ and $k$, we have the following:
    \begin{align*}
        \sum_{r, k} \alpha_k^{\dagger} \alpha_r \sum_{i, i'} \alpha_i^{\dagger} \alpha_{i'}\mathop{\mathbb{E}}_{h_2 \in G} \varrho_{\lambda}(h_2^{-1})^{\dagger}_{i, r} \varrho_{\lambda}(h_2^{-1})_{i', k} &= \frac{1}{\dimlambda}\left(\sum_{r} \alpha_r^{\dagger} \alpha_r\right) \left(\sum_{i} \alpha_i^{\dagger}\alpha_i\right)\\
        &= \frac{1}{\dimlambda}\,.
    \end{align*}
    Here we use the fact that $
    \braket{\$^{\lambda}_{a, b} 
    \,|\, 
    \$^{\lambda}_{c, d}} 
    = 
    \frac{\dimlambda}{\abs{\G}}
    \sum_{h \in \G}
    \varrho_{\lambda}(h^{-1})_{a,b}^*
    \varrho_{\lambda}(h^{-1})_{c,d}
    = \delta_{ac}\delta_{bd}
    $ (\Cref{sec:schur-orthogonality}) to cancel out the terms for which $r \neq k$ and $i \neq i'$, and then we use the fact that $\alpha_i$ come from a normalized quantum state.  To complete the proof, the probability that the swap test accepts on the state is given by
    \begin{equation*}
        \frac{1}{2}\left(1 + \mathrm{F}_{\mathrm{SWAP}}\right) = \frac{1}{2} + \frac{1}{2\dimlambda}\,.
    \end{equation*}
    
    This means that \emph{every} vector in the symmetric state gets mapped to a vector with overlap $1/2 + 1/2\dimlambda$ with the anti-symmetric state.  Thus, if the first swap test returned the symmetric subspace, the second one returns the symmetric subspace with this probability.  

    Now, we need to analyze the anti-symmetric subspace.  Similar to before, we take a basis for the anti-symmetric subspace and analyze what happens.  There is a simple basis described by the $\binom{\dimlambda}{2}$ vectors of the form
    \begin{equation*}
        \frac{1}{\sqrt{2}} \left(\ket{\$^{\lambda}_{i, j}} \otimes  \ket{\$^{\lambda}_{k, \ell}} - \ket{\$^{\lambda}_{k, j}} \otimes  \ket{\$^{\lambda}_{i, \ell}}\right)\,.
    \end{equation*}
    Going through the same steps, after applying $\mathsf{Chal}$, now with a pre-action, we have the following state
    \begin{equation*}
        \frac{1}{\sqrt{2}}\sum_{r, s} \left(\varrho_{\lambda}(h_2^{-1})_{i,r}\varrho_{\lambda}(h_1)_{s,j} \ket{\$^{\lambda}_{r, s}} \otimes \ket{\$^{\lambda}_{k, \ell}} - \varrho_{\lambda}(h_2^{-1})_{k,r}\varrho_{\lambda}(h_1)_{s,j} \ket{\$^{\lambda}_{r, s}} \otimes \ket{\$^{\lambda}_{i, \ell}}\right)\,.
    \end{equation*}
    Now we can examine the probability that a state starting from the symmetric subspace is still in the symmetric subspace (and that a state starting from the anti-symmetric subspace is still in the anti-symmetric subspace) after the \fse and swap test.  When we perform \fse, we have the following state
    \begin{align*}
        \ket{\psi_{i, j, k, \ell}} 
        &= 
        \frac{1}{\sqrt{2}}
        \sum_{r, s}
        \left(
            \varrho_{\lambda}(h_2^{-1})_{i,r}
            \varrho_{\lambda}(h_1)_{s,j} 
            \ket{\phi^{\lambda}_{r}} 
            \otimes 
            \ket{\phi^{\lambda}_{k}} 
            -
            \varrho_{\lambda}(h_2^{-1})_{k,r}
            \varrho_{\lambda}(h_1)_{s,j} 
            \ket{\phi^{\lambda}_{r}} 
            \otimes 
            \ket{\phi^{\lambda}_{i}} 
        \right)
        \otimes 
        \ket{\lambda, \lambda, s, \ell}
        \\
        &= 
        \frac{1}{\sqrt{2}}
        \sum_{r}
        \left(
            \varrho_{\lambda}(h_2^{-1})_{i,r}
            \ket{\phi^{\lambda}_{r}} 
            \otimes 
            \ket{\phi^{\lambda}_{k}} 
            -
            \varrho_{\lambda}(h_2^{-1})_{k,r}
            \ket{\phi^{\lambda}_{r}} 
            \otimes 
            \ket{\phi^{\lambda}_{i}} 
        \right)
        \otimes 
        \sum_{s}
        \varrho_{\lambda}(h_1)_{s,j} 
        \ket{\lambda, \lambda, s, \ell}
        \,.
    \end{align*}
    Since the operations up until now were unitary, we can write every state in the anti-symmetric subspace as a linear combination of vectors of this form.  $\sum_{i, j, k, \ell} \alpha_{i, j, k, \ell} \ket{\psi_{i, j, k, \ell}}$. 
    We need to compute the inner product between this state and the swapped version of this state, which we can compute as

    \begin{align*}
        &\mathop{\mathbb{E}}_{h_2 \in G}\left[\sum_{\substack{i, j, k, \ell \\ i', j', k', \ell'}}\alpha_{i, j, k, \ell}\alpha^{\dagger}_{i',j',k',\ell'}\bra{\psi_{i, j, k, \ell}}
        \mathrm{SWAP}
        \ket{\psi_{i', j', k', \ell'}}\right] \\
        &\hspace{15mm}= 
        \mathop{\mathbb{E}}_{h_2 \in G} \bigg[\frac{1}{2}
        \sum_{i,k, k', i'}\sum_{r, r'}
        \left(
            \varrho_{\lambda}(h_2^{-1})_{i',r}^*
            \bra{\phi^{\lambda}_{r}} 
            \otimes 
            \bra{\phi^{\lambda}_{k'}} 
            -
            \varrho_{\lambda}(h_2^{-1})_{k',r}^*
            \bra{\phi^{\lambda}_{r}} 
            \otimes 
            \bra{\phi^{\lambda}_{i'}} 
        \right)
        \\&\hspace{15mm}
        \qquad\qquad
        \left(
            \varrho_{\lambda}(h_2^{-1})_{i,r'}
            \ket{\phi^{\lambda}_{k}} 
            \otimes 
            \ket{\phi^{\lambda}_{r'}} 
            -
            \varrho_{\lambda}(h_2^{-1})_{k,r'}
            \ket{\phi^{\lambda}_{i}} 
            \otimes 
            \ket{\phi^{\lambda}_{r'}} 
        \right) \left(\sum_{j, \ell} \alpha_{i, j, k, \ell} \alpha^{\dagger}_{i', j, k', \ell}\right)\bigg]
        \\
        &\hspace{15mm}= 
        \mathop{\mathbb{E}}_{h_2 \in G}\bigg[ \frac{1}{2}
        \sum_{i, k, k', i'}\sum_{r, r'}
        \Big(
            \varrho_{\lambda}(h_2^{-1})_{i',r}^*
            \varrho_{\lambda}(h_2^{-1})_{i,r'}
            \braket{\phi^{\lambda}_{r} 
            \,|\,
            \phi^{\lambda}_{k}} 
            \braket{\phi^{\lambda}_{k'} 
            \,|\,
            \phi^{\lambda}_{r'}} 
            -
            \varrho_{\lambda}(h_2^{-1})_{i',r}^*
            \varrho_{\lambda}(h_2^{-1})_{k,r'}
            \braket{\phi^{\lambda}_{r} 
            \,|\,
            \phi^{\lambda}_{i}} 
            \braket{\phi^{\lambda}_{k'} 
            \,|\,
            \phi^{\lambda}_{r'}} 
            \\&\hspace{15mm}
            \qquad\qquad
            -
            \varrho_{\lambda}(h_2^{-1})_{k',r}^*
            \varrho_{\lambda}(h_2^{-1})_{i,r'}
            \braket{\phi^{\lambda}_{r} 
            \,|\,
            \phi^{\lambda}_{k}} 
            \braket{\phi^{\lambda}_{i'} 
            \,|\,
            \phi^{\lambda}_{r'}} 
            +
            \varrho_{\lambda}(h_2^{-1})_{k',r}^*
            \varrho_{\lambda}(h_2^{-1})_{k,r'}
            \braket{\phi^{\lambda}_{r} 
            \,|\,
            \phi^{\lambda}_{i}} 
            \braket{\phi^{\lambda}_{i'} 
            \,|\,
            \phi^{\lambda}_{r'}} 
        \Big) \left(\beta_{i,k,k',i'}\right)\bigg]
        \\
        &\hspace{15mm}= 
        \mathop{\mathbb{E}}_{h_2 \in G} \bigg[\frac{1}{2}\sum_{i, k, k', i'}
        \Big(
            \varrho_{\lambda}(h_2^{-1})_{i',k}^*
            \varrho_{\lambda}(h_2^{-1})_{i,k'}
            -
            \varrho_{\lambda}(h_2^{-1})_{i',i}^*
            \varrho_{\lambda}(h_2^{-1})_{k,k'}
            \\&\hspace{15mm}\qquad\qquad-
            \varrho_{\lambda}(h_2^{-1})_{k',k}^*
            \varrho_{\lambda}(h_2^{-1})_{i,i'}
            +
            \varrho_{\lambda}(h_2^{-1})_{k',i}^*
            \varrho_{\lambda}(h_2^{-1})_{k,i'}
        \Big)\beta_{i, k, k', i'}\bigg]
        \\
        &\hspace{15mm}= 
        \sum_{i, k}\frac{1}{2} \beta_{i, k, k, i}
        \Big(
            \mathop{\mathbb{E}}_{h_2 \in G}\left[\abs{
            \varrho_{\lambda}(h_2^{-1})_{i,k}
            }^2
            +
            \abs{
            \varrho_{\lambda}(h_2^{-1})_{k,i}
            }^2\right]
        \Big)
        \,.
    \end{align*}
    In the first equality, the swap only affects the first two registers, so the final two indices must be the same to survive the inner product. In the third equality, we use the fact that $\braket{\phi^{\lambda}_a \,|\, \phi^{\lambda}_b} = \delta_{ab}$. In getting to the final line, we use the fact that the $i$ indices are never equal to the $k$ indices, by the fact that we are in the anti-symmetric group.  Combining this with the fact from before that $\sum_{h \in G} \varrho_{\lambda}(h^{-1})^*_{a, b}\varrho_{\lambda}(h^{-1})_{c,d} = \delta_{c, d}\delta_{b, d}$, we can remove the two negative terms when averaging over the group elements.  Using the same fact, we have that for the remaining terms, we have
    \begin{align*}
            \frac{1}{\abs{\G}}
            \sum_{h_2 \in \G}
            \abs{
            \varrho_{\lambda}(h_2^{-1})_{i,k}
            }^2
            +
            \frac{1}{\abs{\G}}
            \sum_{h_2 \in \G}
            \abs{
            \varrho_{\lambda}(h_2^{-1})_{k,i}
            }^2 = \frac{2}{\dimlambda}\,.
    \end{align*}
    Since the $\beta_{i, k, k, i}$ sum to $1$ (as they are again the norm of the original vectors), the probability that the swap test succeeds on the second try is exactly
    \begin{align*}
        \frac{1}{2}
        +
        \frac{1}{2}
        \Big(
            \frac{1}{\dimlambda}
        \Big)\,.
    \end{align*}
    Now, we have shown that in the case when there is a pre-action, for \emph{all} states, the probability that the second swap test succeeds is given by 
    \begin{equation*}
        \frac{1}{2} + \frac{1}{2\dimlambda}\,.
    \end{equation*}
    Since the adversary accepts whenever the results are different, the adversary outputs ``preaction'' with probability at least
    \begin{equation*}
        \frac{1}{2} - \frac{1}{2\dimlambda}\,.
    \end{equation*}
    This is also the distinguishing advantage, as we showed that in the case where there is no pre-action, the adversary outputs ``no pre-action'' with probabiliy $1$.  

    If the adversary starts with a state that is $2\dimlambda\epsilon / (\dimlambda + 1)$ close to the tensor product of two copies of $W^{\lambda}$, they can first simply measure the irreducible representation label of both states, and condition on getting $\lambda$ for both run this test.  If the probability they measure $\lambda$ for both is at least the given probability, then their distinguishing advantage will be at least $\epsilon$. 
\end{proof}

\subsubsection{Generalizing to Intransitive Group Actions}\label{sec:intransitive}
Previously, we assumed that the group action was \emph{transitive}. That is, it had a single orbit, such that every element of the set $X$ can be reached from a single starting point $x \in X$.
In this subsection, we generalize to the case in which the group action is \textit{intransitive}. This means that the space is divided up into multiple orbits, with each orbit operating as a new \manifestation of the whole representation space. 

Note that the construction does not need to change for intransitive group actions. We can still have a fixed starting element $x$, whose orbit will be used by the minting algorithm to mint banknotes. 
However, for the proof, we can no longer assume that the two registers produced by the adversary have support on the same orbit\textemdash the orbit of $x$. 
The adversary may in general attempt to mint banknotes with supports on \emph{different orbits}.


We comment on how the security of the previous section generalizes to the intransitive setting.

\paragraph{Intransitive Preaction Security.}
We modify the definitions of preaction security to the intransitive case. Let $(\G, X, *, x)$ be an intransitive group action.

\begin{assumption}[Intransitive Preaction Hardness]
    \label{assum:preaction-is-hard-intran}
    Given $x$, $g*x$, and $h$ for a fixed starting element, $x \in X$, and random $g,h \gets \G$, it is hard to output $gh^{-1} * x$.  That is, there exists an $\epsilon > 0$ such that for all QPT adversaries, $\mathcal{A}$,
    \begin{align*}
        \Pr\left[ z = gh^{-1} * x \; : \; x \gets X, \; g, h \gets \G, \; z \gets \mathcal{A}(x, g*x, h) \right] \le \frac{1}{\abs{\G}} + \epsilon
    \end{align*}
\end{assumption}

\begin{assumption}[Intransitive Preaction Indistinguishability]
    \label{assum:preaction-indist-intran}
    It is hard to distinguish whether a preaction has been performed relative to a set of prefixed starting points.  
    Let $\mathcal{O}_1, \dots, \mathcal{O}_{m}$ be the orbits of the group action and let $x_1, \dots, x_m$ be representatives from each orbit ($x_i \in \mathcal{O}_i$).
    Let $\mathsf{Chal}_b^{h_1, h_2} \; : \; g*x_i \;\mapsto\; h_1 \, g \, h_2^{-b} * x_i$, \; for $b \in \bits$ and $h_1, h_2 \in \G$. Then there exists an $\epsilon > 0$ such that for all QPT adversaries, $\mathcal{A}$,
    that make a single query
    %
    to $\mathsf{Chal}_b^{h_1, h_2}$,
    \begin{align*}
        \Pr\left[ b' = b \; : \; h_1, h_2 \gets \G, \; b \gets \bits, \; \; b' \gets \mathcal{A}^{\mathsf{F}_b^{h_1, h_2}} \right] \le \frac12 + \epsilon
    \end{align*}
    Note that when $b = 0$, the challenger $\mathsf{Chal}_b^{h_1, h_2}$ performs a group action for a random group element $h_1$, and when $b = 1$, it performs both a random group action with $h_1$ and a random group pre-action with $h_2$.  
\end{assumption}

\begin{definition}
\label{def:preaction_secure_intractable}
    We say that a group action of group $\G_{n}$ on set $X_{n}$ with starting element $x$ is $\epsilon$-\emph{preaction secure} if both \Cref{assum:preaction-is-hard-intran} and \Cref{assum:preaction-indist-intran} hold for the group action against any QPT adversary with advantage $\epsilon$. We say that the group action is preaction secure if it is $\negl(n)$-preaction secure for any negligible function $\negl$.
\end{definition}

We also need an additional technical assumption, which says that it is hard to find ``bad'' orbits.

\begin{assumption}[Intractable bad irreps]\label{assum:bad-irreps} We say that a group action has $\delta$-\emph{intractable bad irreps} if any QPT adversary has probability at most $\delta$ of producing an $x\in X$ and an irreducible representation $\lambda$ such that (1) $\dimlambda>1$, but (2)$\lambda$ only has a single \manifestation in the representation of $G$ acting on the orbit $\mathcal{O}_i$ containing $x$.
\end{assumption}
Note that if all orbits $\mathcal{O}_i$ are in bijection with $G$, then the representation of $G$ acting on $\mathcal{O}_i$ will have $\dimlambda$ \manifestations of each irrep $\lambda$. However, if some orbit contains an element $x$ such that $g*x=x$ for some $g$, then the number of \manifestations of $\lambda$ may be smaller. \Cref{assum:bad-irreps} says that it is hard to find such an irrep and representative of such an orbit.

\begin{proposition}
    Suppose that the group action used in the quantum money construction 
    (\Cref{sec:preaction-secure-construction}) 
    is $\epsilon$-intrasitive preaction secure and has $\delta$-intractable bad irreps.  Then no QPT adversary can produce a quantum state on two registers such the probability of measuring both in the same irreducible subspace is greater than $2\epsilon + \delta$.  
\end{proposition}
\begin{proof}
    First, assume that the adversary does not sample a state in an intractable bad orbit.  Since the probability is upper bounded by $\delta$, this increase the probability that they measure a state in the same irreducible subspace by $\delta$.  

    Similar to the proof of \Cref{prop:unclonable_from_preaction_secure}, we assume for the sake of contradiction that the adversary has an $\delta$ probability of measuring two states in the same irreducible representation.  Then we consider the same algorithm, \Cref{alg:preaction_distinguishing_given_cloning}, for distinguishing a black-box that performs a pre-action from a black-box that does not perform a pre-action.  

    First, let $\ket{\$^{\lambda, x}_{i, j}}$ be the money state that corresponds to irrep label $\lambda$, \manifestation $i$, basis vector $j$, and starting element $x$.  Further let $\ket{\phi^{\lambda, x}_{i}}$ be the archetype state corresponding to irrep $\lambda$, \manifestation $i$ and starting element $x$.  
    \paragraph{Case 1: $b = 0$ (there is no preaction).} We begin by analyzing the performance of \Cref{alg:preaction_distinguishing_given_cloning} in the case when $b = 0$, first in the case where the symmetric subspace accepts and then the case when it fails.  Then we can write every state in the symmetric subspace as follows for some choice of $\alpha_{i}^{x}$ and $\beta_{j, \ell}$. 
    \begin{equation*}
        \sum_{i, k, x, y} \alpha_{i}^{x} \alpha_{k}^{y} \ket{\phi^{\lambda, x}_{i}} \otimes \ket{\phi^{\lambda, y}_{k}} \otimes \sum_{j, \ell} \beta_{j, \ell} \ket{\lambda, \lambda, j, \ell}\,.
    \end{equation*} 
    Here we note that this encompasses the case when the states span multiple orbits (indexed by starting elements $x$ and $y$).  Then after inverting the \fse, we get the following state
    \begin{equation*}
        \sum_{i, k, j, \ell, x, y} \alpha_{i}^{x} \alpha_{j}^{y} \beta_{j, \ell} \ket{\$^{\lambda, x}_{i, j}} \otimes \ket{\$^{\lambda, y}_{k, \ell}}\,.
    \end{equation*}
    Then after applying the black box (recall that $b = 0$) to the first register, we have the following
    \begin{equation*}
        \sum_{i, k, j, \ell, x, y} \sum_{s} \alpha_{i}^{x} \alpha_{j}^{y} \beta_{j, \ell} \varrho_{\lambda}(h_1)_{s, j} \ket{\$^{\lambda, x}_{i, s}} \otimes \ket{\$^{\lambda, y}_{k, \ell}}\,.
    \end{equation*}
    After performing \fse again, we get the following state, following the logic in \Cref{prop:unclonable_from_preaction_secure}.
    \begin{equation*}
        \sum_{i, k, x, y} \alpha^{x}_{i} \alpha^{y}_{k} \ket{\phi^{\lambda, x}_{i}} \otimes \ket{\phi^{\lambda, y}_{j}} \otimes \sum_{s, \ell, j} (\beta_{j, \ell} \varrho_{\lambda}(h_1)_{s, j}) \ket{\lambda, \lambda, s, \ell}\,.
    \end{equation*}
    Thus, we measure a state in the symmetric subspace.  Furthermore, since the symmetric subspace is perfectly mapped back to the symmetric subspace under the black box, if the first symmetric subspace measurement outputs the anti-symmetric subspace, the black box will keep the state in the anti-symmetric subspace.  Thus, in the case that the $b = 0$, the algorithm accepts with probability $1$.  

    \paragraph{Case 2: $b = 1$ (there is a preaction).}
    We first start in the case when the symmetric subspace projector accepts.  In this case, we can write the state after the projector accepts, similarly to before as
    \begin{equation*}
        \sum_{i, k, x, y} \alpha_{i}^{x} \alpha_{k}^{y} \ket{\phi^{\lambda, x}_{i}} \otimes \ket{\phi^{\lambda, y}_{k}} \otimes \sum_{j, \ell} \beta_{j, \ell} \ket{\lambda, \lambda, j, \ell}\,.
    \end{equation*}
    After applying the black box (this time with a pre-action), we have the following state
    \begin{align*}
        &\sum_{i, k, x, y, r, s} \alpha_{i}^{x} \alpha_{k}^{y} \beta_{j, \ell} \varrho_{\lambda}(h_2^{-1})_{i, r} \varrho_{\lambda}(h_1)_{s, j}\ket{\phi^{\lambda, x}_{r}} \otimes \ket{\phi^{\lambda, y}_{k}} \otimes \ket{\lambda, \lambda, s, \ell}\\
        &\hspace{15mm}= \sum_{r, k, x, y} \left(\sum_{i} \alpha_{i}^{x}\varrho_{\lambda}(h_2^{-1})_{i, r}\right) \alpha_{k}^{y} \ket{\phi^{\lambda, x}_{r}} \otimes \ket{\phi^{\lambda, y}_{k}} \otimes \sum_{s, \ell} \left(\sum_{j} \beta_{j, \ell} \varrho_{\lambda}(h_1)_{s, j}\right) \ket{\lambda, \lambda, s, \ell}\\
        &\hspace{15mm} \left(\sum_{r, x} \alpha'^{x}_{r} \ket{\phi^{\lambda, x}_{r}}\right) \otimes \left(\sum_{k, y} \alpha_{k}^{y} \ket{\phi^{\lambda, y}_{k}}\right) \otimes \sum_{s, \ell} \left(\sum_{j} \beta_{j, \ell} \varrho_{\lambda}(h_1)_{s, j}\right) \ket{\lambda, \lambda, s, \ell}\,.
    \end{align*}
    We can then write the expression for the fidelity of this state and the swapped version if the state as follows.
    \begin{equation*}
        \mathrm{F}_{\mathrm{SWAP}} = \sum_{r, k, x, y} \left(\left(\alpha'^{x}_{r}\right)^{\dagger} \left(\alpha_r^{x}\right)^{\dagger}\left(\alpha_k^{y}\right)^{\dagger} \left(\alpha_{k}'^{y}\right)\right)
    \end{equation*}
    Here the only difference from before is that the inner product also enforces that the orbits ($x$ and $y$) are the same between the left and right.  Expanding each $\alpha'$ as before, we get the following expression for the fidelity of the state with its swap.
    \begin{equation*}
        \sum_{r, k, x, y} \left(\alpha_{k}^{y}\right)^{\dagger} \left(\alpha_{r}^{x}\right)^{\dagger} \sum_{i, i'} \left(\alpha_{i}^{x}\right)^{\dagger} \alpha_{i'}^{y} \varrho_{\lambda}(h_2^{-1})_{i, r}^{\dagger} \varrho_{\lambda}(h_2^{-1})_{i', k}\,.
    \end{equation*}
    Computing the average over the group and applying \Cref{sec:schur-orthogonality}, we get the following quantity
    \begin{align*}
         &\sum_{r, k, x, y} \left(\alpha_{k}^{y}\right)^{\dagger} \left(\alpha_{r}^{x}\right)^{\dagger} \sum_{i, i'} \left(\alpha_{i}^{x}\right)^{\dagger} \alpha_{i'}^{y} \avg_{h_2 \in G} \left[\varrho_{\lambda}(h_2^{-1})^{\dagger}_{i, r} \varrho_{\lambda}(h_2^{-1})_{i', k}\right] \\
         &\hspace{15mm}= \frac{1}{\dimlambda} \sum_{x, y} \left(\sum_{r} (\alpha_{r}^{x})^{\dagger} \left(\alpha_{r}^{y}\right)\right) \left(\sum_{i} \left(\alpha_{i}^{x}\right)^{\dagger} \left(\alpha_{i}^{y}\right)\right)\\
         &\hspace{15mm}\leq \frac{1}{\dimlambda} \sqrt{\sum_{x, y} \left(\sum_{r} (\alpha_{r}^{x})^{\dagger} \left(\alpha_{r}^{y}\right)\right) \cdot \sum_{x, y} \left(\sum_{i} \left(\alpha_{i}^{x}\right)^{\dagger} \left(\alpha_{i}^{y}\right)\right)}\\
         &\hspace{15mm}= \frac{1}{\dimlambda}\,.
    \end{align*}
    Here after applying the Schur orthogonality rules, we apply Cauchy-Schwarz and then use the fact that both terms in the square roots are the norm of the original vector, so they are $1$.  To complete the proof, we note that the probability that the swap test succeeds is given by
    \begin{equation*}
        \frac{1}{2} (1 + \mathrm{F}_{\mathrm{SWAP}}) = \frac{1}{2} + \frac{1}{2\dimlambda}\,.
    \end{equation*}
    Now we proceed with the analysis in the case that the symmetric subspace measurements outputs the anti-symmetric subspace.  We can similarly write the anti-symmetric subspace on the first register as the span of the following vectors (and analyzing the action of the rest of the algorithm on those vectors will imply the action on every state in the anti-symmetric subspace).  
    \begin{equation*}
        \frac{1}{\sqrt{2}} \left(\ket{\$^{\lambda, x}_{i, j}} \otimes \ket{\$^{\lambda, y}_{k, \ell}} - \ket{\$^{\lambda, y}_{k, j}} \otimes \ket{\$^{\lambda, x}_{i, \ell}} \right)\,.
    \end{equation*}
    Here we require that $\delta_{x, y}\delta_{i, k} = 0$ (i.e. that at least one of the pairs is different). Similar to before we can write out the state after applying the black box (now with a pre-action), and then the \fse as follows
    \begin{align*}
        \ket{\psi_{i, j, k, \ell}^{x, y}} &= \sqrt{1}{2} \sum_{r, s} \left(\varrho_{\lambda}(h_2^{-1})_{i, r} \varrho_{\lambda}(h_1)_{s, j} \ket{\phi^{\lambda, x}_{r}} \otimes \ket{\phi^{\lambda, x}_{k}} - \varrho_{\lambda}(h_2^{-1})_{k, r} \varrho_{\lambda}(h_1)_{s, j} \ket{\phi^{\lambda, y}_{r}} \otimes \ket{\phi^{\lambda, x}_{i}}\right) \otimes \ket{\lambda, \lambda, s, \ell}\\
        &= \frac{1}{\sqrt{2}} \sum_{r} \left(\varrho_{\lambda}(h_2^{-1})_{i, r} \ket{\phi^{\lambda, x}_{r}} \otimes \ket{\phi^{\lambda, y}_{k}} - \varrho_{\lambda}(h_2^{-1})_{k, r} \ket{\phi^{\lambda, y}_{r}} \otimes \ket{\phi^{\lambda, x}_{i}}\right) \otimes \sum_{s} \varrho_{\lambda}(h_1)_{s, j} \ket{\lambda, \lambda, s, \ell}\,.
    \end{align*}
    Now, in a same fashion as before we can write every state in the anti-symmetric subspace as a linear combination of these basis vectors as $\sum_{i, j, k, \ell, x, y} \alpha_{i, j, k, \ell}^{x, y} \ket{\psi_{i, j, k, \ell}^{x, y}}$.  We then need to compute the inner product between this state and the state after swapping with itself, averaged over all group elements.  We get the following
    \begin{align*}
        &\avg_{h_2 \in G} \left[\sum_{\substack{i, j, k, \ell, x, y \\ i', j', k', \ell', x', y'}} \left(\alpha_{i,j, k, \ell}^{x, y}\right)^{\dagger} \left(\alpha_{i', j', k', \ell'}^{x', y'}\right) \bra{\psi_{i, j, k, \ell}^{x, y}} \mathrm{SWAP} \ket{\psi_{i',j',k',\ell'}^{x', y'}} \right]\\
       &\hspace{15mm}
       \begin{aligned}= &\avg_{h_2 \in G} \Bigg[ \frac{1}{2} \sum_{x, y, x', y'} \sum_{i, k, i', k'} \sum_{r, r'} \left(\varrho_{\lambda}(h_2^{-1})^{*}_{i', r} \bra{\phi^{\lambda, x}_{r}} \otimes\bra{\phi^{\lambda, y}_{k'}} - \varrho_{\lambda}(h_2^{-1})^{*}_{k', r} \bra{\phi^{\lambda, y}_{r}} \otimes \bra{\phi^{\lambda, x}_{i'}}\right) \\
       &\hspace{15mm}\left(\varrho_{\lambda}(h_2^{-1})_{i, r'} \ket{\phi^{\lambda, y'}_{k}} \otimes \ket{\phi^{\lambda, x'}_{r'}} - \varrho_{\lambda}(h_2^{-1})_{k, r'} \ket{\phi^{\lambda, x'}_{i}} \otimes \ket{\phi^{\lambda, y'}_{r'}}\right) \left(\sum_{j, \ell} \alpha_{i, j, k, \ell}^{x, y} \left(\alpha_{i', j', k', \ell'}^{x', y'}\right)^{\dagger}\right)\Bigg]
       \end{aligned}\\
       &\hspace{15mm}
       \begin{aligned}
        =&\avg_{h_2 \in G} \Bigg[\frac{1}{2} \sum_{x, y, x', y'} \sum_{i, k, i', k'} \sum_{r, r'} \bigg(\varrho_{\lambda}(h_2^{-1})_{i', r}^* \varrho_{\lambda}(h_2^{-1})_{i, r'}) \braket{\phi^{\lambda, x}_{r}|\phi^{\lambda, y'}_{k}} \braket{\phi^{\lambda, y}_{k'} | \phi^{\lambda, x'}_{r'}}\\
        &\hspace{15mm}-\varrho_{\lambda}(h_2^{-1})_{i', r}^* \varrho_{\lambda}(h_2^{-1})_{k, r'} \braket{\phi^{\lambda, x}_r | \phi^{\lambda, x'}_{i}} \braket{\phi^{\lambda, y}_{k'} | \phi^{\lambda, y'}_{r'}}\\
        &\hspace{15mm}-\varrho_{\lambda}(h_2^{-1})_{k', r}^* \varrho_{\lambda}(h_2^{-1})_{i, r'} \braket{\phi^{\lambda, y}_{r} | \phi^{\lambda, y'}_{k}}\braket{\phi^{\lambda, x}_{i'}|\phi^{\lambda, x'}_{r'}}\\
        &\hspace{15mm}+\varrho_{\lambda}(h_2^{-1})_{k', r}^* \varrho_{\lambda}(h_2^{-1})_{k, r'} \braket{\phi^{\lambda, y}_r|\phi^{\lambda, x'}_{i}}\braket{\phi^{\lambda, x}_{i'} | \phi^{\lambda, y'}_{r'}} \bigg)\left(\beta_{i, k, k', i'}^{x, y, x', y'}\right)\Bigg]\,.
        \end{aligned}\\
        &\hspace{15mm}\begin{aligned}=&\avg_{h_2 \in G} \Bigg[\sum_{x, y, x', y'} \sum_{i, k, i', k'} \bigg( (\varrho_{\lambda}(h_2^{-1})_{i', k})^*\varrho_{\lambda}(h_2^{-1})_{i, k'} \delta_{x, y'} \delta_{y, x'} - \varrho_{\lambda}(h_2^{-1})^*_{i', i} \varrho_{\lambda}(h_2^{-1})_{k, k'}\delta_{x, x'} \delta_{y, y'}\\
        &\hspace{15mm}- \varrho_{\lambda}(h_2^{-1})^*_{k', k} \varrho_{\lambda}(h_2^{-1})_{i, i'} \delta_{x, x'}\delta_{y, y'} + \varrho_{\lambda}(h_2^{-1})^*_{k', i} \varrho_{\lambda}(h_2^{-1})_{k, i'}\delta_{x, y'} \delta_{y, x'} \bigg) (\beta_{i, k, i', k'}^{x, y, x', y'})\Bigg]
        \end{aligned}\\
        &= \sum_{x, y} \sum_{i, k} \frac{1}{2} \beta_{i, k, k, i}^{x, y, y, x} \left(\avg_{h_2 \in G} \left[\left|\varrho_{\lambda}(h_2^{-1})_{i, k}\right|^2 + \left|\varrho_{\lambda}(h_2^{-1})_{k, i}\right|^2\right] \right)\,.
    \end{align*}
    In the first equality, we note that the swap only affects the first two registers, so the final two must be the same to survive the inner product, as before.  Then we use the fact that the inner product of the states $\braket{\phi^{\lambda, x}_{a} | \phi^{\lambda, y}_{b}} = \delta_{x, y} \delta_{a, b}$.  Finally, we use the fact that in the anti-symmetric states, we can not have both the orbits \emph{and} the \manifestations be the same (as noted in the description of the basis states).  This allows us to apply the Schur orthogonality relations and remove the two negative terms when we average over the group elements.

    Now, we can apply the same equality that we noted before to bound this by $\frac{2}{\dimlambda}$, again noting that the $\beta_{i, k, k, i}^{x, y, y, x}$ correspond to a normalized vector.  Thus, the probability that the swap test succeeds is bounded from above by the following
    \begin{equation*}
        \frac{1}{2} + \frac{1}{2\dimlambda}\,.
    \end{equation*}
    At this point, we have completed the analysis of the probability that the state passes the second test.  In particular, when there is no preaction, the preaction distinguisher outputs ``no preaction'' with probability $1$, and if there is a preaction the distinguisher outputs preaction with probability at least $\frac{1}{2} - \frac{1}{2\dimlambda}$.  Thus, if the adversary instead starts with a state that has probability $\dimlambda\epsilon / (\dimlambda + 1) \leq 2\epsilon$ of being measured in the tensor product of two copies of the irreducible space of $\lambda$ that are not intractable bad irreps, they can first measure the state and then apply this distinguisher to break the preaction indinstinguishability with probability $\epsilon$.  Adding in the probability ($\delta$) that the adversary measures a intractable bad irrep, we get the desired bound.
\end{proof}

With this proposition, we have shown that the construction of quantum lightning is secure if instantiated with a $\epsilon$-preaction secure group action (as in \Cref{def:preaction_secure_intractable}) that has $\negl$-intractable bad irreps.  In the next section, we will provide groups that might meet these conditions, providing the first instantiations of quantum lightning in the plain model from plausible assumptions.

\subsection{Instantiations of the Construction}
\label{sec:instantiations}

Here, we discuss some plausible instantiations of our quantum money scheme. Our main focus will be on \emph{symmetric} group actions. First, we note that symmetric group actions have a negligible maximum Plancherel measure~\cite{vershik1985asymptotic}, a necessary condition for having a secure quantum lightning scheme and for our pre-action security assumption to hold. Symmetric groups also admit an efficient quantum Fourier transform~\cite{beals1997quantum}, a necessary condition for the efficiency of our protocol. This makes symmetric group actions a natural candidate for instantiating our scheme.

\paragraph{Graph Isomorphism.} Given a graph $(V,E)$ with $|V|=n$, the symmetric group $S_n$ acts on $(V,E)$ by permuting the vertices. 

Note that the discrete logarithm problem on graphs is exactly the Graph Isomorphism problem. Graph Isomorphism can be solved in (classical) quasi-polynomial time~\cite{STOC:Babai16}. However, it is still conceivable that there is no polynomial-time algorithm, giving a plausible candidate group action.

We also would like $S_n$ to act regularly. If $(V,E)$ has a trivial automorphism group, then $S_n$ acts semiregularly on the orbit of $(V,E)$. ``Most'' graphs have trivial automorphism groups~\cite{LinMos17}. Unfortunately, it is in general presumably hard to identify the orbit of a graph $(V,E)$, since this would solve the graph isomorphism problem for $(V,E)$. We therefore appeal to our generalization to intransitive group actions in Section~\ref{sec:intransitive}. Therefore, even if there are multiple orbits, our security proof still works.

\paragraph{Permuting Matrices.} The symmetric group $S_n$ acts on the set of $n\times n$ symmetric matrices via $\sigma*M=\sigma\cdot M\cdot \sigma^T$. That is, permute the rows and columns of $M$ by $\sigma$. This is in fact a generalization of the graph isomorphism group action, using the adjacency matrix of the graph.

\paragraph{McEliece Cryptosystem.} The McEliece cryptosystem~\cite{Mceliece78} contains an implicit group action. Here, we have the symmetric group acting on matrices, though the operation is quite different. Let $\mathbb{F}$ be a field and $m>n$ be integers. Then consider the set $R_{n,m}$ of row-reduced matrices in $\mathbb{F}^{n\times m}$. Then $S_m$ acts on $R_{n,m}$, with $\sigma*M\rightarrow M'$ where $M'$ is the result of:
\begin{itemize}
    \item First permute the columns of $M$ according to $\sigma$.
    \item Then row-reduce the result.
\end{itemize}
Note that the McEliece cryptosystem uses the orbit of a specific matrix $M$ that has good error correcting properties. The original proposal in~\cite{Mceliece78} is to use binary Goppa codes. This original proposal has so far resisted (quantum) cryptanalysis efforts.

Note that we do not need any specific properties of $M$, allowing us to use basically any matrix $M$. Thus, even if the McEliece cryptosystem is broken, we still get a plausible quantum money candidate.

As for regularity, for a sufficiently wide matrix and/or sufficiently large field $\mathbb{F}$, $\S_m$ will act semiregularly on the orbit of ``most'' matrices, as shown in the lemma below. As with the Graph Isomorphism case, we do not expect to be able to identify the orbits of typical matrices, so we instead appeal to our generalization to non-transitive matrices.

\begin{lemma}Let $m\geq 2n+1$. Consider sampling $M$ such that (1) the left $n\times n$ matrix is the identity, and (2) the right $n\times (m-n)$ matrix is uniform. Then except with probability $p:=m^2 n^2 |\mathbb{F}|^{-1}+(m!) |\mathbb{F}|^{-n}$, $S_m$ will act semiregularly on the orbit of $M$. In particular, if $|\mathbb{F}|=m^{\omega(1)}$, then $p$ is negligible.
\end{lemma}

\begin{proof}Fix a permutation $\sigma\in S_m$ other than the identity. We bound the probability that $\sigma*M=M$.

Let us first suppose that the right $n\times(m-n)$ sub-matrix of $M$ contains all distinct and non-zero entries. Since the entries are uniform and independent, this occurs with probability at most $[mn(mn-1)+mn]|\mathbb{F}|^{-1}=m^2 n^2|\mathbb{F}|^{-1}$.

Now consider permuting the columns of $M$ according to $\sigma$. Denote the result by $M'$. Let $M''$ then be the result of row-reducing $M'$. We now consider three cases:
\begin{itemize}
    \item Suppose $\sigma(i)=i$ for all $i\leq n$, meaning $\sigma$ does not permute the first $n$ columns. In this case, $M''=M'$. Since the columns are distinct by our conditions on $M$ and $\sigma$ is not the identity, we have that $M'\neq M$. Thus, in this case $\sigma*M\neq M$. 
    \item $\sigma(i)>n$ for all $i>n$.  In other words, $\sigma$ separately permutes the first $n$ entries and permutes the remaining $m-n$ entries. In this case, $M'$ is obtained from $M$ by permuting the right $n\times(m-n)$ sub-matrix, and $M''$ is obtained from $M'$ by simply permuting some of the rows of $M'$. In other words, $M''$ is obtained from $M$ by permuting the rows and columns of the right $n\times(m-n)$ sub-matrix. As long as the entries of this sub-matrix are distinct, any such permutation of rows and columns will not preserve the matrix. 
    \item $\sigma(i)\leq n$ for some $i>n$. In this case, $M''\neq M'$. Let $D\in\mathbb{F}^{n\times n}$ be the matrix such that $M''=D^{-1}\cdot M'$. Then we know that $D$ is not the identity. 
\end{itemize}

We now focus on the last case above. If the first $n$ columns of $M'$ are not full rank, then $M''\neq M$. So suppose that the first $n$ columns of $M'$ are full rank. This means that $D$ is exactly the first $n$ columns of $M'$. If $M''=M$, we then have that $D\cdot M = M'$. In other words, if we take the first $n$ columns of $M'$ (which is just a permuted version of $M$), and multiply this with $M$, we get exactly $M'$. 

Moreover, since we are in the case $\sigma(i)\leq n$ for some $i>n$, this also means that $\sigma(j)>n$ for some $j\leq n$. Thus at least one of the columns of $D$ came from the right $n\times (m-n)$ sub-matrix of $M$.

Since $m\geq 2n+1$, there will be some column $i$ such that $i>n$ and $\sigma(i)>n$. In other words, this column is not among the first $n$ in either $M$ nor in $M'$. This column is therefore independent of $D$. Let $v$ denote the vector of elements in this column. 

For $M''=D\cdot M'$ to be equal to $M$, we need that $D\cdot v$ is among the original columns of $M$. There are two possibilities:
\begin{itemize}
    \item $\sigma(i)\neq i$. In this case, let $w$ be the $\sigma(i)$-th column of $M$. Then $M''=M$ implies  $D\cdots v=w$. Since $v$ is random and independent of $D,w$, this occurs with probability $|\mathbb{F}|^{-n}$.
    \item $\sigma(i)=i$. In this case, $M''=M$ implies that $D\cdots v=v$. Fix a $v$ such that all the entries of $v$ are non-zero. Since $v$ came from the right sub-matrix of $M$ and we are assuming all the entries there are non-zero, we can assume that $v$ satisfies this property. Now consider sampling $D$. $D$ contains some columns that are fixed (those that were originally among the first $n$ in $M$) and some that are random (those that were not among the first $n$ in $M$). Moreover, at least one of the columns is random, since one of the rows came from the right sub-matrix of $M$. Since $v$ is non-zero in all positions, it particular it is non-zero in some position corresponding to a random column of $D$. Thus, $D\cdot v$ is a random vector. The probability $D\cdot v=v$ over the choice of $D$ is therefore at most $|\mathbb{F}|^{-n}$. 
\end{itemize}

Therefore, the probability that there exists some $\sigma$ such that $\sigma*M=M$ is at most the sum of
\begin{itemize}
    \item The probability that the right $n\times (m-n)$ sub-matrix contains non-distinct entries, which is upper bounded by $m^2 n^2 |\mathbb{F}|^{-1}$
    \item For each $\sigma\in S_n$, $|\mathbb{F}|^{-n}$.
\end{itemize}
Thus, the overall probability that there exists some $\sigma$ is at most $m^2 n^2 |\mathbb{F}|^{-1}+m! |\mathbb{F}|^{-n}$.\end{proof}

\subsection{Dual-Mode One-way Homomorphisms}

\newcommand{\dmowh}{{dual-mode one-way homomorphism}\xspace}
\newcommand{\dmowhs}{{dual-mode one-way homomorphisms}\xspace}

In the previous sections, we gave a construction of quantum money/lightning when the group action is easy but its corresponding preaction is hard. In other words, for any group element $g$, encoded by the group action as $g*x$, we could only act on $g$ from the left to get $hg * x$, but not from the right to get $gh^{-1} * x$.
In this section, we explore how the construction of \cref{sec:qm-construction} changes if we explicitly allow acting on the encoded group element from \emph{both sides}. In this case, we have two different but related representations of the same group\textemdash one for the action and one for the preaction. One important difference is that this allows verification to recover both of the fine Fourier indices (compare with the hardness of recovering $i$ in \Cref{claim:right-fourier-measurement-hard}).

In fact, we will see that when we allow the encoding to be a \emph{homomorphism}, we get the surprising but useful property that four different notions of security are all identical, including the hardness of worst-case cloning, average cloning, minting a collision (i.e., breaking lightning security), and preparing the specific uniform superposition state corresponding to the trivial irrep.

We begin by giving a useful security definition for one-way homomorphisms:

\begin{definition}
    \label{def:dmowh}
    An injective (but not surjective) homomorphism $P: \G \to \bbH$%
    \footnote{
        Technically, it is a \emph{family} of homomorphisms, $\{P_{n}: \G_{n} \to \bbH_{n}\}_{n \in \N}$ but we omit the security parameter in the notation for succinctness.
    } 
    is a \emph{\dmowh} if there exists a fooling function $S: \G \to \bbH$%
    \footnote{
        We refer to $P$ and $S$ as the ``public'' and ``secret'' transformations, respectively. 
        $S$ may itself be a homomorphism but it need not be.
        Notice that we do not require $S$ to be efficiently computable. 
    }
    such that they satisfy the following properties:

    \begin{itemize}
        \item \textbf{Efficiency:} There is a QPT algorithm to efficiently compute $P$. There are also efficient QPT algorithms for computing the group operations on $\G$ and $\bbH$.%
        \footnote{
            In general, the algorithms for computing $P$ and the group operations need only be approximately correct. 
            Moreover, because of the inaccessibility condition, we only the algorithm for group operations on $\bbH$ to be correct and homomorphic on the image~of~$P$.
        }
        \item \textbf{Statistical Distance:} Let $\bbH_P$ be the image of $P$ and $\bbH_S$ be the image of $S$. Then
        $\bbH_S$ is sufficiently far from $\bbH_P$.
        Specifically we require that,%
        \footnote{
            This condition prevents $P$ being a fooling function for itself.
            Note that it is equivalent to the condition that 
            $
                \Pr
                [
                    S(h) \notin \bbH_{P} 
                    \;|\;
                    h \gets \G
                ]
                \ge
                \frac{1}{\poly(n)}
            $. Or in other words, that ${\abs{\bbH_S \setminus \bbH_P}}/{\abs{\bbH_S}} \ge \frac{1}{\poly(n)}$ in the case where $S$ is injective (which it need not be).
            The reason we prefer to write the statistical distance condition in this way is that if the promised algorithm, $A$, for group operations on $\bbH$ is randomized, we can take the probability to also be over this randomness.
            \begin{align*}
                \Pr
                [
                    u \in \bbH_{P} 
                    \;|\;
                    u \gets A(P(g),S(h)),
                    \;
                    g, h \gets \G
                ]
                \le
                1 - 
                \frac{1}{\poly(n)}
            \end{align*}
        }
        \begin{align*}
            \Pr
            [
                P(g)S(h)
                \in \bbH_{P} 
                \;|\;
                g, h \gets \G
            ]
            \le
            1 - 
            \frac{1}{\poly(n)}
        \end{align*}
        
        \item \textbf{Indistinguishability:} It is hard to distinguish the images of $P$ and $S$. Formally, for all QPT adversaries $A$, 
        \begin{align*}
            \Pr\left[
                A(h) = b 
                \;\middle\vert\; 
                \begin{array}{cc}
                    b \gets \bits \quad 
                    g \gets \G \\ 
                    h \gets \begin{cases}
                        P(g) \quad b = 0 \\
                        S(g) \quad b = 1
                    \end{cases}
                \end{array}
            \right]
            \le \frac12 + \negl(n)
        \end{align*}
        \item \textbf{Inaccessibility:} It is hard to sample an element of $\bbH \setminus \bbH_P$.
        Formally, for all QPT adversaries $A$,
        \begin{align*}
            \Pr\left[
                h \in \bbH \setminus \bbH_P
                \;\middle\vert\;
                h \gets A(1^{\lambda})
            \right] \le \negl(n)\\ 
            \barak{I wonder if this can be relaxed to $\frac{1}{\poly(n)}$}
        \end{align*}
    \end{itemize}
\end{definition}

\begin{remark}
    Note that while we do not explicitly require $P$ to be one-way, this is implied by the definition: any adversary for inverting $P$ can be used to break the indistinguishability security.
\end{remark}
\begin{remark}
    Combined with statistical distance, inaccessibility provides that the fooling function $S$ is hard to compute even in the forward direction. In a cryptographic implementation we would sample a key pair of a public key $\pk$ and secret key $\sk$, which would allow computing $P$ and $S$, respectively (though we omit this in the definition for generality and simplicity).
    In other words, in the security game, $S$ is a ``secret'' function that only the challenger knows. This is why we call it a \dmowh: the is a public mode $P$ that is publicly computable, and a private mode $S$ that is only privately computable but mimics $P$ to the public. 
\end{remark}

\begin{observation}
    We can build a plausible candidate \dmowh from any group action for which the computational Diffie-Hellman problem (CDH) is easy but the Discrete Logarithm problem (DLog) is hard.%
    \footnote{
        Note that such a group action is only possible for non-Abelian groups, since CDH and DLog are known to be computationally equivalent for Abelian groups~\cite{montgomery2022full}, further demonstrating the necessity of non-Abelian-ness in our generalizations.
    } Note that cryptographers typically would like CDH to be hard, since it allows for justifying the security of more varied protocols. Our construction  therefore gives a ``win-win'' result, showing that in any group action for which DLog is hard, either (1) the more useful CDH problem is actually also hard, or (2) we obtain a plausible candidate \dmowh and hence a plausible quantum lightning scheme based on DLog.  This win-win result is remeniscent of win-win results in~\cite{zhandry2021quantum}, though the details are entirely different.
\end{observation}

\begin{proof}[Main idea]
    We give the informal idea here 
    but we leave a formal construction of \dmowhs from group actions to future work.
    Suppose we have a group $\G$ which acts on set $X$. Assume that the CDH problem is easy, and let $A$ be the CDH adversary, which takes two elements $a*x, b*y \in X$ and outputs $ab * x$ if $x=y$, and behaves arbitrarily if $x \ne y$. We set $\bbH$ to be the set $X$ with the ``group operation'' defined by $A$.%
    \footnote{Since the adversary may act arbitrarily when $x \ne y$, this is not exactly a group operation. That is, it only defines a group operation on elements within the same orbit, but this will be sufficient for our purposes.}
    Let $x, y \in X$ be two set elements in different orbits of the group $\G$, and let $P(g) = g * x$ and $S(g) = g * y$.
    Statistical distance comes from the fact that $x$ and $y$ are in different orbits.%
    \footnote{
        Depending on the CDH adversary,
        the roles of $x$ and $y$ may need to be reversed in order to formally satisfy the statistical distance property. If the CDH adversary, when run on a random element of the orbit of $x$ and a random element of the orbit of $y$, is more likely to produce an element of the orbit of $x$ then we switch the roles. In other words, we set $P(g) = g * y$ and $S(g) = g * x$. In any case, one of the two choices suffices.
    }
    Indistinguishability would come from the hardness of deciding if two elements are in the same orbit (one example being the hardness of the graph isomorphism problem). Inaccessibility would arise from the hardness of sampling a valid set element outside the orbit of any known elements.%
    \footnote{This can be argued in the generic group model~\cite{ITCS:Zhandry24a}.}
    Although we argue that these are reasonable assumptions, we do not know if any specific instantiations of group actions satisfy these requirements. We leave finding concrete instantiations of \dmowhs to future work.
    \todo{Formalize this.}
\end{proof}

The lack of concrete instantiation of a \dmowh is a certainly disadvantage (as opposed to our construction from preaction security in \Cref{sec:preaction-secure-construction}, to which we give concrete candidate instatiations in \Cref{sec:instantiations}). 
However, we believe that our construction from \dmowh is interesting in its own right. 
For instance, we have the property that four different security definitions\textemdash including average-case and worst-case cloning, as well as quantum lightning security\textemdash are all equivalent.
To the best of our knowledge, this is the first plausible quantum money construction to have this useful property.

\subsubsection{Quantum Money Construction}
\label{sec:qmoney-from-dmowh-construction}

Let $\G$ be a group satisfying the requirements in \Cref{sec:qm-construction}, and let $(P, S)$ be a \dmowh from $\G$ to $\bbH$.
We build a quantum lightning scheme, 
$(\mint, \ver)$ as follows:

\paragraph{$\mint(1^n) \to \left(\sigma, \; \ket{\$^{\sigma}}\right)$:} 
Consider the group action of $\G$ on $\bbH$ that comes from left-multiplying an element $h \in \bbH$ by the image of a group element $g \in \G$ under $P$, with the starting element $x \in \bbH$ being the identity element of $\bbH$. That is, we have the group action $g * y \mapsto P(g) \, y$ for any $y \in \bbH$.
Observe that because our starting element is the identity of $\bbH$, we have an efficiently computable preaction as well, by multiplying in the same way on the right.

Minting follows from the construction in \Cref{sec:qm-construction}, and produces a serial number $\sigma \gets \varrho$ denoting the measured irrep, and quantum money state  
$\ket{\$^\lambda} = \sum_{i,j \in [\dimlambda]} \alpha_{i, j} \ket{\$^{\lambda}_{i, j}}$, where

\begin{align*}
    \ket{\$^{\lambda}_{i, j}} := \sqrt{\frac{\dimlambda}{|\G|}}\sum_{h \in \G} \varrho_{\lambda}\left(h^{-1}\right)_{i, j} \; \big|\; P(h) \;\big\rangle
\end{align*}

\paragraph{$\ver(\sigma, \ket{\pounds}) \to \{\textnormal{accept, reject}\}$:} We follow the framework in \Cref{sec:qm-construction} to verify under the two group actions (the action and the preaction) consisting of the left and right group operations on the encoded element.
Note that within the image of $P$, this verification accepts any state of the original minted form, as well as states that are of that form, but that are shifted by an element of the center of $G$. In the security analysis, we show that these are the only states that pass verification.

\subsubsection{Security Analysis}

\begin{theorem}
    If $(P, S)$ is a secure \dmowh for $\G_n$, then $(\mint, \ver)$ is a secure quantum lightning scheme.
\end{theorem}
\begin{proof}
    Let $C$ be an adversary for the quantum lightning scheme, which outputs a state on two registers which both pass verification for the same serial number $\varrho$. We will show that it can be used to break the \dmowh.
    
    Specifically, we will show that we can use $C$ to break either the inaccessibility security or the indistinguishability security of the \dmowh.

    \begin{claim}
        We can assume without loss of generality that the output of $C$ is a tensor product and that the states both have the form 
        $\ket{\$^{\lambda}_{i, j}}$, where  
        \begin{align*}
            \ket{\$^{\lambda}_{i, j}} := \sqrt{\frac{\dimlambda}{|\G|}}\sum_{g \in \G} \varrho_{\lambda}\left(g^{-1}\right)_{i, j} \; \big|\; P(g) \;\big\rangle
        \end{align*}
        for some $i,j \in [\dimlambda]$. 
    \end{claim}
    \begin{proof}
        We begin by observing that if the quantum money state had non-negligible support on $\bbH \setminus \bbH_{P}$, then we can measure to get an element outside of the image of $P$ and break the inaccessibility security of the \dmowh.\todo{Formalize this.} Therefore, up to negligible error, we can assume that they both have support only on the image of $P$. Furthermore, 
        we can assume that both have the same $i$ and $j$, since we can perform a \fse twice on each state\textemdash both on the action and on the preaction\textemdash to get the corresponding $i$ and $j$, and then change them to match. This gives us a tensor product state $\ket{\$^{\lambda}_{i, j}} \otimes \ket{\$^{\lambda}_{i, j}}$.
    \end{proof}

    We now show how to break indistinguishability security. We consider one of the copies, setting aside the other copy for now.

    Suppose we get as input an element $z \in \bbH$ that is either a in the image of $P$, that is $z = P(h)$, or the image of $S$, $z = S(h)$, for group element $h \in \G_{\lambda}$.
    We left-multiply the quantum money state by $z$. If it is in the image of $P$, then we get
    \begin{align*}
        z \cdot \ket{\$^{\lambda}_{i, j}} 
        &= 
        \sqrt{\frac{\dimlambda}{|\G|}}\sum_{g \in \G} \varrho_{\lambda}\left(g^{-1}\right)_{i, j} \; \big|\; P(h) P(g) \;\big\rangle
        \\
        &= 
        \sqrt{\frac{\dimlambda}{|\G|}}\sum_{g \in \G} \varrho_{\lambda}\left(g^{-1}\right)_{i, j} \; \big|\; P(hg) \;\big\rangle
        \\
        &= 
        \sqrt{\frac{\dimlambda}{|\G|}}\sum_{g \in \G} \varrho_{\lambda}\left(g^{-1} h \right)_{i, j} \; \big|\; P(g) \;\big\rangle
        \\
        &= 
        \sqrt{\frac{\dimlambda}{|\G|}}\sum_{g \in \G, k \in [\dimlambda]} \varrho_{\lambda}\left(g^{-1}\right)_{i, k} \varrho_{\lambda}\left(h \right)_{k, j} \; \big|\; P(g) \;\big\rangle
        \\
        &= 
        \sum_{k \in [\dimlambda]} \varrho_{\lambda}\left(h \right)_{k, j} \sqrt{\frac{\dimlambda}{|\G|}}\sum_{g \in \G} \varrho_{\lambda}\left(g^{-1}\right)_{i, k} \; \big|\; P(g) \;\big\rangle
        \\
        &= 
        \sum_{k \in [\dimlambda]} \varrho_{\lambda}\left(h \right)_{k, j} \ket{\$^{\lambda}_{i, k}}
    \end{align*}
    If it is in the image of $S$, then we similarly get 
    \begin{align*}
        z * \ket{\$^{\lambda}_{i, j}} 
        &= 
        \sqrt{\frac{\dimlambda}{|\G|}}\sum_{g \in \G} \varrho_{\lambda}\left(g^{-1}\right)_{i, j} \; \big|\; S(h) * P(g) \;\big\rangle
        \\
        &= 
        \sqrt{\frac{\dimlambda}{|\G|}}\sum_{g \in \G} \varrho_{\lambda}\left(g^{-1}\right)_{i, j} \; \big|\; \tilde{S}(hg) \;\big\rangle
        \\
        &= 
        \sum_{k \in [\dimlambda]} \varrho_{\lambda}\left(h \right)_{k, j} \sqrt{\frac{\dimlambda}{|\G|}}\sum_{g \in \G} \varrho_{\lambda}\left(g^{-1}\right)_{i, k} \; \big|\; \tilde{S}(g) \;\big\rangle
    \end{align*}
    where $\tilde{S}$ is some function implied by $S$ that is guaranteed by the statistical distance property of \Cref{def:dmowh} to have at least inverse polynomial support outside of the image of $P$.%
    \footnote{
        Note that this does not break inaccessibility security, since in this case we are given $z$ which is itself already outside the image of $P$. 
    }
        
    We finally perform a \fse and swap test with the copy that was set aside. If $z$ was in the image of $P$, then the swap test will certainly pass. Otherwise, we observed that the two tested states will have orthogonal support that is at least inverse polynomial (since $\tilde{S}$ is far from $P$), and the swap test will fail with probability $1 - \frac{1}{\poly(n)}$. This therefore breaks the indistinguishability security of the \dmowh.
\end{proof}

\subsubsection{Worst-case to Average-case Reduction for Cloning}
Remarkably, the problem of cloning any specific (worst-case) money state in this construction can be reduced to that of producing two copies of an an average case money state, and therefore to cloning an average-case state. Moreover, all of these are equivalent to the problem of preparing the trivial irrep state (that is, the positive uniform superposition over the image of $P$).

\needspace{6\baselineskip}
\begin{theorem}[Worst-case to Average-case Cloning Reduction and Money/Lightning Equivalence]\label{thm:worst-average-cloning}\phantom{ }
    For the quantum money/lightning scheme defined in \Cref{sec:qmoney-from-dmowh-construction}, the following are equivalent:
    \begin{enumerate}
        \item\label{item:worst-cloning}
            There exists an efficient worst-case cloner that clones all valid money states $\ket{\$_{i, j}^{\lambda}}$.
        \item\label{item:avg-cloning}
            There exists an efficient average-case cloner that clones an average-case money state $\ket{\$_{i, j}^{\lambda}}$, where $\lambda$ is sampled according the the Plancherel measure of the group.
        \item\label{item:lightning}
            There exists an efficient lightning adversary that produces two copies of the same money state $\ket{\$_{i, j}^{\lambda}}$, where $\lambda$ is sampled according the the Plancherel measure of the group.
        \item\label{item:prep-triv-irrep}
            There exists an efficient preparation device that prepares the trivial irrep state $\ket{\$^{\id}}$, that is, the positive uniform superposition over image of the homomorphism $P$.
    \end{enumerate}
\end{theorem}
In other words, all four tasks (%
    worst-case cloning, 
    average-case cloning, 
    sampling state doublets, 
    and 
    trivial irrep state preparation%
) are computationally equivalent.
\begin{proof}
    It can be seen directly that \ref{item:worst-cloning} $\Rightarrow$ \ref{item:avg-cloning} (since cloning in the worst case trivially implies doing so in the average case), and that \ref{item:avg-cloning} $\Rightarrow$ \ref{item:lightning} (using the $\mint$ function to mint a state and then using the cloner to clone it). So it remains to show that \ref{item:lightning} $\Rightarrow$ \ref{item:prep-triv-irrep} and that \ref{item:prep-triv-irrep} $\Rightarrow$ \ref{item:worst-cloning}.
    We start by showing that \ref{item:prep-triv-irrep} $\Rightarrow$ \ref{item:worst-cloning}, and then \ref{item:lightning} $\Rightarrow$ \ref{item:prep-triv-irrep} will follow from applying the same process in reverse on the doublet produced by the lightning adversary.

    Suppose that we had a quantum money state with irrep label $\lambda$ 
    that we would like to clone:
    \begin{align*}
        \ket{\$_{i, j}^{\lambda}} = \sqrt{\frac{\dimlambda}{|G|}}\sum_{h \in G} \varrho_{\lambda}\left(h^{-1}\right)_{i, j} \; \big|\; P(h) \;\big\rangle
    \end{align*}
    We run the trivial irrep state preparation adversary to prepare the positive uniform superposition over the image of $P$: 
    \begin{align*}
        \ket{\$^{\id}} = \frac{1}{\sqrt{|G|}}\sum_{g \in G} \big|\; P(g) \;\big\rangle
    \end{align*}
    
    We left multiply the first register (the money state) by the inverse of the second register (the trivial irrep state), producing
    \begin{align*}
        \ket{\$_{i, j}^{\lambda}} \otimes \ket{\$^{\id}}
        &\rightarrow
        \sqrt{\frac{\dimlambda}{|G|}}\sum_{h \in G} \varrho_{\lambda}\left(h^{-1}\right)_{i, j} \; \big|\; P(g^{-1}h) \;\big\rangle
        \otimes
        \frac{1}{\sqrt{|G|}}\sum_{g \in G} \big|\; P(g) \;\big\rangle
        \\
        &=
        \sqrt{\frac{\dimlambda}{|G|}}\sum_{h \in G} \varrho_{\lambda}\left(h^{-1}g^{-1}\right)_{i, j} \; \big|\; P(h) \;\big\rangle
        \otimes
        \frac{1}{\sqrt{|G|}}\sum_{g \in G} \big|\; P(g) \;\big\rangle
        \\
        &=
        \sqrt{\frac{\dimlambda}{|G|}}\sum_{k \in [\dimlambda]}\sum_{h \in G} \varrho_{\lambda}\left(h^{-1}\right)_{i, k}\varrho_{\lambda}\left(g^{-1}\right)_{k, j} \; \big|\; P(h) \;\big\rangle
        \otimes
        \frac{1}{\sqrt{|G|}}\sum_{g \in G} \big|\; P(g) \;\big\rangle
        \\
        &=
        \frac{1}{\sqrt{\dimlambda}} \sum_{k \in [\dimlambda]}
        \sqrt{\frac{\dimlambda}{|G|}}\sum_{h \in G} \varrho_{\lambda}\left(h^{-1}\right)_{i, k} \; \big|\; P(h) \;\big\rangle
        \otimes
        \sqrt{\frac{\dimlambda}{|G|}}\sum_{g \in G} \varrho_{\lambda}\left(g^{-1}\right)_{k, j} \big|\; P(g) \;\big\rangle
        \\
        &=
        \frac{1}{\sqrt{\dimlambda}} \sum_{k \in [\dimlambda]}
        \ket{\$_{i, k}^{\lambda}}
        \otimes
        \ket{\$_{k, j}^{\lambda}} 
        \addtag\label{eq:epr-money-doublet}
    \end{align*}
    We now have two states that are both valid quantum money states for irrep label $\lambda$.

    \begin{observation}
    \label{sec:cloning_exact_from_extraction}
    If we would like both registers to be exact copies of the original state in tensor product, we can do that as well. 
    \end{observation}

    \begin{proof}[Proof of \Cref{sec:cloning_exact_from_extraction}.]
    We apply a left Fourier measurement on the left register (that is, a Fourier measurement corresponding to left action by plaintext group elements) and a right Fourier measurement on the right register (corresponding to right action), to produce
    \begin{align*}
        &
        \frac{1}{{\dimlambda}^{3/2}} 
        \sum_{k, \ell, m \in [\dimlambda]}
        \ket{\$_{i, \ell}^{\lambda}}
        \otimes
        \ket{L_{\ell, k}^{\lambda}}
        \otimes
        \ket{L_{k,m}^{\lambda}}
        \otimes
        \ket{\$_{m,j}^{\lambda}}
        \\
        \xrightarrow[]{\text{QFT}}&
        \frac{1}{{\dimlambda}^{3/2}} 
        \sum_{k, \ell, m \in [\dimlambda]}
        \ket{\$_{i, \ell}^{\lambda}}
        \otimes
        \ket{\lambda, \ell, k}
        \otimes
        \ket{\lambda, k, m}
        \otimes
        \ket{\$_{m,j}^{\lambda}}
    \end{align*}
    Now note that the registers containing $k$ are in the pure state $\frac{1}{\sqrt{\dimlambda}} \sum_{k \in [\dimlambda]} \ket{k}\ket{k}$, in tensor product with the rest of the state. Replace these registers with the state $\ket{j} \ket{i}$ to get 
    \begin{align*}
        &
        \frac{1}{{\dimlambda}} 
        \sum_{\ell, m \in [\dimlambda]}
        \ket{\$_{i, \ell}^{\lambda}}
        \otimes
        \ket{\lambda, \ell, j}
        \otimes
        \ket{\lambda, i, m}
        \otimes
        \ket{\$_{m, j}^{\lambda}}
    \end{align*}
    Now uncompute the two {\fse}s we have just performed to get $\ket{\$_{i, j}^{\lambda}} \otimes \ket{\$_{i, j}^{\lambda}}$ as desired.
    \end{proof}

    We now continue with the proof of \Cref{thm:worst-average-cloning} and show that \ref{item:lightning} $\Rightarrow$ \ref{item:prep-triv-irrep}. Given a doublet pair of quantum money states for the same irrep label, we show how to prepare the trivial irrep state $\ket{\$^{\id}} = \frac{1}{\sqrt{|G|}}\sum_{g \in G} \big|\; P(g) \;\big\rangle$.
    This doublet produced by the lightning adversary will be a state on two registers that passes verification, of the form%
    \footnote{This is assuming the inaccessibility security of the \dmowh, which is what prevents the state from having non-negligible support outside the image of $P$.}
    \begin{align*}
        \sum_{i,j,k,\ell \in [\dimlambda]}
        \alpha_{i,j,k,\ell} \;
        \ket{\$_{i, j}^{\lambda}}
        \otimes
        \ket{\$_{k, \ell}^{\lambda}}
    \end{align*}
    As above, we can use {\fse}s to convert this to the state $\ket{\$_{i, j}^{\lambda}} \otimes \ket{\$_{i, j}^{\lambda}}$ or even to the state $\frac{1}{\sqrt{\dimlambda}} \sum_{k \in [\dimlambda]}\ket{\$_{i, k}^{\lambda}} \otimes \ket{\$_{k, j}^{\lambda}}$. By reversing the process described above, from \Cref{eq:epr-money-doublet} backwards, we then recover the state $\ket{\$_{i, j}^{\lambda}} \otimes \ket{\$^{\id}}$, where the second register is the trivial irrep state, as desired.   
\end{proof}

So cloning an average-case state is as hard as cloning a worst case state, both of which are as hard as preparing the positive uniform superposition over the image of the homomorphism (the trivial irrep state). Moreover, quantum money security and quantum lightning security are equivalent for this scheme.

\begin{remark}
    Note that a task that is absent from \Cref{thm:worst-average-cloning} is the ability to prepare any quantum money state given its irrep label (ie. its serial number). An adversary for this task would clearly imply one for all four of the tasks mentioned in \Cref{thm:worst-average-cloning}, but it is not clear if the opposite is true. That is, an adversary that breaks the quantum money/lightning scheme nevertheless might not be able to prepare specific money states on command (only at random).
    In section \Cref{sec:qfire}, we take advantage of precisely this gap to propose a new quantum cryptographic primitive: quantum \emph{fire}.
\end{remark}

\section{Quantum Fire: Quantum States that are Clonable but Untelegraphable}
\label{sec:qfire}
In this section, we introduce a new quantum cryptographic primitive, ``quantum fire'', a cryptographic version of the clonable-but-untelegraphable states introduced by~\cite{ITCS:NZ23}. 
Much like fire is an entropic state of matter that is hard to spark on command, but easy to spread around as long as it is kept alive, \emph{quantum fire} is a quantum state that is hard to prepare but easy to clone as long as it is maintained in coherent quantum form. 
More specifically, a \emph{quantum fire state}, $\ket{\phi_i}$, comes from a collection $\{\ket{\phi_i}\}_i$ of states that 
\begin{itemize}
    \item \textbf{Efficiently sparked: }
    there is an efficient sparking algorithm that outputs a random $\ket{\phi_i}$, from some distribution over $i$,
    \item \textbf{Efficiently clonable: } there is an efficient cloner that maps one copy of $\ket{\phi_i}$ to two copies, 
    \item \textbf{Un-telegraphable}: no efficient adversary can encode $\ket{\phi_i}$ into a classical string that can later be revived back into $\ket{\phi_i}$.
\end{itemize}
We also allow an efficiently verifiable version, in which we have the additional property,
\begin{itemize}
    \item \textbf{Verifiable: }
    there is a verification algorithm that takes a label $i$ as well as a claimed state $\ket{\phi_i'}$, and outputs whether $\ket{\phi_i'}$ is a valid quantum fire state.%
    \footnote{In the general case, the verification algorithm for the quantum fire state may allow a larger space of states than those that would be produced by the sparking algorithm. In this case the cloning algorithm and the telegraphing adversary must output any state(s) that pass verification.}
\end{itemize}
Quantum fire has been demonstrated to exist relative to a quantum oracle in~\cite{ITCS:NZ23}.
However, until now, no plausible construction in the plain model was known. Even the task of finding quantum states that are efficiently clonable without an oracle\textemdash but not trivial enough to be described classically\textemdash has been a challenge. 
We give the first candidate construction of quantum fire 
in the plain model. We challenge the cryptographic community to find either a proof of its security from reasonable assumptions or to break it. We further challenge the community to find and propose other reasonable candidate constructions for quantum fire. Much like the 15-year challenge of finding reasonable constructions of public-key quantum money has led to a variety of new techniques for proving unclonability, we expect the task of finding candidate quantum fire constructions to prove to be a challenging task and to require new and specialized techniques for showing untelegraphability.

\subsection{Definition}
Quantum fire was implicit in the oracle construction of~\cite{ITCS:NZ23}, but no formal definition was given.
We give a definition of public-key quantum fire here. 

\begin{definition}[Public-key Quantum Fire]
    \label{def:qfire}
    A \emph{public-key quantum fire} scheme consists of four quantum algorithms $\mathcal{S} = (\keygen, \spark,\clone,\ver)$ where 
    \begin{itemize}
        \item $\keygen(1^{n})$ takes as input the security parameter $1^n$ and outputs a private/public key pair $(\sk,\pk)$, 
        \item $\spark(\pk)$ takes the public key and outputs a serial number $s$ and a quantum fire state $\ket{\phi^{s}}$, which we refer to as a \emph{flame}, 
        \item $\clone(\pk, s,\ket{\phi^{s}})$ takes as input the public key $\pk$, a serial number $s$, and a flame $\ket{\phi^{s}}$, and outputs two registers $\reg{AB}$ in some potentially entangled state $\sigma_{\reg{AB}}^s$,
        \item $\ver(\pk, s,\sigma)$ takes as input the public key $\pk$, a serial number $s$, and an alleged flame $\sigma$, and either accepts or rejects.%
        \footnote{
            $\ver$ may not exist for unverifiable quantum fire, or it may require the secret key $\sk$ for secretly verifiable quantum~fire.
        }
    \end{itemize}
    A quantum fire scheme $\mathcal{S}$ satisfies correctness if for all $\lambda$, sparking is correct
    \begin{equation*}
        \Pr 
        \left[ 
            \ver(\pk, s, \ket{\phi^{s}}) 
            \text{ accepts}
            \;\;\middle|
            \begin{array}{c}
            (\sk,\pk) \leftarrow  \mathsf{KeyGen}(1^\lambda) \\
            (s,\ket{\phi^{s}}) \leftarrow \spark(\pk) 
            \end{array}
       \right] 
       \geq 
       1 - \negl(n)
       \,,
    \end{equation*}
    and cloning is correct%
    \footnote{
        Technically, we need that cloning is correct even for the outputs of the cloner itself. That is, we should be able to apply the cloner repeatedly on its own outputs any polynomial number of times. This rules out trivial cases such as where the sparking mints two quantum money states, and verification accepts either one or both of them, but they are otherwise unclonable. This would allow ``cloning'' once by dividing up the two money states, but those two ``copies'' cannot be cloned further. Quantum fire definitionally requires its flame states to be unboundedly clonable.
    }
    \begin{equation*}
        \Pr 
        \left[ 
            \ver(\pk, s, \cdot) 
            \text{ accepts both registers of }
            \sigma_{\reg{AB}}^s
            \;\;\middle|
            \begin{array}{c}
            (\sk,\pk) \leftarrow  \mathsf{KeyGen}(1^\lambda) 
            \\
            (s,\ket{\phi^{s}}) \leftarrow \spark(\pk) 
            \\
            \sigma_{\reg{AB}}^s \leftarrow \clone(\pk, s,\ket{\phi^{s}}) 
            \end{array}
       \right] 
       \geq 
       1 - \negl(n)
       \,.
    \end{equation*}
\end{definition}

Untelegraphability~\cite{ITCS:NZ23} of quantum fire means that it is hard to encode a flame as a classical encoding which can later be brought back. That is, once the flame is extinguished, it is gone. We model this as a pair of adversaries. The first is tasked with deconstructing the flame into a classical message, and the second must use the deconsructed classical message to reconstruct the state.%
\footnote{
    Note that the two adversaries should \emph{not} be entangled, as this allows them to \emph{teleport} the state. Furthermore, maintaining any entanglement implies having to store a quantum register, which is what telegraphing aims to avoid.
}
   
\begin{longfbox}[breakable=false, padding=1em, margin-top=1em, margin-bottom=1em]
\begin{algo}[{\bf Quantum Fire Telegraphing Security Game}]\label{prot:quantum_fire_security}
\end{algo}
\begin{enumerate}
    \item Challenger generates $(\sk,\pk) \leftarrow \keygen(1^{\lambda})$, $(s, \ket{\phi^s})\leftarrow \spark(\sk)$ and send $(\pk, s, \ket{\phi^s})$ to adversary $A$.
    \item Adversary $A$ returns a classical encoding of the flame $c \in \bits^*$.
    \item Challenger passes $c \in \bits^*$ to adversary $B$.
    \item Adversary $B$ returns a claimed quantum state $\sigma$.
    \item Challenger runs $\ver(\pk, s, \sigma)$ and outputs the result.%
    \footnote{
        In the case of unverifiable quantum fire, the challenger verifies the telegraphing by measuring in a basis containing the valid flame $\ket{\phi^s}$.
    }
\end{enumerate}
\end{longfbox}

\begin{definition}[Quantum fire security]
    A quantum fire scheme $\mathcal{S}$ satisfies \emph{$\epsilon$-quantum-fire security} if for all pairs of efficient adversaries $A$ and $B$, the success probability of $A$ in the Telegraphing Security Game (\Cref{prot:quantum_fire_security}) is at most $\epsilon(n)$.
\end{definition}

We require that $\spark$ and $\clone$ are efficient (QPT) algorithms. $\ver$ may be an efficient public verification, an efficient private verification, or an inefficient verification, leading to three different kinds of quantum fire (publicly verifiable quantum fire, privately verifiable quantum fire, and statistically verifiable quantum fire).
Furthermore, the weaker notion of security is against standard non-interactive telegraphing. A stronger security notion, and one that is more amenable to cryptographic applications, additionally rules out interactive telegraphing, where $A$ and $B$ may exchange any number of classical messages in \Cref{prot:quantum_fire_security}.

\subsection{Construction}
Let $\G$ and $\bbH$ be two groups, and let $f: \G \to \bbH$ be an injective homomorphism between them, which we assume to be one-way. 
Let $\bbH_{f}$ be the image of $f$.
Let $\rep: \G \to U(\mathcal{H})$ be the representation of $\G$ which acts as $\rep(g)\ket{h} = \ket{f(g) \cdot h}$.
Let the set of quantum fire labels (or serial numbers) be $\widehat \G$, the set of irreps of $\G$, and for each $\lambda \in \widehat\G$, we let valid flames be any state in isotypic component of irrep $\lambda$, that is, any state of the form

\begin{align*}
    \ket{\phi^{\lambda}_{i,j}} = \sum_{g \in \G} \varrho_{\lambda}(g^{-1})_{i,j} \ket{f(g)}
\end{align*}

We further assume that there is an efficient algorithm to prepare a uniform superposition over the image group $\bbH_{f}$: $\ket{\Phi} = \sum_{h \in \bbH_{f}} \ket{h}$, or a quantum state that approximates it.%
\footnote{Supposedly, this image group is known to all parties, while the specific mapping between $\G$ and $\bbH_{f}$ could be arbitrary.}

\paragraph{Verification and Sparking.}
Verification is the same as verification for the quantum money construction of \Cref{sec:qm-construction}: we perform a course Fourier measurement to produce the irrep label $\lambda$ and compare with the claimed quantum fire label. Likewise, to spark a quantum fire state---that is, to prepare a fire state with a random label---run the same verification process on the identity element of $H$, which samples the irrep label $\lambda$ according to the Plancherel measure of $G$.

\paragraph{Cloning.}
We are given a quantum fire state $\ket{\phi^{\lambda}_{i, j}}$ with label $\lambda$, and we would like to output two such fire states, both of which pass verification for the same label $\lambda$.

We first prepare $\ket{\Phi} = \sum_{h \in \bbH_{f}} \ket{h}$, a uniform superposition over the image group $\bbH_{f}$. Together with the fire state, we now have the overall state
\begin{align*}
    \ket{\phi^{\lambda}_{i, j}} \otimes \ket{\Phi}
    &=
    \sum_{\substack{g \in \G \\ h \in \bbH_{f}}} \varrho_{\lambda}(g^{-1})_{i, j} \ket{f(g)} \ket{h}\,.
\end{align*}
\noindent Since $f$ is injective, and therefore bijective between $\G$ and $\bbH_{f}$, we can reindex the sum over $h$ as 
\begin{align*}
    &=
    \sum_{g, h \in \G} \varrho_{\lambda}(g^{-1})_{i, j} \ket{f(g)} \ket{f(h)}
\end{align*}

Both registers contain an element of $H$. We apply the inverse group operation of the second register into the first register on the left to get
\begin{align*}
    &\rightarrow
    \sum_{g, h \in \G} \varrho_{\lambda}(g^{-1})_{i, j} \ket{f(h)^{-1} \cdot f(g)} \ket{f(h)}
    \\
    &=
    \sum_{g, h \in \G} \varrho_{\lambda}(g^{-1})_{i, j} \ket{f(h^{-1} \cdot g)} \ket{f(h)}
    \\
    &=
    \sum_{g, h \in \G} \varrho_{\lambda}(g^{-1} \cdot h^{-1})_{i,j} \ket{f(g)} \ket{f(h)}
    \\
    &=
    \sum_{\substack{g, h \in \G \\ k \in [\dimlambda}} \varrho_{\lambda}(g^{-1})_{ik} \varrho_{\lambda}(h^{-1})_{k, j} \ket{f(g)} \ket{f(h)}
    \\
    &=
    \sum_{k \in [\dimlambda]} \sum_{g \in \G} \varrho_{\lambda}(g^{-1})_{i, k} \ket{f(g)} \sum_{h \in \G} \varrho_{\lambda}(h^{-1})_{k, j} \ket{f(h)}
    \\
    &=
    \sum_{k \in [\dimlambda]} \ket{\phi^{\lambda}_{i, k}} \otimes \ket{\phi^{\lambda}_{k, j}}
\end{align*}

This produces two quantum fire states that both pass verification for the same original label $\lambda$. While not necessary, if we wish, we could even force the two new fire states to have the same $i$ and $j$ values as the original, and in doing so disentangle them. We simply perform a \fse on both states from both the left and right side to extract out the new $i$ and $j$ values, replace them with the old $i$ and $j$, and uncompute%
\footnote{See \Cref{sec:cloning_exact_from_extraction} for more details on how to do this.}
to get the tensor product:
\begin{align*}
    \ket{\phi^{\lambda}_{i, j}} \otimes \ket{\phi^{\lambda}_{i, j}}
\end{align*}

\paragraph{Untelegraphability.}
We have shown above that these states are efficiently clonable. In order for the construction to be a secure quantum fire scheme, the states must also be \emph{untelegraphable}. That is, there must be no way to deconstruct the states into a classical message that can later be reconstructed back into the quantum state, or at least one that properly passes verification. We leave as an open problem to find sufficient conditions on the one-way homomorphism that would allow showing untelegraphability in the plain model. 

\begin{remark}
The untelegraphability of such a scheme is known to be difficult to prove even relative to a classical oracle: Nehoran and Zhandry~\cite{ITCS:NZ23} show the security of their implicit quantum fire scheme relative to a unitary quantum oracle. Unfortunately, they also show that the same quantum fire construction leads to a unitary oracle separation between the complexity classes $\mathsf{clonableQMA}$ and $\QCMA$, and therefore between $\QMA$ and $\QCMA$. As generalization of this, they observe that any provably secure and public-key quantum fire scheme relative to a \emph{classical} oracle will likely lead to a \emph{classical} oracle separation between $\QMA$ and $\QCMA$, an major longstanding open problem of Aharonov and Naveh~\cite{aharonov2002quantum} that remains unresolved despite recent progress.
\end{remark}

\begin{observation}
We observe that while one-wayness may not be a sufficient condition for untelegraphability, is a \emph{necessary} condition. This is because if we can invert the homomorphism---and we can also perform a quantum Fourier transform on the group---then we can telegraph the state as the classical description of $\lambda$, $i$, and $j$.
\end{observation}

\begin{proof}[Proof sketch.]
Suppose, for instance, that we can invert $f$ perfectly. Then we can do the following. 
Alice starts with a quantum fire state $\ket{\phi^{\lambda}_{i, j}} = \sum_{g \in \G} \varrho_{\lambda}(g^{-1})_{i, j} \ket{f(g)}$ and inverts $f$ to get
\begin{align*}
    \sum_{g \in \G} \varrho_{\lambda}(g^{-1})_{i, j} \ket{f(g)} \ket{g}
    \,.
\end{align*}
She now uncomputes $f(g)$ in the first register to get
\begin{align*}
    \sum_{g \in \G} \varrho_{\lambda}(g^{-1})_{i, j} \ket{g}
    \,,
\end{align*}
which is the left-regular Fourier basis state $\ket{\mathcal{L}^{\lambda}_{i,j}}$. Taking the quantum Fourier transform of this state then yields $\ket{\lambda} \ket{i} \ket{j}$, which is a classical string that Alice can send to Bob.
Bob can then invert this process to recover $\ket{\phi^{\lambda}_{i, j}}$.
\end{proof}

The notion of quantum fire was featured implicitly in the work of Nehoran and Zhandry~\cite{ITCS:NZ23}, where they show that such an object exists relative to a unitary quantum oracle. Their construction uses two oracles: a (quantumly accessible) random oracle, which serves effectively as a verification oracle, and a unitary oracle, which is used for cloning the resulting states.
Unfortunately, the scheme of \cite{ITCS:NZ23} offers little hope of leading to a plain-model instantiation. This is because, as they note, the unitary implemented by the unitary cloning oracle is one that cannot be implemented efficiently.

One approach to strengthen their result is to give a construction from classical functionality.
A priori, however, it is not even clear that \emph{any} classical functionality can bestow clonability on a state that cannot be encoded classically.
To the best of our knowledge, every known method of efficiently cloning quantum states first passes through the classical description of the states, copies this classical description, and then recovers two clones of the quantum state from the classical descriptions.
However, this automatically means that such states are efficiently telegraphable---they can be stored as their classical descriptions.
\emph{How can we clone a quantum state (using efficient classical functionality) without ever going through a classical description?}

We answer this question here by giving a proof of concept that this kind of cloning is in fact possible, along with a framework for using it to construct quantum fire with conjectured security. 
An interesting aspect of our cloning procedure is that the quantum states of the two registers inherently become entangled during the course of the procedure, and only become disentangled at the end. Furthermore, it requires applying a controlled group operation between the two registers. These aspects together give intuitive justification for the untelegraphability of this cloning procedure: controlled operations and more general entangling procedures \emph{cannot} occur over a classical channel.

\ifanon
\else

\ifstocproceedings
\begin{acks}

\else
\section*{Acknowledgments}
\fi

The authors thank Henry Yuen, Chinmay Nirkhe, Adam Bouland, and Tudor Giurgica-Tiron for insightful discussions that contributed to the direction of the paper, and we thank Takashi Yamakawa, Tomoyuki Morimae, and Dakshita Khurana for enlightening discussions on the significance of our results and further cryptographic applications of our duality theorem.
J.B. is supported by Henry Yuen's AFORS (award FA9550-21-1-036) and NSF CAREER (award CCF2144219).
This work was done in part while J.B. and B.N. were visiting the Simons Institute for the Theory of Computing, supported by NSF QLCI Grant No. 2016245.

\ifstocproceedings
\end{acks}
\fi

\fi

\bibliographystyle{alpha}
\bibliography{references}

\end{document}